\documentclass[runningheads]{llncs}
\usepackage[T1]{fontenc}
\usepackage{soul}
\usepackage{amsmath}
\usepackage{amssymb}
\usepackage[english]{babel}
\usepackage[utf8]{inputenc}
\usepackage{float}
\usepackage{tabularx}
\usepackage{hyperref}
\usepackage{nameref,cleveref}
\usepackage{multicol}
\usepackage{pifont}

\usepackage[boxed]{algorithm}
\usepackage[noend]{algpseudocode}

\usepackage[inline, shortlabels]{enumitem}

\usepackage{aligned-overset}

\allowdisplaybreaks
\usepackage{multirow}
\usepackage[disable]{todonotes}
\usepackage{makecell}
\usepackage{booktabs}
\usepackage{subcaption}
\usepackage{thmtools,thm-restate}
\usepackage{tikz}
\usepackage{stmaryrd}
\usepackage{subcaption}
\usepackage{framed}
\usepackage{mathtools}
\usepackage{aligned-overset}
\usepackage{apptools}
\AtAppendix{\counterwithin{lemma}{section}}

\usepackage{aligned-overset}

\crefname{figure}{Fig.}{Figs.}

\usepackage{tikz}
\usetikzlibrary{positioning,arrows,automata,calc,shapes}

\definecolor{darkgreen}{rgb}{0.0, 0.5, 0.0}
\definecolor{darkred}{rgb}{0.9, 0.0, 0.0}
\definecolor{darkblue}{rgb}{0.0, 0.0, 0.9}

\tikzset{
	state/.style={rounded rectangle,draw=black,inner sep=1.5mm,minimum width=9mm,minimum height=5mm},
	rstate/.style={rectangle,draw=black,inner sep=1ex,minimum width=9mm,minimum height=5mm},
	istate/.style={minimum width=5mm, minimum height=6mm},
	bullet/.style={circle,draw=black,fill=black,inner sep=0cm, minimum size=0.7mm},
	initial text={},
	every initial by arrow/.style={->,>=stealth'},
	ptran/.style={rounded corners, ->,>=stealth',auto},
	ntran/.style={rounded corners, -,auto}
}

  \tikzset{
    invisible/.style={opacity=0, fill opacity=0},
    visible on/.style={alt={#1{}{invisible}}},
    alt/.code args={<#1>#2#3}{%
      \alt<#1>{\pgfkeysalso{#2}}{\pgfkeysalso{#3}} %
    },
  }

\usetikzlibrary{shapes.geometric, arrows}
\usetikzlibrary{decorations.pathreplacing}
\usetikzlibrary{positioning}

\usepackage[skins]{tcolorbox}

\definecolor{milpColor}{HTML}{778da9}
\definecolor{mdpColor}{HTML}{a6a6a8}
\definecolor{queryColor}{HTML}{cecece}
\definecolor{certColor}{HTML}{c6dabf}
\definecolor{fixedPointColor}{HTML}{4c956c}
\definecolor{cmarkColor}{HTML}{38b000}
\definecolor{xmarkColor}{HTML}{c1121f}

\newtcolorbox{myframe}[2][]{%
center,
  enhanced,colback=white,colframe=black,coltitle=black,
  rounded corners,boxrule=0.4pt,
  attach boxed title to top left={yshift=-0.5\baselineskip-0.4pt,xshift=2mm},
  boxed title style={tile,size=minimal,left=0.5mm,right=0.5mm,
    colback=white,before upper=\strut},
  title=#2,#1,
  width=1\linewidth,
  boxsep=1pt,left=2pt,right=2pt,top=2pt,bottom=2pt
}

\tikzstyle{milp} = [rectangle,
minimum width=3.3cm, 
minimum height=1.5cm,
text centered, 
text width=3.3cm, 
draw=black, 
fill=milpColor!50]

\tikzstyle{result} = [rectangle,
minimum width=3cm, 
minimum height=0.6cm,
text centered, 
text width=3cm, 
draw=black, 
fill=queryColor,
rounded corners
]

\tikzstyle{certificate} = [rectangle,
minimum width=3cm, 
minimum height=0.6cm,
text centered, 
text width=3cm, 
draw=black, 
fill=certColor,
rounded corners
]

\tikzstyle{fixedPoint} = [rectangle,
minimum width=3cm, 
minimum height=0.7cm,
text centered, 
text width=3cm, 
draw=black, 
fill=fixedPointColor,
rounded corners
]

\tikzstyle{mdp} = [rectangle, 
minimum width=3.95cm, 
minimum height=0.9cm, 
text centered, 
text width=3.95cm, 
draw=black, 
fill=mdpColor]

\tikzstyle{query} = [rectangle, 
minimum width=3.95cm, 
minimum height=0.6cm, 
text centered, 
text width=3.95cm, 
draw=black, 
fill=queryColor]

\tikzstyle{mdpShort} = [rectangle, 
minimum width=3cm, 
minimum height=0.9cm, 
text centered, 
text width=3cm, 
draw=black, 
fill=mdpColor]

\tikzstyle{queryShort} = [rectangle, 
minimum width=3cm, 
minimum height=0.6cm, 
text centered, 
text width=3cm, 
draw=black, 
fill=queryColor]

\tikzstyle{decision} = [diamond, 
minimum width=3cm, 
minimum height=1cm, 
text centered, 
draw=black, 
fill=green!30]

\tikzstyle{arrow} = [thick,->,>=stealth]
\tikzstyle{resultArrow} = [thick,->]

\newcommand{\union}{\cup}
\newcommand{\Union}{\bigcup}
\newcommand{\intersection}{\cap}
\newcommand{\prob}{\mathsf{Pr}}
\newcommand{\vect}[1]{\vec{#1}}
\newcommand{\reals}{\mathbb{R}}
\newcommand{\realsnn}{\reals_{\geq 0}}
\DeclareMathOperator{\suppOp}{supp}
\newcommand{\supp}[2][]{\suppOp_{#1}(#2)}
\DeclareMathOperator{\stateSuppOp}{state-supp}

\newcommand{\stateSupp}[2][]{\stateSuppOp_{#1}(#2)}
\newcommand{\naturals}{\mathbb{N}}

\newcommand{\card}[1]{{\lvert {#1} \rvert}}
\DeclareMathOperator{\distrOp}{Distr}
\newcommand{\distr}[1]{\distrOp(#1)}

\newcommand{\tuple}[1]{\langle #1 \rangle}

\DeclareMathOperator*{\argmax}{arg\,max}

\DeclareMathOperator*{\opt}{opt}

\newcommand{\boldLambda}{\pmb{\lambda}}
\newcommand{\boldMu}{\pmb{\mu}}
\newcommand{\boldDelta}{\pmb{\delta}}

\newcommand{\boldBeta}{\pmb{\beta}}

\DeclareMathOperator{\downOp}{down}
\newcommand{\down}[1]{\downOp(#1)}

\newcommand{\Conj}{\bigwedge}
\newcommand{\Disj}{\bigvee}

\newcommand{\eventually}{\lozenge}
\newcommand{\globally}{\square}

\DeclareMathOperator{\UntilOp}{\mathsf{U}}
\DeclareMathOperator{\NextOp}{\mathsf{X}}

\newcommand{\targetSet}{G}
\newcommand{\mdp}{\mathcal{M}}
\newcommand{\dtmc}{\mathcal{C}}
\newcommand{\states}{S}
\newcommand{\state}{s}
\newcommand{\targets}{F}

\newcommand{\actions}{Act}
\newcommand{\action}{a}
\newcommand{\init}{\bar{\state}}
\newcommand{\transMat}{\vect{P}}
\newcommand{\exit}{\bot}

\newcommand{\SM}{\vect{A}}
\newcommand{\SA}{\mathsf{En}}
\newcommand{\TM}{\vect{T}}
\newcommand{\scheduler}{\mathfrak{G}}
\newcommand{\schedulers}{\Sigma}

\DeclareMathOperator{\pathsOp}{Paths}
\DeclareMathOperator{\lastOp}{last}
\newcommand{\last}[1]{\lastOp({#1})}
\newcommand{\paths}{\pathsOp}
\newcommand{\pathsFin}{\pathsOp_{\mathsf{fin}}}
\newcommand{\Path}{\pi}
\newcommand{\initDistr}{\pmb{\delta}_{in}}

\newcommand{\mec}{\mathcal{E}}
\newcommand{\ec}{\mathcal{E}}

\newcommand{\MECS}{\mathsf{MEC}}
\newcommand{\ECS}{\mathsf{EC}}

\newcommand{\mecQuotient}{\widehat{\mdp}}

\newcommand{\labels}{L}

\newcommand{\labeling}{\Lambda}
\newcommand{\mecClass}[2][]{{[\![ #2 ]\!]}_{#1}}

\newcommand{\universalDQ}{\ensuremath{(\forall, \lor)}}

\newcommand{\existsCQ}{\ensuremath{(\exists, \land)}}

\newcommand{\pairs}{P}
\newcommand{\modelsRabin}{\models_{\mathsf{Rabin}}}

\newcommand{\queryProb}{\mathbb{P}}

\newcommand{\query}{\Psi}
\newcommand{\existsQuery}{\query}
\newcommand{\forallQuery}{\Phi}
\newcommand{\prop}{\varphi}
\newcommand{\valueVector}{\vect{v}}
\newcommand{\partition}{\mathcal{D}}

\newcommand{\lang}[1]{\mathcal{L}(#1)}
\newcommand{\automaton}{\mathcal{A}}
\newcommand{\acceptanceCondition}{\alpha}
\newcommand{\alphabet}{\Sigma}
\newcommand{\autoStates}{Q}
\newcommand{\autoState}{q}

\newcommand{\autoInit}{\bar{\autoState}}
\DeclareMathOperator{\infiniteOp}{\mathsf{inf}}

\newcommand{\infinite}[1]{\infiniteOp(#1)}

\newcommand{\run}{r}
\newcommand{\word}{w}
\newcommand{\letter}{\sigma}
\newcommand{\autoTransition}{\delta}

\usepackage{graphicx}

\usepackage{color}

\begin{document}
\title{Certificates and Witnesses for Multi-objective $\omega$-regular Queries in Markov Decision Processes
\thanks{The authors were supported by the German Federal Ministry of Education and Research (BMBF) within the project SEMECO Q1 (03ZU1210AG) and by the German Research Foundation (DFG) through the Cluster of Excellence EXC 2050/1 (CeTI, project ID 390696704, as part of Germany’s Excellence Strategy) and the DFG Grant 389792660 as part of TRR 248 (Foundations of Perspicuous Software System).}}
\titlerunning{Certificates and Witnesses for Multi-Objective $\omega$-regular Queries} 
\author{Christel Baier\inst{1} \and Calvin Chau\inst{1} \and Volodymyr Drobitko\inst{1} \and Simon Jantsch\inst{2} \and Sascha Klüppelholz\inst{1}}
\authorrunning{Baier et al.}
\institute{Technische Universität Dresden, Dresden, Germany \and Siemens EDA, Munich, Germany}
\maketitle              %
\begin{abstract}
Multi-objective probabilistic model checking is a powerful technique for verifying stochastic systems against multiple (potentially conflicting) properties. To enhance the trustworthiness and explainability of model checking tools, we present independently checkable \emph{certificates} and \emph{witnesses} for multi-objective $\omega$-regular queries in Markov decision processes. For the certification, we extend and improve existing certificates for the decomposition of maximal end components and reachability properties. We then derive mixed-integer linear programs (MILPs) for finding \emph{minimal witnessing subsystems}. For the special case of Markov chains and LTL properties, we use \emph{unambiguous Büchi automata} to find witnesses, resulting in an algorithm that requires single-exponential space. Existing approaches based on deterministic automata require doubly-exponential space in the worst case. Finally, we consider the practical computation of our certificates and witnesses and provide an implementation of the developed techniques, along with an experimental evaluation, demonstrating the efficacy of our techniques.

\end{abstract}

\renewcommand{\restriction}{\mathord{\upharpoonright}}

\section{Introduction}
Probabilistic model checking (PMC) is a well-established approach for analysing and formally verifying stochastic systems, e.g. randomized algorithms or security systems \cite{kwiatkowsa_prism_2012,hartmanns_quantitative_2019}.
Model checkers, like \textsc{Prism} \cite{kwiatkowska_prism_2011-1} or \textsc{Storm} \cite{hensel_probabilistic_2022}, receive a probabilistic model of the system and formal specification and output whether the specification is satisfied.
Unsurprisingly, as with any complex software, model checkers have been observed to contain bugs \cite{wimmer_demand_2008}.
In the context of verifying (safety-critical) systems, this raises the question of how to ensure that model checking results can be trusted.
An approach for mitigating this issue are \emph{certifying algorithms}~\cite{mcconnell_certifying_2011} which also return a \emph{certificate} accompanying the result. This can be used to independently verify the correctness of the result.

Another challenge for PMC is the \emph{explication} of verification results.
When a system fails to meet the specification, it is desirable to obtain diagnostic information that allows the user to understand the behavior and debug the system. Indeed, counterexample traces are a crucial feature of classical model checking techniques in practice.
A canonical notion of explication does not exist for PMC, e.g. sets of paths~\cite{han_counterexample_2009-1,aljazzar_generation_2009-1} or subsystems~\cite{jansen_hierarchical_2011,wimmer_minimal_2012,jansen_counterexamples_2015,funke_farkas_2020} have been considered. Witnessing subsystems are parts of the system that on their own violate (or satisfy) the specification and can therefore serve as useful explications.

This work addresses the \emph{certification} and \emph{computation of witnesses} for \emph{multi-objective $\omega$-regular queries} in {Markov decision processes} (MDPs)~\cite{puterman_markov_1994}. MDPs are the standard model for systems with stochastic and nondeterministic behavior. Further, multi-objective $\omega$-regular queries are expressive formalisms, enabling the specification of multiple (possibly conflicting) requirements \cite{etessami_multi-objective_2008}. For instance, for a server system one might specify that ``every request is served with probability greater than $0.9$'' while ``the probability of failure is less than $0.01$'' \cite{etessami_multi-objective_2008,forejt_quantitative_2011}.

To provide certificates, we revisit the steps of the verification procedure for these queries (e.g. \cite{forejt_quantitative_2011}) and make them certifying. More precisely, we devise certificates for the \emph{classification} of \emph{end components} based on the properties they satisfy. Further, we improve on the certification of the \emph{decomposition into maximal end components} (MECs) \cite{de_alfaro_formal_1997} from \cite{jantsch_certificates_2022}, using recent ideas from \cite{chatterjee_fixed_2025}, and provide a first implementation thereof. The MEC decomposition is a common preprocessing step in PMC, and its certification is thus interesting on its own. In combination with previous work on multi-objective reachability queries \cite{baier_certificates_2024}, this allows us to devise \emph{sound and complete} certificates for $\omega$-regular queries. Our certificates can then be used to \emph{independently validate} that a query is indeed satisfied (or violated), thereby enhancing trust in the verification process.

As in the setting of reachability queries \cite{funke_farkas_2020,jantsch_certificates_2022,baier_certificates_2024}, we show that our certificates \emph{induce witnessing subsystems}. Since small witnesses highlight critical parts of the MDP and might be easier to understand, we present algorithms based on mixed-integer linear programs (MILPs) for computing \emph{minimal witnessing subsystems}.

For the special case of discrete-time Markov chains (DTMCs) and LTL properties, we devise a novel algorithm for computing minimal witnesses based on unambiguous B\"uchi automata (UBAs).
Checking DTMCs against probabilistic LTL specifications is \textsf{PSPACE}-complete \cite{baier_markov_2023-1}, implying that minimal witnesses can be computed in \textsf{PSPACE} by enumerating and checking all subsystems.
As this is impractical, previous approaches relied on computing deterministic Rabin automata (DRAs) \cite{jansen_counterexamples_2015}, requiring \emph{doubly-exponential} space.\todo{SJ. I believe this may even have been triple-exponential, because they linearly many binary variables in the size of the DRA. But I'd have to check.}
Instead, we define a \emph{single-exponential} MILP based on the product of the DTMC with a UBA for the LTL property, building on results for model checking DTMCs against UBAs~\cite{baier_markov_2023-1}, which is possible in \textsf{PTIME}.

Proofs are in the appendix and the implementation and data are given in \cite{baier_2025_15680332}.

\medskip

\noindent\textbf{Contributions.}
\begin{itemize}
\item We devise certificates for the satisfaction (resp. violation) of Rabin properties in end components and improve the certification of the MEC decomposition (\Cref{subsection:satisfying-ecs}) from \cite{jantsch_certificates_2022} based on ideas from \cite{chatterjee_fixed_2025}.
\item In \Cref{subsection:certificate-rabin-streett} we present certificates for multi-objective queries with $\omega$-regular properties in MDPs which are composed of the aforementioned certificates and the certificates for multi-objective reachability from \cite{baier_certificates_2024}. We also discuss the computation of the latter via the algorithm of \cite{forejt_pareto_2012}.
\item \Cref{section:witnesses-mdps} presents techniques for computing minimal witnessing subsystems for multi-objective $\omega$-regular queries, supporting various minimality criteria.
\item For LTL queries in DTMCs, we give a single-exponential algorithm for computing minimal witnessing subsystems, utilizing unambiguous B\"uchi automata (\Cref{section:dtmc}).
\item All presented techniques are implemented in an efficient C\texttt{++} tool built on top of \textsc{Storm}, for which we provide an experimental evaluation (\Cref{section:experimental}).
\end{itemize}

\noindent\textbf{Related work.}
\emph{Multi-objective queries} in MDPs have been extensively studied, e.g. in \cite{etessami_multi-objective_2008,forejt_quantitative_2011,forejt_pareto_2012,randour_percentile_2015,quatmann_verification_2023}. Techniques for verifying multi-objective $\omega$-regular queries based on LPs are presented in \cite{etessami_multi-objective_2008,forejt_quantitative_2011}, where \cite{forejt_quantitative_2011} considers queries with rewards. The works of \cite{forejt_pareto_2012,quatmann_verification_2023} consider the verification of multi-objective queries via value iteration. Queries with mean-payoff objectives have been studied in \cite{10.1109/LICS.2011.10}. In \cite{randour_percentile_2015} complexity results for percentile queries are established. Certifying verification algorithms have not been considered in the mentioned works.

\emph{Farkas certificates} \cite{funke_farkas_2020,jantsch_certificates_2022} certify \emph{single} constraints on \emph{reachability probabilities} in MDPs. They are solutions to linear inequalities, derived from known LP characterizations \cite{puterman_markov_1994}, and can also be used to obtain \emph{minimal witnessing subsystems}. Certificates for the MEC decomposition are presented in \cite{jantsch_certificates_2022}, however without implementation.
The work of \cite{baier_certificates_2024} generalizes Farkas certificates to the setting of \emph{multi-objective reachability}. 
\emph{Fixed-point certificates} \cite{chatterjee_fixed_2025} can be seen as certified bounds on the optimal reachability probabilities and expected rewards in MDPs, based on fixed-points of the \emph{Bellman operator} for reachability probabilities and expected rewards (see \cite{puterman_markov_1994,baier_principles_2008-1}) and newly developed certificates for qualitative reachability. An implementation for computing the certificates and a verified checker for validating them is also provided. Similarly, \cite{Schäffeler_Abdulaziz_2023} provides a verified implementation of solution methods for expected rewards in MDPs.

\emph{Martingale-based certificates}, e.g. \cite{chakarov_probabilistic_2013-1,takisaka_ranking_2021,henzinger_supermartingale_2025-1,abate_quantitative_2025-1}, are important for the control and verification of stochastic systems with infinite state space, e.g. arising from probabilistic programs. The certificates from \cite{abate_quantitative_2025-1} combine Streett supermartingales with stochastic invariants and are sound and complete (i.a. for finite systems). However, multiple objectives are not supported and their implementation cannot be easily reused. The certificates from \cite{henzinger_supermartingale_2025-1} are sound, but incomplete, and multiple objectives are not considered.

\emph{Explications} for PMC have been considered in the form of sets of paths~\cite{han_counterexample_2009-1,aljazzar_generation_2009-1}, fault-trees \cite{kuntz_probabilistic_2011} and witnessing subsystems~\cite{jansen_hierarchical_2011,wimmer_minimal_2012,jansen_counterexamples_2015,funke_farkas_2020}.

\section{Preliminaries}
\label{section:prelims}
\noindent\textbf{Notation.}
We define $[k] \coloneqq \{1, \dots, k\}$. Vectors and matrices are written in boldface, e.g. $\vect{x}$ and $\SM$. Let $\states = \{\state_0, \dots, \state_n \}$ be a finite set. We write $\vect{x} \in \reals^\states$ instead of $\vect{x} \in \reals^\card{\states}$ and $\vect{x}(\state_i)$ to denote its $i$th entry. Let $\vect{1}_{\state} \in \{0, 1\}^{\states}$ be the vector that is \emph{Dirac} in $\state$. %
The \emph{support} of $\vect{x} \in \realsnn^\states$ is defined by $\supp{\vect{x}} = \{\state \in \states \mid \vect{x}(\state) > 0\}$. The set of probability distributions over $S$ is denoted by $\distr{\states}$.

\medskip

\noindent\textbf{$\omega$-regular property.}
An \emph{automaton} $\automaton = (\autoStates, \alphabet, \autoInit, \autoTransition,\allowbreak \acceptanceCondition)$ is a tuple where $\autoStates$ is a finite set of \emph{states}, $\alphabet$ a finite \emph{alphabet}, $\autoInit$ the \emph{initial state}, $\autoTransition \subseteq \autoStates \times \alphabet \times \autoStates$ a \emph{transition relation} and $\acceptanceCondition \subseteq 2^\autoStates \times 2^\autoStates$ an \emph{acceptance condition}.
A \emph{run} on an \emph{$\omega$-word} $\word = \word_0 \word_1 \ldots \in \alphabet^\omega$ in $\automaton$ is an infinite sequence $\run = \autoState_0 \autoState_1 \ldots \in \autoStates^\omega$ such that $\autoState_0 = \autoInit$ and $(\autoState_i, \word_i, \autoState_{i+1}) \in \autoTransition$ for all $i \geq 0$. Let $\infinite{\run}$ be the states appearing infinitely often in $\run$. A \emph{Rabin} automaton accepts a run $\run$ iff $\infinite{\run} \intersection F \neq \emptyset$ and $\infinite{\run} \intersection \autoStates = \emptyset$ for some $(F, E) \in \acceptanceCondition$. A \emph{Streett} automaton accepts a run $\run$ iff $\infinite{\run} \intersection F = \emptyset$ or $\infinite{\run} \intersection \autoStates \neq \emptyset$ for all $(F, E) \in \acceptanceCondition$.
An $\omega$-word $\word$ is \emph{accepted} by $\automaton$ iff there is an accepting run on $\word$. The \emph{language} $\lang{\automaton}$ recognized by $\automaton$ consists of all accepted words. 
The \emph{complement language} $\alphabet^\omega \setminus \lang{\automaton}$ of a Rabin automaton $\automaton$ is recognized by the structurally equivalent Streett automaton, and vice versa.
A \emph{property} $\prop \subseteq \alphabet^\omega$ is \emph{$\omega$-regular} if it is recognized by a Rabin or Streett automaton. Its complement is denoted by $\bar{\prop}$.
An automaton $\automaton$ is \emph{deterministic} if for every $\autoState \in \autoStates$ and $\letter \in \alphabet$ there is exactly one $\autoState' \in \autoStates$ with $(\autoState, \letter, \autoState') \in \autoTransition$ and \emph{unambiguous} if for every $\omega$-word there is at most one accepting run. A \emph{Büchi} automaton is a Rabin automaton with $\acceptanceCondition = \{(F, \emptyset) \}$.

\medskip

\noindent\textbf{Markov decision process.} An \emph{MDP} $\mdp$ is a tuple $(\states, \actions, \init, \transMat)$ where $\states$ is a finite set of \emph{states}, $\actions$ a finite set of \emph{actions}, $\init \in \states$ the initial state and $\transMat \colon \states \times \actions \to \distr{\states}$ a \emph{partial transition function}.
We often write $\transMat(\state, \action, \state')$ instead of $\transMat(\state, \action)(\state')$. 
A \emph{state-action pair}~$(\state, \action) \in \states \times \actions$ is \emph{enabled} if $\transMat(\state, \action)$ is defined. We denote the set of enabled pairs by $\SA_\mdp$ and the set of enabled actions in $\state$ by $\actions(\state)$.
A state $\state$ is \emph{absorbing} if $\transMat(\state, \action, \state) = 1$ for all $\action \in \actions(\state)$. A \emph{path} $\Path$ in $\mdp$ is a sequence $\Path = \state_0 \action_0 \state_1 \action_1 \ldots$ where $\transMat(\state_i, \action_i, \state_{i+1}) > 0$ for all $i \geq 0$. The last state of a finite path $\Path$ is denoted by $\last{\Path}$. The set of all infinite (resp. finite) paths starting in $\init$ is denoted by $\paths(\mdp)$ (resp. $\pathsFin(\mdp)$). 

The \emph{sub-MDP} $\mdp[D]$ induced by $D \subseteq \states$ consists of states $D$ and all state-action pairs $(\state, \action) \in \SA$ with $\state \in D$ and $\transMat(\state, \action, D) = 1$. %

A \emph{discrete-time Markov chain} (DTMC) $\dtmc$ is an MDP with a single action enabled in all states. Thus, we omit the actions from DTMCs and write $(\states, \init, \transMat)$.%

A \emph{scheduler} is a function $\scheduler \colon \pathsFin(\mdp) \to \distr{\actions}$ with $\supp{\scheduler(\Path)} \subseteq \actions(\last{\Path})$ for all $\Path$. The set of schedulers for $\mdp$ is denoted by $\schedulers^\mdp$. We consider the standard probability measure $\prob_{\mdp, \init}^\scheduler$ on infinite paths (\cite[Ch. 10]{baier_principles_2008-1}). For $\targetSet \subseteq \states$, let $\eventually \targetSet$, $\eventually \globally \targetSet$ and $\globally \eventually \targetSet$ denote the events of eventually visiting $G$, eventually staying in $G$ and visiting $G$ infinitely often, respectively. \emph{Rabin} (resp. \emph{Streett}) properties are of the form $\Disj_{(F, E) \in \pairs} (\globally \eventually F \land \eventually \globally E)$ (resp. $\Conj_{(F, E) \in \pairs} (\eventually \globally F \lor \globally \eventually E)$) where $\pairs \subseteq 2^\states \times 2^\states$. We often identify Rabin properties with $\pairs$. The \emph{dual} (or negation) of a Rabin property is a Streett property, and vice versa.

A set $\emptyset \subset \ec \subseteq \SA_\mdp$ is an \emph{end component} (EC) of $\mdp$ if the induced sub-MDP is \emph{strongly connected}. We often identify $\ec$ with the induced sub-MDP and denote its states by $\states(\ec)$. An EC $\mec$ is \emph{maximal} (MEC) if there is no EC $\ec'$ with $\mec \subset \ec'$. The set of ECs (resp. MECs) of $\mdp$ is denoted by $\ECS(\mdp)$ (resp. $\MECS(\mdp)$). Rabin properties can be satisfied \emph{almost surely} in an EC $\ec$, i.e. $\sup_{\scheduler}\prob_{\ec}^\scheduler(\Disj_{(F, E) \in \pairs} (\globally \eventually F \land \eventually \globally E)) = 1$, if $\states(\ec) \intersection F \neq \emptyset$ and $\states(\ec) \subseteq E$ for some $(F, E) \in \pairs$. We write $\ec \modelsRabin \pairs$ to denote the satisfaction.

The \emph{product} of $\mdp$ and a deterministic automaton $\automaton = (\autoStates, \states, \delta, \autoInit, \acceptanceCondition)$ is given by the MDP
$
\mdp \times \automaton = (\states \times \autoStates, \actions, \tuple{\init, \delta(\autoInit, \init)}, \transMat')
$
where for all $(\state, \action) \in \SA$, $\state \in \states$ and $\autoState, \autoState' \in \autoStates$ we have $\transMat'\bigl(\tuple{\state, \autoState}, \action, \tuple{\state', \autoState'} \bigr) = \transMat(\state, \action, \state')$ if $\autoState' = \autoTransition(\autoState, \state')$ and zero otherwise. A Rabin condition $\acceptanceCondition$ can be lifted to the Rabin property $\Disj_{(F, E) \in \acceptanceCondition}\globally \eventually (\states \times F) \land \eventually \globally (\states \times (\autoStates \setminus E))$. Streett conditions are lifted analogously.

\medskip

\noindent\textbf{Multi-objective query.}
A \emph{multi-objective $\omega$-regular query} $\query$ is a \emph{syntactic expression} $\exists \scheduler \centerdot \Conj_{i=1}^k \queryProb_{\triangleright \lambda_i}(\prop_i)$ or $\forall \scheduler \centerdot \Disj_{i=1}^k \queryProb_{\triangleright \lambda_i}(\prop_i)$ where $\triangleright \in \{\geq, >\}$, $\prop_i \subseteq \states^\omega$ is an $\omega$-regular property and $\lambda_i \in [0,1]$ for all $i \in [k]$. The former is a $\existsCQ$- and the latter a $\universalDQ$-query.\todo{Is $\queryProb$ defined? Maybe consider emphasizing that it depends on $\scheduler$. Also, not sure if the question-mark helps in this notation?}
We call $\query$ a \emph{Rabin} (resp. \emph{Streett}) \emph{query} if it only contains Rabin (resp. Streett) properties.
Given an MDP $\mdp$, we define $\models$ by
{
\begin{center}
$\mdp \models \exists \scheduler \centerdot \Conj_{i=1}^k \queryProb_{\triangleright \lambda_i}(\prop_i) \quad \text{iff} \quad \exists \scheduler \in \schedulers^\mdp \centerdot \Conj_{i=1}^k \prob_{\mdp, \init}^\scheduler(\prop_i) \triangleright \lambda_i$
 $\mdp \models \forall \scheduler \centerdot \Disj_{i=1}^k \queryProb_{\triangleright \lambda_i}(\prop_i) \quad \text{iff} \quad \forall \scheduler \in \schedulers^\mdp \centerdot \Disj_{i=1}^k \prob_{\mdp, \init}^\scheduler(\prop_i) \triangleright \lambda_i$
\end{center}
}%
We emphasize the \emph{duality} between $\existsCQ$- and $\universalDQ$-queries, i.e. the negation of a $\existsCQ$-query is a $\universalDQ$-query and vice versa. We also refer to the negated query as \emph{dual query}. Hence, by definition, a query is satisfied iff its dual is not.
A \emph{subsystem} $\mdp'$ of an MDP $\mdp$ is a (sub-stochastic\footnote{In \cite{jantsch_certificates_2022,funke_farkas_2020,baier_certificates_2024}, subsystems are defined with an explicit trap state $\exit$. However, this is equivalent to the definition provided here. More details are given in \Cref{appendix:subsystems}.}) MDP $(\states', \actions, \init, \transMat')$ where $\init \in \states' \subseteq \states$ and $\transMat'(\state, \action, \state') = \transMat(\state, \action, \state')$ for all $(\state, \action) \in \SA_\mdp$ with $\state \in \states'$ and $\state' \in \states'$.
A subsystem is induced by $\states' \subseteq \states$ if its states are given by $\states'$. We denote such induced subsystems by $\mdp_{\states'}$.
A \emph{witnessing subsystem} for a query $\query$ is a subsystem $\mdp'$ with $\mdp' \models \query$. If $\mdp'$ satisfies $\query$ then $\mdp$ does as well \cite{jantsch_certificates_2022}.
\begin{remark}
It suffices to consider lower-bounded queries since $\omega$-regular properties are closed under complementation, i.e. $\prob_{\mdp, \init}^\scheduler(\prop) \leq \lambda$ iff $\prob_{\mdp, \init}^\scheduler(\bar{\prop}) \geq 1 - \lambda_i$. Further, $(\exists,\lor)$- and $(\forall,\land)$-queries can be reduced to the considered query types.
\end{remark}

\section{Overview}
\label{section:overview}
We provide an overview of our approach to certificates and witnesses for multi-objective $\omega$-regular queries. First, we revisit the procedure from \cite{forejt_quantitative_2011} for \emph{checking} whether a query is satisfied. 
Let $\mathcal{N}$ be an MDP and $\rho_1, \ldots, \rho_k$ be $\omega$-regular properties recognized by DRAs $\automaton_1, \ldots, \automaton_k$. The \emph{complement} properties are given by $\bar{\rho}_1, \ldots, \bar{\rho}_k$. Note that $\bar{\rho}_i$ is recognized by $\automaton_i$ when interpreting it as a Streett automaton. Let $\existsQuery_\mathcal{N} = \exists \scheduler \centerdot \Conj_{i=1}^k \queryProb_{\triangleright \lambda_i}(\rho_i)$ and $\forallQuery_\mathcal{N} = \forall \scheduler \centerdot \Disj_{i=1}^k \queryProb_{\triangleright \lambda_i}(\bar{\rho}_i)$.

\smallskip

\noindent\textbf{Step 1}: The product MDP $\mdp = \mathcal{N} \times \automaton_1 \times \cdots \times \automaton_k$ is constructed to reason about the probabilities of satisfying the properties. The acceptance conditions are lifted to Rabin (resp. Streett) properties $\prop_1, \ldots, \prop_k$, resulting in the corresponding \emph{Rabin} query $\existsQuery = \exists \scheduler \centerdot \Conj_{i=1}^k \queryProb_{\triangleright \lambda_i}(\prop_i)$ (resp. Streett query $\forallQuery = \forall \scheduler \centerdot \Disj_{i=1}^k \queryProb_{\triangleright \lambda_i}(\bar{\prop}_i)$). Note that $\mathcal{N} \models \existsQuery_{\mathcal{N}}$ iff $\mdp \models \existsQuery$ (resp. $\mathcal{N} \models \forallQuery_\mathcal{N}$ iff $\mdp \models \forallQuery$), i.e. the query is satisfied in $\mathcal{N}$ iff the corresponding query is satisfied in the product MDP. We denote the \emph{Rabin} pairs of property $\prop_i$ by $\pairs_i$.

\smallskip

\noindent\textbf{Step 2}: Now the satisfying MECs of the product MDP $\mdp$ are identified. For every MEC $\mec \in \MECS(\mdp)$ and $I \subseteq [k]$, it is checked if there is an EC $\ec' \subseteq \mec$ that satisfies all \emph{Rabin} properties in $\{\prop_i \mid i \in I\}$, i.e. for all $i \in I$ there is a Rabin pair $(F, E) \in \pairs_i$ with $\states(\ec') \intersection F \neq \emptyset$ and $\states(\ec') \subseteq E$. For a fixed $I$, this can be done in time polynomial in $\card{\mdp}$ and exponential in the number of properties \cite{forejt_quantitative_2011}. All pairs $(\states(\mec), I)$ of such satisfying MECs $\mec$ and $I$ are collected in the set $\mathcal{I}^*$.

\smallskip

\noindent\textbf{Step 3}: To analyse the probabilities of realizing the identified ECs, a \emph{quotient} of $\mdp$ is considered. For a state partition $\partition$ of $\mdp$, we write $\mecClass[\partition]{\state}$ to denote the class of $\state$ and let $\SA_{\partition} = \{(\state,\action) \in \SA \mid \transMat(\state, \action, \mecClass[\partition]{\state}) < 1\}$ be the set of pairs leaving their class. W.l.o.g. assume that actions of states in $\mdp$ are pairwise disjoint.
\begin{definition}
Let $\partition$ be a state partition of $\mdp$ and $\mathcal{I} \subseteq \partition \times 2^{[k]}$.
The \emph{quotient MDP} w.r.t. $\partition$ and $\mathcal{I}$ is $\mdp_{/ \partition}^{\mathcal{I}} = (\partition \union \{\bot_I \}_{I \subseteq [k]}, \actions \union \{\tau_I\}_{I \subseteq [k]}, \mecClass[\partition]{\init}, \transMat')$ where
\begin{itemize}[align=left, leftmargin=*, itemsep=0mm, topsep=0.8mm, parsep=0mm]
\item $\transMat'(\mecClass[\partition]{\state}, \action, \mecClass[\partition]{\state'}) = \transMat(\state, \action, \mecClass[\partition]{\state'})$ for all $(\state, \action) \in \SA_\partition$ and $\state' \in \states$,
\item $\transMat'(\mecClass[\partition]{\state}, \tau_I, \exit_I) = 1$ for all $(\mecClass[\partition]{\state}, I) \in \mathcal{I}$, and $\exit_I$ is absorbing for all $I \subseteq [k]$.
\end{itemize}
\end{definition}
In $\mdp_{/ \partition}^{\mathcal{I}}$ classes of $\partition$ are collapsed into a single state, class-internal actions are omitted and transitions to absorbing states $\bot_I$ are added according to $\mathcal{I}$.

Let $\partition_{\MECS}$ be the partition of the states of $\mdp$ according to the MEC they are contained in (states not part of any MEC are put in singletons). Together with the set $\mathcal{I}^*$ identified in Step 2, the \emph{MEC quotient} $\mecQuotient \coloneqq \mdp_{/ \partition_{\MECS}}^{\mathcal{I}^*}$ is constructed, where MECs $\mec$ of $\mdp$ are collapsed into a single state and transitions from such state to an absorbing state $\exit_I$ are added if it is possible to satisfy all properties $\{\prop_i \mid i \in I\}$ in $\mec$. Intuitively, reaching $\exit_I$ corresponds to staying in an EC in which the indexed properties are satisfied.
\begin{restatable}{theorem}{theoremReduction}
\label{lemma:index}
Let $\widehat{\query} = \exists \scheduler \centerdot \Conj_{i=1}^k \queryProb_{\triangleright \lambda_i}(\eventually G_i)$ and $\widehat{\Phi} = \forall \scheduler \centerdot \Disj_{i=1}^k \queryProb_{\triangleright \lambda_i}(\eventually \bar{G}_i)$ where $G_i = \{\exit_I \mid I \subseteq [k], i \in I\}$ and $\bar{G}_i = \{\exit_I \mid I \subseteq [k], i \notin I\}$ for all $i \in [k]$. Then:
{
\begin{center}
\begin{enumerate*}[1)]
\item $\mdp \models \query \iff \mdp_{/ \partition_{\MECS}}^{\mathcal{I}^*} \models \widehat{\query}$. \qquad\qquad\qquad
\item $\mdp \models \Phi \iff \mdp_{/ \partition_{\MECS}}^{\mathcal{I}^*} \models \widehat{\Phi}$.
\end{enumerate*}
\end{center}
}
\end{restatable}
\noindent For $\forallQuery = \forall \scheduler \centerdot \Disj_{i=1}^k \queryProb_{\triangleright \lambda_i}(\bar{\prop}_i)$, each $\bar{\targetSet}_i$ contains the absorbing states corresponding to \emph{not} satisfying $\prop_i$. In total, \Cref{lemma:index} reduces the problem of checking a multi-objective Rabin (or Streett) query to checking a \emph{reachability} query $\widehat{\existsQuery}$ (resp. $\widehat{\forallQuery}$), which can be done via LP \cite{etessami_multi-objective_2008} or value iteration \cite{forejt_pareto_2012}.

\begin{figure}[t]
\centering
\begin{subfigure}{0.2\textwidth}
\centering
\scalebox{0.75}{
\begin{tikzpicture}[x=18mm,y=15mm,font=\small]
          \node[state] (s0) {$s_0$};
          \node[state] (s3) [below = 0.82cm of s0]  {$s_3$};
          \node[state] (s1) [right = 0.4cm  of s0] {$s_1$};
          \node[state] (s2) [right = 0.4cm of s1]  {$s_2$};
          \node[state] (s4) [below = 0.82cm of s2]  {$s_4$};

          \node[bullet] (s0s3s1) [below = .2cm of s0] {};
          
                    \node (init) [above = 0.3cm of s0] {};
          \draw (init) edge[ptran] (s0);
          
            {\color{green!50!black}
				\draw (s0) -- node[left,pos=.8]{$a$} (s0s3s1);
				\draw (s0s3s1) edge[ptran] coordinate[pos=.3] (bs0s3) node[right,pos=.5]{\scriptsize{$0.5$}} (s3);
				\draw (s0s3s1) edge[ptran, bend right, out=0.2] coordinate[pos=.08] (bs0s1) node[above,pos=.5, yshift=-0.3mm]{\scriptsize{$0.5$}} (s1);
				\draw (bs0s3) to[bend right] (bs0s1);
			}
			
			{\color{orange!90}
			
				\draw (s1) edge[ptran, bend right=18] node[below,pos=.5]{$b$} (s2);
			}
			
			{\color{red!100}
			
				\draw (s2) edge[ptran, bend right=18] node[above,pos=.5]{$c$} (s1);
			}
			
			{\color{blue!90}
			
				\draw (s2) edge[ptran, loop below, looseness=6] node[below,pos=.5]{$d$} (s2);
			}
			
			{\color{darkgray}
			
				\draw (s3) edge[ptran, bend right=14] node[above,pos=.5]{$e$} (s4);
			}
			
			{\color{purple}
			
				\draw (s4) edge[ptran, bend right=14] node[above,pos=.5]{$f$} (s3);
			}
			
\end{tikzpicture}
}
\caption{\scriptsize{MDP $\mdp$}}
\label{subfig:product-mdp}
\end{subfigure}
\hfill
\begin{subfigure}{0.26\textwidth}
\centering
\scalebox{0.75}{
\begin{tikzpicture}[x=18mm,y=15mm,font=\scriptsize]
          \node[state] (s0) {$\{s_0\}$};
          \node[state] (s12) [right = 0.5cm  of s0] {$\{s_1, s_2\}$};
          \node[state] (s34) [right = 0.2cm of s12]  {$\{s_3, s_4\}$};

          \node[state] (exit1) [below = 0.4cm  of s0, fill=gray!25] {$\exit_{\{1\}}$};
          \node[state] (exit2) [below = 0.4cm  of s12, fill=gray!25] {$\exit_{\emptyset}$};
          \node[state] (exit3) [below = 0.4cm  of s34, fill=gray!25] {$\exit_{\{2\}}$};

          \node[bullet] (s0s3s1) [above right = .2cm of s0, xshift=0.2cm, yshift=-2mm] {};
          
                    \node (init) [above = 0.3cm of s0] {};
          \draw (init) edge[ptran] (s0);
          
            {\color{green!50!black}
				\draw (s0) -- node[above,pos=1]{$a$} (s0s3s1);
				
				\draw (s0s3s1) edge[ptran, bend left=20] coordinate[pos=.06] (bs0s3) node[below,pos=.5, xshift=0.5cm, yshift=0mm]{\scriptsize{$0.5$}} (s34.north);
				
				\draw (s0s3s1) edge[ptran, bend right] coordinate[pos=.3] (bs0s1) node[below,pos=.25, yshift=-1mm]{\scriptsize{$0.5$}} (s12.west);
				
				\draw (bs0s3) to[bend left] (bs0s1);
			}
			
			{\color{black}
			
				\draw (s12) edge[ptran] (exit1);
			}
			{\color{black}
			
				\draw (s12) edge[ptran] (exit2);
			}
			{\color{black}
			
				\draw (s12) edge[ptran] (exit3);
			}

			{\color{black}
			
				\draw (s34) edge[ptran] (exit3);
			}
			{\color{black}
			
				\draw (s34) edge[ptran] (exit2);
			}
			
\end{tikzpicture}
}
\caption{\scriptsize{MEC quotient $\mecQuotient$}}
\label{subfig:mec-quotient}
\end{subfigure}
\hfill
\begin{subfigure}{0.2\textwidth}
\centering
\scalebox{0.75}{
\begin{tikzpicture}[x=18mm,y=15mm,font=\scriptsize]
          \node[state] (s0) {$\{s_0\}$};
          \node[state] (s12) [right = 0.5cm  of s0] {$\{s_1, s_2\}$};

          \node[state] (exit1) [below = 0.4cm  of s0, fill=gray!25] {$\exit_{\{1\}}$};
          \node[state] (exit2) [right = 0.3cm  of exit1, fill=gray!25] {$\exit_{\emptyset}$};
          \node[state] (exit3) [right = 0.3cm  of exit2, fill=gray!25] {$\exit_{\{2\}}$};

                    \node (init) [above = 0.3cm of s0] {};
          \draw (init) edge[ptran] (s0);
          
            {\color{green!50!black}
				\draw (s0) edge[ptran] node[above,pos=0.5]{$a$} node[below,pos=0.5]{$0.5$} (s12.west);
				
			}
			
			{\color{black}
			
				\draw (s12) edge[ptran] (exit1);
			}
			{\color{black}
			
				\draw (s12) edge[ptran] (exit2);
			}
			{\color{black}
			
				\draw (s12) edge[ptran] (exit3);
			}
			
\end{tikzpicture}
}
\caption{\scriptsize{Subsystem $\mecQuotient'$}}
\label{subfig:mec-quotient-subsystem}
\end{subfigure}
\hfill
\begin{subfigure}{0.2\textwidth}
\centering
\scalebox{0.75}{		
\begin{tikzpicture}[x=18mm,y=15mm,font=\small]
          \node[state] (s0) {$s_0$};
          \node[state] (s1) [right = 0.4cm  of s0] {$s_1$};
          \node[state] (s2) [right = 0.4cm of s1]  {$s_2$};

                    \node (init) [above = 0.3cm of s0] {};
          \draw (init) edge[ptran] (s0);
          
            {\color{green!50!black}
				\draw (s0) edge[ptran] coordinate[pos=.08] (bs0s1) node[above,pos=.5, yshift=-0.3mm]{\scriptsize{$a$}} node[below,pos=.5, yshift=-0.3mm]{\scriptsize{$0.5$}} (s1);
			}
			
			{\color{orange!90}
			
				\draw (s1) edge[ptran, bend right=18] node[below,pos=.5]{$b$} (s2);
			}
			
			{\color{red!100}
			
				\draw (s2) edge[ptran, bend right=18] node[above,pos=.5]{$c$} (s1);
			}
			
			{\color{blue!90}
			
				\draw (s2) edge[ptran, loop above, looseness=6] node[above,pos=.5]{$d$} (s2);
			}

\end{tikzpicture}
}

\caption{\scriptsize{Subsystem $\mdp'$}}
\label{subfig:witness}
\end{subfigure}
\caption{MDP, MEC quotient and subsystems. We omit probability $1$ labels and sometimes self-loops and action names for readability.}
\label{fig:mdps}
\end{figure}
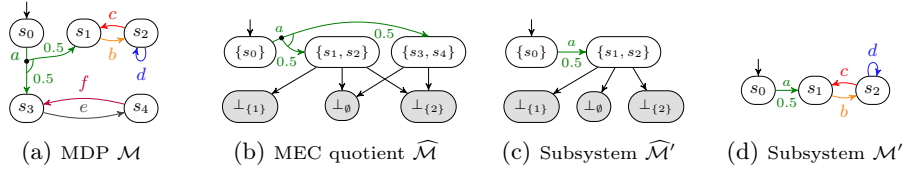

\begin{example}
Let $\query = \exists \scheduler \centerdot \queryProb_{\geq 0.25}(\varphi_1) \land  \queryProb_{\geq 0.75}(\varphi_2) $ where $\varphi_1 = \globally \eventually \state_1 \lor \eventually \globally \state_4$ and $\varphi_2 = \eventually \globally \{\state_2, \state_3, \state_4\}$ and consider the MDP $\mdp$ in \cref{subfig:product-mdp}.
The Rabin properties $\varphi_1$ and $\varphi_2$ can be satisfied \emph{almost surely} in the MEC formed by $\state_1$ and $\state_2$, but \emph{not simultaneously}. While $\varphi_1$ requires that $\state_1$ is visited infinitely often, $\varphi_2$ demands it to be visited finitely often. We have $\partition_{\MECS} = \{\{\state_0\}, D_1, D_2\}$ where $D_1 = \{\state_1, \state_2\}$ and $D_2 = \{ \state_3, \state_4\}$. Further, $\mathcal{I}^*$ contains the pairs $(D_1, \{1\})$, $(D_1, \{2\})$ and $(D_2, \{2\})$ resulting in the MEC quotient $\mecQuotient$ in \cref{subfig:mec-quotient} and reachability query $\widehat{\query} = \exists \scheduler \centerdot \queryProb_{\geq 0.25}(\eventually \bot_{\{1\}}) \land  \queryProb_{\geq 0.75}(\eventually \bot_{\{2\}})$. We have $\mecQuotient \models \widehat{\query}$ and thus $\mdp \models \query$.
\label{example:reduction}
\end{example}

\noindent\textbf{Certifying the procedure.} Our goal is to provide certificates with which it can be validated that the returned result of this procedure is correct. If a query is determined to be satisfied, we provide certificates that prove the satisfaction. Otherwise, we return certificates for the satisfaction of the \emph{dual query}.

Our certificates will be composed of certificates demonstrating the correctness of the individual steps. No explicit certificate is needed to validate that the product MDP $\mdp$ has been constructed correctly in Step 1, since $\mdp$ can be returned as a by-product of the procedure and for each of its states it can be checked whether the transition function is defined correctly. The \emph{crux} is to provide certificates for the multi-objective Rabin (resp. Streett) queries, which will be based on \Cref{lemma:index}. The idea is to return the sets $\partition$ and $\mathcal{I}$ computed by the procedure and certify that $\partition = \partition_{\MECS}$ and $\mathcal{I} \subseteq \mathcal{I}^*$ (resp. $\mathcal{I} \supseteq \mathcal{I}^*$) hold. For this we adapt and improve the certificates from \cite{jantsch_certificates_2022} for the MEC decomposition in \Cref{subsection:satisfying-ecs}. For the reachability queries in \Cref{lemma:index}, we use the certificates from \cite{baier_certificates_2024} and discuss means to compute them in \Cref{subsection:certificate-rabin-streett}. Finally, we combine the introduced certificates in \Cref{subsection:certificate-rabin-streett} to certify Rabin and Streett queries.

\medskip

\noindent\textbf{Witnesses.} To increase the understandability of the returned result, we additionally use \Cref{lemma:index} in \Cref{section:witnesses-mdps} to obtain minimal witnesses. \Cref{section:dtmc} shows that for DTMCs we can also find witnesses via UBAs instead of DRAs.

\section{Certificates for End Components}
\label{subsection:satisfying-ecs}
In this section we fix an MDP $\mdp = (\states, \actions, \init, \transMat)$ and Rabin properties $P_1, \ldots, P_k$. We start with a certificate for the existence of an EC $\ec$ in $\mdp$ in which all properties are satisfied simultaneously, i.e. for all $i \in [k]$ we have $\ec \modelsRabin P_i$.
\begin{lemma}[EC certificate {\cite[Lemma 3.27]{jantsch_certificates_2022}}]
Let $D \subseteq \states$. An \emph{EC certificate} for $D$ is a tuple $(\mathtt{f}, \allowbreak \mathtt{b}) \in \naturals^D \times \naturals^D$ such that there is a unique $\state^* \in D$ with $\mathtt{f}(\state^*) = \mathtt{b}(\state^*) = 0$ and for each $\state \in D \setminus \{\state^*\}$:
\begin{enumerate}[1), align=left, leftmargin=*, itemsep=0mm, topsep=0.8mm, parsep=0mm]
\item $\exists \action \in \actions(\state) \centerdot \exists \state' \in D \centerdot \mathtt{f}(\state) > \mathtt{f}(\state')\land \transMat(\state, \action, \state') > 0 \land \transMat(\state, \action, D) = 1$.
\item $\exists \state' \in D \centerdot \exists \action' \in \actions(\state) \centerdot \mathtt{b}(\state) > \mathtt{b}(\state') \land \transMat(\state', \action', \state) > 0 \land \transMat(\state', \action', D) = 1$.
\end{enumerate}
\label{lemma:ec-certificate}
The sub-MDP $\mdp[D]$ is strongly connected iff there is an EC certificate for $D$.
\end{lemma}
Intuitively, $\mathtt{f}$ shows that all states in $D$ can reach $\state^*$, and $\mathtt{b}$ that $\state^*$ can reach all states in $D$. They can both be computed via breadth-first search (BFS).
We call $D$ \emph{trivial} if it is a singleton $\{\state\}$ and there is no action $\action \in \actions(\state)$ with $\transMat(\state, \action, \state) = 1$.
To certify a satisfying EC $\ec$, we return an EC certificate $(\mathtt{f}, \mathtt{b}) \in \naturals^{\states(\ec)} \times \naturals^{\states(\ec)}$ for $\states(\ec)$. It is validated by checking that it is non-trivial and for all $i \in [k]$ there is $(F, E) \in \pairs_i$ such that $\states(\ec) \cap F \neq \emptyset$ and $\states(\ec) \subseteq E$.
We can also certify the \emph{MEC decomposition}, i.e. given a partition $\mathcal{D} \subseteq 2^\states$, we define a certificate that proves that $\mathcal{D} = \mathcal{D}_{\MECS}$. To this end, we combine the EC certificates with the certificates for positive reachability from \cite{chatterjee_fixed_2025}.
\begin{restatable}[MEC certificate]{proposition}{mecCertificate}
An \emph{MEC certificate} for a partition $\partition$ is a tuple $(\mathtt{ECS}, \mathtt{r})$ where $\mathtt{ECS}$ contains \emph{EC certificates} for all $D \in \partition$ and $\mathtt{r} \in \naturals^\partition$ such that $\mathtt{r}(\mecClass[\partition]{\state}) \geq 1 + \min_{\state' \in \supp{\transMat(\state, \action)}} \mathtt{r}(\mecClass[\partition]{\state'})$ for all $(\state, \action) \in \SA_\mathcal{D}$.
We have $\partition = \partition_{\MECS}$ if and only if there exists an MEC certificate for $\partition$.
\label{lemma:mec-certificate}
\end{restatable}
The EC certificates show that each class in $\partition$ is strongly connected and the vector $\mathtt{r}$ from \cite[Proposition~1]{chatterjee_fixed_2025} proves that each class is \emph{maximal} in that regard. Intuitively, $\mathtt{r}$ shows that the quotient induced by $\partition$ has no ECs and therefore $\partition$ corresponds to the MEC decomposition. Similar certificates have been presented in \cite{jantsch_certificates_2022}, but instead of $\mathtt{r}$, use a certificate for \emph{expected rewards} to show maximality. The latter needs to be computed via LPs or value iteration, while $\mathtt{r}$ can be determined via a BFS \cite{chatterjee_fixed_2025}, making the generation of our MEC certificates more efficient. From an MEC certificate we can directly obtain the corresponding partition $\partition$, since each EC certificate corresponds to one of its classes. Thus, we sometimes incorporate $\partition$ into the certificate and write $(\mathtt{EC}, \mathtt{r}, \partition)$.

Now we certify the \emph{absence} of an EC $\ec$ that satisfies all Rabin properties at the same time, i.e. there is no EC $\ec$ with $\ec \modelsRabin \pairs_i$ for all $i \in [k]$.
An EC $\ec$ that satisfies a combination $(F_1, E_1) \in \pairs_1, \ldots, (F_k, E_k) \in \pairs_k$ may only contain states from $\bigcap_{i=1}^k E_i$. Thus, we consider the induced sub-MDP and show that each of its MECs violates the conditions that states from $F_i$ need to be contained. %
\begin{definition}[Absence certificate]
An \emph{absence certificate} $\mathtt{Abs}$ for Rabin properties $\pairs_1, \ldots, \pairs_k$ is a set of \emph{MEC certificates}. For all combinations $(F_1, E_1) \in \pairs_1, \ldots, (F_k, E_k) \in \pairs_k$, there is $(\mathtt{EC}, \mathtt{r}, \partition) \in \mathtt{Abs}$ for the sub-MDP induced by $\bigcap_{i=1}^k E_i$, such that each $D \in \partition$ is trivial or $D \intersection F_i = \emptyset$ for some $i \in [k]$.
\end{definition}
\begin{restatable}{proposition}{absenceCertificate}
\label{proposition:absence-certificate}
There does not exist an EC $\ec$ in $\mdp$ with $\Conj_{i=1}^k \ec \modelsRabin \pairs_i$ if and only if there exists an \emph{absence certificate} $\mathtt{Abs}$ for $\pairs_1, \ldots, \pairs_k$.
\end{restatable}
\begin{example}
Recall the MDP $\mdp$ and properties $\varphi_1, \varphi_2$ from \Cref{example:reduction}. To certify that the partition $\partition = \{\{\state_0\}, \{\state_1, \state_2\}, \{\state_3, \state_4\}\}$ corresponds to $\partition_{\MECS}$, we show that $\{\state_1, \state_2\}$ is an EC, by using the EC certificate $(\mathtt{f}, \mathtt{b})$ with $\mathtt{f}(\state_1) = \mathtt{b}(\state_1) = 0$ and $\mathtt{f}(\state_2) = \mathtt{b}(\state_2) = 1$ and similarly for $\{\state_3, \state_4\}$. Lastly, $\mathtt{r} = (2, 1, 1)$ demonstrates the maximality of $\partition$. To certify that $\varphi_1$ and $\varphi_2$ cannot be simultaneously satisfied in $\{\state_1, \state_2\}$, we consider all combinations of Rabin pairs of $\pairs_1 = \{(\{\state_1\}, \states), (\states, \{\state_4\})\}$ and $\pairs_2 = \{(\states, \{\state_2, \state_3, \state_4\})\}$. For instance, for $(\{\state_1\}, \states)$ and $(\states, \{\state_2, \state_3, \state_4\})$ we provide a MEC certificate for the sub-MDP induced by $\state_2$ which then certifies that such combination is not satisfied.
\end{example}

The size of the EC and MEC certificates are linear in the size of the MDP $\card{\mdp}$. In contrast, the size of the absence certificates is linear in $\card{\mdp}$, but exponential in the number of properties $k$. Below we state that deciding the existence of a satisfying EC is \textsf{NP}-complete. Thus, unless $\textsf{NP} = \textsf{coNP}$, small certificates for the absence of a satisfying EC can generally not be hoped for.

\begin{restatable}{lemma}{hardness} \todo{CC: I could not find any literature on that, so I proved it in the appendix.}
For an MDP $\mdp$ and Rabin properties $\pairs_1, \ldots, \pairs_k$, deciding the existence of an EC $\ec$ in $\mdp$ with $\ec \modelsRabin \pairs_i$ for all $i \in [k]$ is \textsf{NP}-complete.
\end{restatable}

\section{Certificates for Rabin and Streett Queries}
\label{subsection:certificate-rabin-streett}
Before we can define the certificates for Rabin and Streett queries, we revisit the certificates for multi-objective reachability queries from \cite{baier_certificates_2024}. In particular, we discuss how the algorithm from \cite{forejt_pareto_2012} can be modified to compute them, which has not been considered in \cite{baier_certificates_2024}. An MDP $\mdp = (\states \union \targets, \actions, \init, \transMat)$ is said to be in \emph{reachability form} \cite{funke_farkas_2020,jantsch_certificates_2022} if $F$ contains all absorbing states and all states can reach $F$. For such MDPs, we redefine $\SA$ to exclude state-action pairs of $F$.
The \emph{system matrix} \cite{funke_farkas_2020,jantsch_certificates_2022} $\SM \in \reals^{\SA \times \states}$ of $\mdp$ is defined by $\SM((\state, \action), t) = \vect{1}_{\state}(t) - \transMat(\state, \action, t)$ for all $(\state, \action) \in \SA$ and $t \in \states$. For $\targetSet_1, \ldots, \targetSet_k \subseteq F$, let $\TM \in \reals^{\SA \times [k]}$ be defined by $\TM((\state, \action), i) = \transMat(\state, \action, \targetSet_i)$ for all $(\state, \action) \in \SA$ and $i \in [k]$. %
\begin{lemma}[Farkas certificates \cite{baier_certificates_2024}] Let $\mdp = (\states \union \targets, \actions, \init, \transMat)$ be an MDP in reachability form. Further, let $\targetSet_1, \ldots, \targetSet_k \subseteq F$. Then, we have:
{
\begin{enumerate}[1), align=left, leftmargin=*, itemsep=0mm, topsep=0.8mm, parsep=0mm]
\item $\mdp \models \exists \scheduler \centerdot \Conj_{i = 1}^k \queryProb_{\triangleright \lambda_i}(\eventually \targetSet_i) \Leftrightarrow \exists \vect{y} \in \realsnn^\SA \centerdot \SM^\top \vect{y} \leq \initDistr \land \TM^\top \vect{y} \triangleright \boldLambda$.
\end{enumerate}
Further, if all ECs in $\mdp$ are formed by states in $F$:
\begin{enumerate}[1), align=left, leftmargin=*, itemsep=0mm, topsep=0.8mm, parsep=0mm]
\setcounter{enumi}{1}
\item $\mdp \models \forall \scheduler \centerdot \Disj_{i = 1}^k \queryProb_{\triangleright \lambda_i}(\eventually \targetSet_i) \Leftrightarrow \exists \vect{x} \in \reals^\states \centerdot \exists \vect{z} \in \realsnn^{[k]} \centerdot \SM \vect{x} \leq \TM \vect{z} \land \vect{x}(\init) \triangleright \pmb{\lambda}^\top \vect{z}$.
\end{enumerate}
}%
\label{lemma:reach-certs}
\end{lemma}
The vector $\vect{y}$ can be interpreted as the expected number of times a state-action pair is visited under a witness scheduler and $\vect{x}$ and $\vect{z}$ are solutions to a \emph{weighted reachability problem}.\todo{So $\vect{z}$ is a means to distribute the probability on the different targets?} See \cite{baier_certificates_2024} for more details. Given such solutions, we can easily validate whether a query is indeed satisfied, by checking if the inequalities are satisfied. Let $\mathcal{F}_\triangleright(\mdp, \{\targetSet_i\}_{i \in [k]}, \boldLambda)$ and $\mathcal{H}_\triangleright(\mdp, \{\targetSet_i\}_{i \in [k]}, \boldLambda)$ denote the polyhedra formed by the conditions for $\existsCQ$- and $\universalDQ$-queries in \Cref{lemma:reach-certs}, respectively.

Multi-objective reachability queries can be verified via LP \cite{etessami_multi-objective_2008}, but the value (or policy) iteration algorithm from \cite{forejt_pareto_2012} has been observed to be more efficient in practice \cite{quatmann_verification_2023}. In \cite{baier_certificates_2024} only the computation of the certificates via LP was considered. We briefly outline how the algorithm from \cite{forejt_pareto_2012} can be modified, so that Farkas certificates are produced. Given an MDP $\mdp$ and query $\query = \exists \scheduler \centerdot \Conj_{i = 1}^k \queryProb_{\triangleright \lambda_i}(\eventually \targetSet_i)$, in each iteration, the algorithm from \cite{forejt_pareto_2012} maintains an \emph{under-approximation} $U$ of the set of achievable bounds $\textsf{Ach}$, i.e. the set of all bounds for which $\query$ can be satisfied in $\mdp$. It is checked whether $\boldLambda \notin \mathsf{Ach}$ holds, by solving a weighted reachability problem via value (or policy) iteration. If the result of the check is $\boldLambda \notin \mathsf{Ach}$, then $\mdp \not\models \query$ and the certificate $(\vect{x}, \vect{z})$ essentially arises as by-product. Otherwise, either $U$ is enlarged or it can be determined that $\boldLambda \in U$. In the latter case, the algorithm terminates and concludes that $\mdp \models \query$ holds, but no certificate is generated. To obtain a certificate, we additionally compute a witness scheduler and its expected visiting times, corresponding to the certificate $\vect{y}$. The full modified algorithm is described in \Cref{appendix:certifying-verification-algorithm}.
\begin{remark}
The fixed-point certificates \cite{chatterjee_fixed_2025} for \emph{single-objective reachability} are, in certain cases, similar to Farkas certificates from \cite{funke_farkas_2020,jantsch_certificates_2022}. The former bound the optimal reachability probabilities in all states, while the latter only bound them in the initial state. We compare both certificates and discuss how Farkas certificates can be modified to also bound the probabilities in all states in \Cref{appendix:fixed-point-farkas-comparison}. The modified version is supported by our implementation.
\end{remark}
We are now ready to define the certificates for Rabin and Streett queries. For the remainder of this section, we fix an MDP $\mdp = (\states, \actions, \init, \transMat)$ and queries $\query = \exists \scheduler \centerdot \Conj_{i=1}^k \queryProb_{\triangleright \lambda_i}(\prop_i)$ and $\Phi = \forall \scheduler \centerdot \Disj_{i=1}^k \queryProb_{\triangleright \lambda_i}(\bar{\prop}_i)$ where each $\prop_i$ is a \emph{Rabin property} with pairs $P_i$ and $\bar{\prop}_i$ the \emph{complement Streett property}. We let $\targetSet_i = \{\exit_I \mid i \in I \subseteq [k]\}$ and $\bar{\targetSet}_i = \{\exit_I \mid i \notin I \subseteq [k]\}$ for all $i \in [k]$.
\begin{definition}[$\existsCQ$-certificate]
A \emph{$\existsCQ$-certificate} for $\mdp \models \query$ is a tuple $\left(\partition, \mathcal{I}, \mathtt{MEC}, \allowbreak \mathtt{ECS}, \vect{y}\right)$ where $\partition \subseteq 2^\states$ is a partition and $\mathcal{I} \subseteq \partition \times 2^{[k]}$ such that
\begin{enumerate}[1), itemsep=0mm, topsep=0.8mm, parsep=0mm]
\item $\mathtt{MEC}$ is a \emph{MEC certificate} for partition $\partition$.
\item for all $(D,I) \in \mathcal{I}$ there is $D' \subseteq D$ such that $D'$ is \emph{non-trivial}, $\mathtt{ECS}$ contains an \emph{EC certificate} for $D'$ and $\Conj_{i \in I} \Disj_{(F, E) \in \pairs_i} \left((D' \intersection F \neq \emptyset) \land (D' \subseteq E)\right)$.
\item $\vect{y} \in \mathcal{F}_\triangleright(\mdp_{/ \partition}^\mathcal{I}, \{\targetSet_i\}_{i \in [k]}, \{\lambda_i\}_{i \in [k]})$.
\end{enumerate}
\label{def:exists-certificate}
\end{definition}
\todo{Doesn't the $F,E$ in 2) have to relate to $I$ somehow? Can there be multiple pairs $(D,I_1), (D,I_2) \in \mathcal{I}$ with the same $D$. I think that would not make sense }
\begin{definition}[$\universalDQ$-certificate]
A \emph{$\universalDQ$-certificate} for $\mdp \models \Phi$ is a tuple $\left(\partition, \mathcal{I}, \mathtt{MEC}, \mathtt{ABS}, \vect{x}, \allowbreak \vect{z} \right)$ where $\partition \subseteq 2^\states$ is a partition and $\mathcal{I} \subseteq \partition \times 2^{[k]}$ such that
\begin{enumerate}[1), itemsep=0mm, topsep=0.8mm, parsep=0mm]
\item $\mathtt{MEC}$ is a \emph{MEC certificate} for partition $\partition$.
\item for all nontrivial $D \in \partition$ and $I \subseteq [k]$ with $(D,I) \notin \mathcal{I}$ there is an \emph{absence certificate} for the sub-MDP $\mdp[D]$ and $\{\pairs_i\}_{i \in I}$ in $\mathtt{ABS}$.
\item $(\vect{x}, \vect{z}) \in \mathcal{H}_\triangleright(\mdp_{/ \partition}^\mathcal{I}, \{\bar{\targetSet}_i\}_{i \in [k]}, \{\lambda_i\}_{i \in [k]})$.
\end{enumerate}
\label{def:forall-certificate}
\end{definition}
Recall from \Cref{section:overview} that $\mathcal{I}^* \subseteq \partition_\MECS \times 2^{[k]}$ contains a pair $(\states(\mec), I)$ iff it is possible in an MEC $\mec$ to satisfy the Rabin properties $\{\prop_i\}_{i \in I}$. Due to \Cref{lemma:index}, we need to certify that $\partition = \partition_{\MECS}$ and $\mathcal{I} \subseteq \mathcal{I}^*$ (resp. $\mathcal{I} \supseteq \mathcal{I}^*$) hold. For the former, the MEC certificate $\mathtt{MEC}$ is included in both certificates. For $\existsCQ$-queries, we use the EC certificates in $\mathtt{ECS}$ to show $\mathcal{I} \subseteq \mathcal{I}^*$, by including an EC certificate for each satisfying EC. To certify $\mathcal{I} \supseteq \mathcal{I}^*$ for $\universalDQ$-queries, we show that for each $(D, I) \notin \mathcal{I}$ we have $(D, I) \notin \mathcal{I}^*$. To this end, we include for each $(D, I) \notin \mathcal{I}$ an absence certificate in $\mathtt{ABS}$, showing that it is not possible to satisfy $\{\prop_i\}_{i \in I}$ in $D$ and thus $(D, I) \notin \mathcal{I}^*$. Finally, we use the certificates for multi-objective reachability to certify the satisfaction of the reachability query in the quotient induced by $\partition$ and $\mathcal{I}$. Soundness and completeness are stated below.
\begin{restatable}[Soundness and completeness]{theorem}{soundAndComplete}
\begin{enumerate}[1), itemsep=0mm, topsep=0.8mm, parsep=0mm]
\item $\mdp \models \query$ if and only if there exists a \emph{$\existsCQ$-certificate} for $\query$.
\item $\mdp \models \Phi$ if and only if there exists a \emph{$\universalDQ$-certificate} for $\Phi$.
\end{enumerate}
\end{restatable}
Note that $\partition$ and $\mathcal{I}$ need not be \emph{explicitly} incorporated into the certificates, since they can be directly derived from $\texttt{MEC}$, $\mathtt{ECS}$ and $\mathtt{ABS}$. The size of both certificates is linear in $\card{\mdp}$ and exponential in the number of objectives $k$.

\section{Witnesses for Multi-objective $\omega$-regular Queries}
\label{section:witnesses-mdps}
In this section, we tackle the \emph{witness problem}. Given an MDP $\mathcal{N}$, multi-objective $\omega$-regular query $\query$, a finite \emph{set of labels} $\labels$ and \emph{labeling} $\labeling \colon \states \to 2^\labels$, it asks to find a witnessing subsystem $\mathcal{N}'$ (with states $\states'$) for $\query$, such that $\card{\labeling(\states')} = \card{\Union_{\state \in \states'}\labeling(\state)}$ is \emph{minimal} \cite{funke_farkas_2020,jantsch_certificates_2022}.
The labels and labeling can be used to encode various notions of minimality, e.g. \emph{state-minimality} ($L = \states$ and $\labeling_\mathsf{state}(\state) = \{\state\}$ for all $\state \in \states$).
\todo{Is this a common notion of action-minimality? For existential queries I would expect that to count the actions actually used by the scheduler.}

We can reduce the problem to the \emph{witness problem} for Rabin (resp. Streett) queries, by lifting the labeling $\labeling$ to the product MDP $\mdp = \mathcal{N} \times \automaton_1 \times \cdots \times \automaton_k$ as $\labeling_\mdp(\tuple{\state, \autoState_1, \ldots, \autoState_k}) = \labeling(\state)$ for all product states $\tuple{\state, \autoState_1, \ldots, \autoState_k}$. We obtain a minimal witness for $\mathcal{N}$ by computing a minimal one for $\mdp$ w.r.t. $\labeling_\mdp$ and ignoring the automata components\todo{this is not technically precise, because actually one needs to check which states are present in the subsystem for the product and then take the corresponding ones for $\mdp$}. Thus, it suffices to consider witnesses for multi-objective Rabin (resp. Streett) queries. Our approach is based on the MEC quotient and a reduction to the witness problem for reachability \cite{baier_certificates_2024}, yielding minimal witnesses for $\universalDQ$-queries and specific labelings. For $\existsCQ$-queries, a more fine-grained approach is necessary and is discussed in \Cref{appendix:milp-exists-queries}.
\todo{We should also say here that $\existsCQ$-queries are also considered, but require a more fine grained approach for minimality.}

\medskip

\noindent\textbf{MEC-quotient approach.}
We have used the certificates for reachability queries from \cite{baier_certificates_2024} for our $\existsCQ$- and $\universalDQ$-certificates (\Cref{def:exists-certificate,def:forall-certificate}). Besides that, they can also be used to find witnesses. Let us fix $\query = \exists \scheduler \centerdot \Conj_{i=1}^k \queryProb_{\triangleright \lambda_i}(\prop_i)$ where each $\prop_i$ is a Rabin property and $\Phi = \forall \scheduler \centerdot \Disj_{i=1}^k \queryProb_{\triangleright \lambda_i}(\bar{\prop}_i)$. Let $\mecQuotient = \mdp_{/ \partition_{\MECS}}^{\mathcal{I}^*}$ and queries $\widehat{\query}$ and $\widehat{\Phi}$ be defined as in \Cref{lemma:index}. It was shown in \cite{baier_certificates_2024} that the vectors
\begin{center}
$\vect{y} \in \mathcal{F}_\triangleright(\mecQuotient, \{\targetSet_i\}_{i \in [k]}, \{\lambda_i\}_{i \in [k]}) \qquad (\vect{x}, \vect{z}) \in \mathcal{H}_\triangleright(\mecQuotient, \{\bar{\targetSet}_i\}_{i \in [k]}, \{\lambda_i\}_{i \in [k]})$
\end{center}
induce witnessing subsystems for $\mecQuotient$ and the reachability queries $\widehat{\existsQuery}$ and $\widehat{\forallQuery}$, respectively. Specifically, the subsystems are induced by the support of the vectors $\vect{y}$ and $\vect{x}$, given by $\stateSupp{\vect{y}} = \{\state \in \states \mid \exists \action \centerdot \vect{y}(\state, \action) > 0 \}$ and $\supp{\vect{x}}$, respectively. For a labeling $\labeling$ of $\mdp$, we define its lifting to $\mecQuotient$ by $\widehat{\labeling}(D) = \Union_{\state \in D} \labeling(\state)$ for all $D \in \partition_\MECS$. Then, the MILPs in \Cref{fig:milps-reachability} can be used to find certificates with a \emph{minimal support} w.r.t. $\widehat{\labeling}$, thereby yielding minimal witnessing subsystems for $\mecQuotient$. In both MILPs, the \emph{binary variable} $\boldBeta$ determines whether a label is contained in the subsystem. In case $\boldBeta(l) = 0$ holds, the second constraint enforces that no state labeled with $l$ may be contained in the subsystem.

\begin{figure}[t]
\begin{minipage}{0.49\textwidth}
\centering
\begin{myframe}{{MILP for $\existsCQ$-queries}}

\smallskip \footnotesize

$\min_{\boldBeta \in \{0, 1\}^\labels} \sum\nolimits_{l \in \labels} \boldBeta(l)$ s.t.
{
\scriptsize
\begin{enumerate}[(1), itemsep=0mm, topsep=0.8mm, parsep=0mm]
\item $\vect{y} \in \mathcal{F}_\triangleright(\mecQuotient, \{\targetSet_i\}_{i \in [k]}, \{\lambda_i\}_{i \in [k]})$
\item $\forall (\state, \action) \in \SA, l \in \widehat{\labeling}(\state)$: $\vect{y}(\state, \action) \leq \boldBeta(l) \cdot M$
\end{enumerate}
}
\end{myframe}
\end{minipage}
\begin{minipage}{0.5\textwidth}
\centering
\begin{myframe}{{MILP for $\universalDQ$-queries}}

\smallskip \footnotesize

$\min_{\boldBeta \in \{0, 1\}^\labels} \sum\nolimits_{l \in \labels} \boldBeta(l)$ s.t.
{
\scriptsize
\begin{enumerate}[(1), itemsep=0mm, topsep=0.8mm, parsep=0mm]
\item $(\vect{x}, \vect{z}) \in \mathcal{H}_\triangleright(\mecQuotient, \{\bar{\targetSet}_i\}_{i \in [k]}, \{\lambda_i\}_{i \in [k]})$
\item $\forall \state \in \states, l \in \widehat{\labeling}(\state)$: $\vect{x}(\state) \leq \boldBeta(l) \cdot M$
\end{enumerate}
}
\end{myframe}
\end{minipage}
\caption{MILPs for minimal witnesses. Here, $M$ is a sufficiently large constant.}
\label{fig:milps-reachability}
\end{figure}

Given a subsystem of $\mecQuotient$ with states $\partition \subseteq \partition_{\MECS}$, the corresponding subsystem $\mdp'$ of $\mdp$ is induced by $\Union_{D \in \partition} D$. For $\universalDQ$-query $\forallQuery$, this corresponding subsystem is minimal w.r.t. $\labeling$. This relation is formally captured below.
\begin{restatable}{theorem}{forallQueriesMinimality}
If $\mecQuotient'$ is a witnessing subsystem for $\widehat{\forallQuery}$ and minimal w.r.t. $\widehat{\labeling}$, then the corresponding subsystem $\mdp'$ is witnessing for $\forallQuery$ and minimal w.r.t. $\labeling$.
\label{theorem:forall-minimal}
\end{restatable}
Notably, \Cref{theorem:forall-minimal} also allows us to obtain minimal witnesses for DTMCs. For $\existsCQ$-query $\existsQuery$, the minimality of the subsystem is generally not preserved. However, for labelings $\labeling$ that are \emph{MEC-equivalent} w.r.t. $\mdp$, i.e. all states in the same MEC are assigned the same set of labels, we can state a similar result.
\begin{restatable}{theorem}{mecQuotientWitness}
Suppose $\labeling$ is an MEC-equivalent labeling w.r.t. $\mdp$. If $\mecQuotient'$ is a witnessing subsystem for $\widehat{\existsQuery}$ and minimal w.r.t. $\widehat{\labeling}$, then the corresponding subsystem $\mdp'$ is witnessing for $\existsQuery$ and minimal w.r.t. $\labeling$.
\label{theorem:mec-equivalent-labelings}
\end{restatable}
\begin{example}
Consider the query $\forallQuery = \forall \scheduler \centerdot \queryProb_{\geq 0.25} (\eventually \globally \neg \state_1) \lor \queryProb_{\geq 0.25}(\globally \eventually \neg \state_2)$ for $\mdp$ (\cref{fig:mdps}). The corresponding reachability query for the quotient $\mecQuotient$ (\cref{subfig:mec-quotient}) is given by $\widehat{\forallQuery} = \forall \scheduler \centerdot \queryProb_{\geq 0.25} (\eventually \{\exit_\emptyset, \exit_{\{2\}}\}) \lor \queryProb_{\geq 0.25}(\eventually \{\exit_\emptyset, \exit_{\{1\}}\})$. Let $\labeling$ be defined by $\labeling(\state_i) = \{\state_i\}$ for all $0 \leq i \leq 4$. The subsystem $\mecQuotient'$ in \Cref{subfig:mec-quotient-subsystem} is a witness for $\widehat{\forallQuery}$ and minimal w.r.t. $\widehat{\labeling}$, yielding a state-minimal subsystem $\mdp'$ for $\mdp$ (\cref{subfig:witness}). 
\end{example}
\noindent\textbf{Minimal witnesses for $\existsCQ$-queries.} Unfortunately, the approach based on the MEC quotient does not always yield minimal witnesses for $\existsCQ$-queries, because witnessing subsystems computed via this approach either \emph{completely} include an MEC or omit it. However, for $\existsCQ$-queries it might be beneficial to \emph{partially} include an MEC \cite{jantsch_certificates_2022}. We provide a MILP encoding for $\existsCQ$-queries in \Cref{appendix:milp-exists-queries} and note that it is currently not supported by our implementation.

\section{Witnesses for Markov Chains via UBAs}
\label{section:dtmc}
For the special case of DTMCs, we now show that \emph{unambiguous Büchi automata} (UBAs) can be used to find minimal witnessing subsystems. This can enable a more practical approach to computing witnesses for LTL properties, since the translation from LTL to UBAs only incurs a single- instead of a double-exponential blowup as for \emph{deterministic} $\omega$-automata. We recall work on DTMCs and UBAs from \cite{baier_markov_2023-1} and, based on this, devise a MILP characterization for finding minimal witnesses.
For the remainder, we fix a (potentially sub-stochastic) DTMC $\dtmc = (\states,\init,\transMat)$ and UBA $\automaton = (\autoStates, \states, \autoInit, \autoTransition, \{(F, \emptyset)\})$. Let $\prop$ denote the property recognized by $\automaton$. Since DTMCs do not contain any nondeterminism, we drop the quantification from queries, e.g. we write $\queryProb_{\geq \lambda}(\prop)$ instead of $\exists \scheduler \centerdot \queryProb_{\geq \lambda}(\prop)$.

\medskip

\noindent\textbf{DTMCs and UBAs.}
The \emph{value vector} of $\dtmc$ and $\automaton$ is $\valueVector_{\dtmc,\automaton} \in [0,1]^{\states \times \autoStates}$, defined by $\valueVector_{\dtmc,\automaton}(\state, \autoState) = \prob_{\dtmc,\state}(\lang{\automaton, \autoState})$ for all $\state \in \states$ and $\autoState \in \autoStates$.
It was shown in \cite{baier_markov_2023-1} that the value vector $\valueVector_{\dtmc \times \automaton}$ can be computed in time polynomial in $\card{\dtmc}$ and $\card{\automaton}$, by analysing the \emph{product of} $\dtmc$ \emph{and} $\automaton$. The product $\vect{B}_{\dtmc \times \automaton} \in [0,1]^{(\states \times \autoStates) \times (\states \times \autoStates)}$ is defined by $\vect{B}_{\dtmc \times \automaton}\left((\state,\autoState),(\state',\autoState')\right) =
  \transMat(\state,\state')$ if $\autoState' \in \autoTransition(\autoState,\state)$ and zero otherwise, for all $\state \in \states$ and $\autoState \in \autoStates$.
Note that $\vect{B}_{\dtmc \times \automaton}$ is not guaranteed to be stochastic if $\automaton$ is unambiguous. In the following, we drop the subscript $\dtmc \times \automaton$.

The \emph{directed graph induced} by $\vect{B}$ consists of vertices $\states \times \autoStates$ and contains an edge $(d, d')$ iff $\vect{B}(d, d') > 0$. Then, $\vect{B}$ is \emph{strongly connected} if its induced graph is strongly connected. For $D \subseteq \states \times \autoStates$, let $\vect{B}_{D}$ be the restriction of $\vect{B}$ to the rows and columns in $D$. Then, $D$ is a \emph{strongly connected component} (SCC) of $\vect{B}$ if $\vect{B}_{D}$ is strongly connected and $D$ inclusion-maximal.
An SCC $D$ of $\vect{B}$ is \emph{recurrent} if the dominant eigenvalue of $\vect{B}_{D}$ is equal to one, and \emph{accepting}, if there is $(\state,\autoState) \in D$ such that $\autoState \in F$. Recurrent SCCs are the counterparts to bottom SCCs in the deterministic setting \cite{baier_markov_2023-1}.
Let $\mathcal{D}^+$ denote the accepting recurrent SCCs of $\vect{B}$, and $\mathcal{D}^-$ denote the non-accepting recurrent SCCs of $\vect{B}$.

The value vector $\valueVector$ satisfies $\vect{B} \valueVector = \valueVector$ \cite[Lemma 4]{baier_markov_2023-1}, but is not the only solution of the system.
To uniquely characterize $\valueVector$, an additional normalizing equation $\boldMu_D$ (called a \emph{$D$-normalizer}) is added for each accepting recurrent SCC $D$ of $\vect{B}$. These normalizers can be computed in polynomial time \cite{baier_markov_2023-1}.
The following set of equations then uniquely characterizes $\valueVector$ \cite[Lemma 12]{baier_markov_2023-1}:
\begin{enumerate*}[label=(\roman*)]
\item $\vect{B} \vect{v} = \vect{v}$,
\item $\boldMu_D^{\top} \vect{v}_D = 1$ for all $D \in \mathcal{D}^+$, and
\item $\vect{v}_D = \vect{0}$ for all $D \in \mathcal{D}^-$
\end{enumerate*}.
\begin{remark}
In principle, we could certify queries in DTMCs via solutions of the system from \cite{baier_markov_2023-1}. This would require certificates for the correctness of the normalizers, which we leave open for future work.
\end{remark}

\medskip

\noindent\textbf{Minimal witnesses via UBAs.}
Let $\prop \subseteq \states^\omega$ denote the property recognized by $\automaton$. W.l.o.g. assume that $\vect{B}$ has no non-accepting recurrent SCCs, since we can remove them beforehand.
For a vector $\vect{x} \in [0,1]^{\states \times \autoStates}$ we define $\stateSupp{\vect{x}} = \{\state \in \states \mid \exists \autoState \in \autoStates \centerdot (\state,\autoState) \in \supp{\vect{x}}\}$.
We claim that any nonnegative solution $\vect{x}$ of 
{
\begin{center}
$\vect{B}  \vect{x} \geq \vect{x}, \qquad \boldMu_D^{\top} \ \vect{x}_D \leq 1 \; \text{ for all } D \in \mathcal{D}^+, \quad \text{ and } \quad \vect{x}(\init,\autoInit) \geq \lambda$
\end{center}
}
\noindent induces a witnessing subsystem for $\queryProb_{\geq \lambda}(\prop)$. Conversely, for every witnessing subsystem for $\queryProb_{\geq \lambda}(\prop)$ with states $\states'$ there is a corresponding solution $\vect{x}$ with $\stateSupp{\vect{x}} \subseteq \states'$. For convenience, we now treat a subsystem $\dtmc'$ with states $\states'$ as a DTMC with states $\states$ where transitions of states in $\states \setminus \states'$ are set to zero. More precisely, the transition matrix is given by $\transMat' \in [0, 1]^{\states \times \states}$ where $\transMat'(\state, \state') = \transMat(\state, \state')$ if $\state, \state' \in \states'$ and $0$ otherwise for all $\state, \state' \in \states$.
\begin{restatable}{proposition}{solToSubsystem}
\label{prop:soltosubsys}
Let $\vect{x} \in [0,1]^{\states \times \autoStates}$ such that $\vect{B} \vect{x} \geq \vect{x}$, $\boldMu_D^{\top} \vect{x}_D \leq 1$ for each $D \in \mathcal{D}^+$, and $\vect{x}(\init,\autoInit) \geq \lambda$, where $\lambda \in [0, 1]$. Let $\states' \subseteq \states$ with $\stateSupp{\vect{x}} \subseteq \states'$. Then, we have $\dtmc_{\states'} \models \queryProb_{\geq \lambda}(\prop)$.
\end{restatable}
\begin{restatable}{proposition}{subsystemToSol}
  \label{prop:subsystosol}
Let $\dtmc'$ be a subsystem of $\dtmc$. Then its value vector $\valueVector' \in [0,1]^{S \times Q}$ satisfies $\vect{B} \valueVector' \geq \valueVector'$ and $\boldMu_D^\top \valueVector'_D \leq 1$ for each $D \in \mathcal{D}^+$.
\end{restatable}
Using \Cref{prop:soltosubsys,prop:subsystosol} we can prove the following relation between witnesses and solutions of the corresponding set of linear inequalities.

\begin{restatable}{theorem}{mcAndUba}
\label{theorem:markov-chain-uba-subsystems}
Let $\states' \subseteq \states$. The following are equivalent:
  \begin{enumerate}[align=left, leftmargin=*, itemsep=0mm, topsep=0.8mm, parsep=0mm]
  \item The induced subsystem $\dtmc_{\states'}$ is a witnessing subsystem for $\queryProb_{\geq \lambda}(\prop)$.
  \item There exists a vector $\vect{x} \in [0,1]^{\states \times \autoStates}$ satisfying $\vect{B} \vect{x} \geq \vect{x}$, $\boldMu_D^{\top} \vect{x}_D \leq 1$ for each $D \in \mathcal{D}^+$ and $\vect{x}(\init,\autoInit) \geq \lambda$ with $\stateSupp{\vect{x}} \subseteq \states'$.
  \end{enumerate}
\end{restatable}

\begin{figure}[t]
\begin{myframe}{{$\mathsf{DtmcUBA}\bigl(\dtmc, \labeling, \{\automaton_i\}_{i \in [k]}, \{\lambda_i\}_{i \in [k]}\bigr)$}}

\smallskip \footnotesize{

$\min \sum\nolimits_{l \in \labels} \boldBeta(\state)$ s.t. $\boldBeta \in \{0, 1\}^\labels$ and $\{\vect{x}_i \in \mathcal{S}_{\dtmc \times \automaton_i}\}_{i \in [k]}$ and
\scriptsize{
\begin{enumerate}[(1), itemsep=0mm, topsep=0.8mm, parsep=0mm]
\item for all $i \in [k]$, $(\state, \autoState) \in \states \times \autoStates_i$ and $l \in \labeling(\state)$: $\vect{x}_i(\state, \autoState) \leq \boldBeta(l)$, \label{milp:dtmc-indicator}
\item and for all $i \in [k]$: $\vect{x}_i(\init,\autoInit) \geq \lambda$.
\end{enumerate}
}
}
\end{myframe}
\caption{MILP for finding minimal witnesses for DTMCs and UBAs.}
\label{milp:dtmc-uba}
\end{figure}
\Cref{theorem:markov-chain-uba-subsystems} enables us to find minimal witnesses for single-objective queries via a MILP-based approach. In fact, we can even handle multi-objective queries $\query =\Conj_{i=1}^k \queryProb_{\geq \lambda_i}(\prop_i)$, by observing that a subsystem $\dtmc'$ is a witnessing subsystem for $\query$ if and only if $\dtmc' \models \queryProb_{\geq \lambda_i}(\prop_i)$ for all $i \in [k]$. For MDPs this generally does not hold.
Using this observation, we now devise a corresponding MILP. Let $\prop_i$ be recognized by UBA $\automaton_i$ for all $i \in [k]$.
For all UBA $\automaton_i$, we define $\mathcal{S}_{\dtmc \times \automaton_i}$ to be the set of vectors $\vect{x} \in [0,1]^\states$ that satisfy $\vect{B}_{\dtmc \times \automaton_i} \vect{x} \geq \vect{x}$ and $\boldMu_D^{\top} \vect{x}_D \leq 1$ for all $D \in \mathcal{D}_{\dtmc \times \automaton_i}^+$. Given a query $\Conj_{i=1}^k \queryProb_{\geq \lambda_i}(\prop_i)$, we can use the MILP in \Cref{milp:dtmc-uba} to determine minimal witnessing subsystems. Intuitively, for each $\vect{x}_i$ we have that $\states' \subseteq \states$ with $\stateSupp{\vect{x}_i} \subseteq \states'$ induces a witnessing subsystem for $\queryProb_{\geq \lambda_i}(\prop_i)$. Hence, all $\states' \subseteq \states$ that satisfy $\Union_{i \in [k]} \stateSupp{\vect{x}_i} \subseteq \states'$, induce a witnessing subsystem for $\query$. For disjunctive queries $\Disj_{i=1}^k \queryProb_{\geq \lambda_i}(\prop_i)$, we can compute the minimal subsystem for each $\queryProb_{\geq \lambda_i}(\prop_i)$ and then return the smallest among them.

For LTL properties $\varphi_i$, the number of states in the corresponding UBAs $\automaton_i$ might be exponential in $\card{\varphi_i}$, but the number of labels is linear in our input problem. Thus, the resulting MILP can be solved in single-exponential time.

\section{Experimental Evaluation}
\label{section:experimental}
We have implemented the computation of certificates and witnesses for multi-objective $\omega$-regular queries in a C\texttt{++} tool based on \textsc{Storm} \cite{hensel_probabilistic_2022}, using \textsc{Gurobi} \cite{gurobi_optimization_llc_gurobi_2024} to optimize MILPs and \textsc{Spot} \cite{duret.22.cav} for translating LTL to $\omega$-automata. Our tool represents the MDPs \emph{exactly}, i.e. transition probabilities are represented as rationals. The certificates for multi-objective reachability underlying our certificates are computed using the approach in \Cref{subsection:certificate-rabin-streett} (via \emph{exact policy iteration}) and the validity of the computed certificates is checked w.r.t. \emph{exact} arithmetic.
We are interested in addressing the following research questions:
{
\setlist[description]{font=\normalfont}
\begin{description}
\item[\textsf{RQ1}:] What is the overhead of certifying the MEC decomposition?
\item[\textsf{RQ2}:] What is the overhead of certifying multi-objective $\omega$-regular queries?
\item[\textsf{RQ3}:] Does the UBA approach offer an advantage?
\end{description}
}
\noindent\textbf{Setup.} We study the overhead of certifying the MEC decomposition (\textsf{RQ1}) by comparing the runtime of the \emph{uncertifying} MEC decomposition in \textsc{Storm} with the \emph{certifying} variant in our tool on $167$ MDP instances from the \emph{Quantitative Verification Benchmark Set} (QVBS) \cite{hartmanns_quantitative_2019}.
We evaluate the overhead of certifying multi-objective $\omega$-regular queries (\textsf{RQ2}) on models and queries from \cite{baier_markov_2023-1,forejt_quantitative_2011,sickert_mochiba_2016}, in total $165$ instances.
\textsc{Prism} supports the verification of multi-objective $\omega$-regular queries, while \textsc{Storm} does not \cite{quatmann_verification_2023}. We first compare the runtime of our tool when computing certificates and when \emph{deactivating} the computation of certificates, allowing us to study the \emph{overhead} of certification, instead of differences in tool performances. The comparison of the uncertifying variant of our tool with \textsc{Prism} can be found in \Cref{appendix:experiments}. For the special case of \emph{single-objective reachability}, we consider $72$ reachability instances from QVBS and compare our tool against the fixed-point certificate implementation from \cite{chatterjee_fixed_2025} with \emph{exact policy iteration} as certification method, which we refer to as \textsc{StormCert}. Lastly, we evaluate the effectiveness of the UBA approach for DTMCs from \Cref{section:dtmc} (\textsf{RQ3}), by comparing it with the approach for MDPs from \Cref{section:witnesses-mdps}.

In the following, runtimes are wallclock times and include the time for parsing and building the model, but exclude the time for exporting the certificates or witnesses. All experiments have been run on a machine with Ubuntu 22.04, an Intel i7-1165G7 CPU and 32GB RAM. Every single execution is given a timeout of 10 minutes. We refer to \cite{baier_2025_15680332} and \Cref{appendix:experiments} for details on the considered instances. The implementation and experimental data are available at \cite{baier_2025_15680332}.
\begin{figure*}[t!]
    \centering
    \begin{subfigure}[t]{0.33\textwidth}
        \centering
        \includegraphics[scale=0.24]{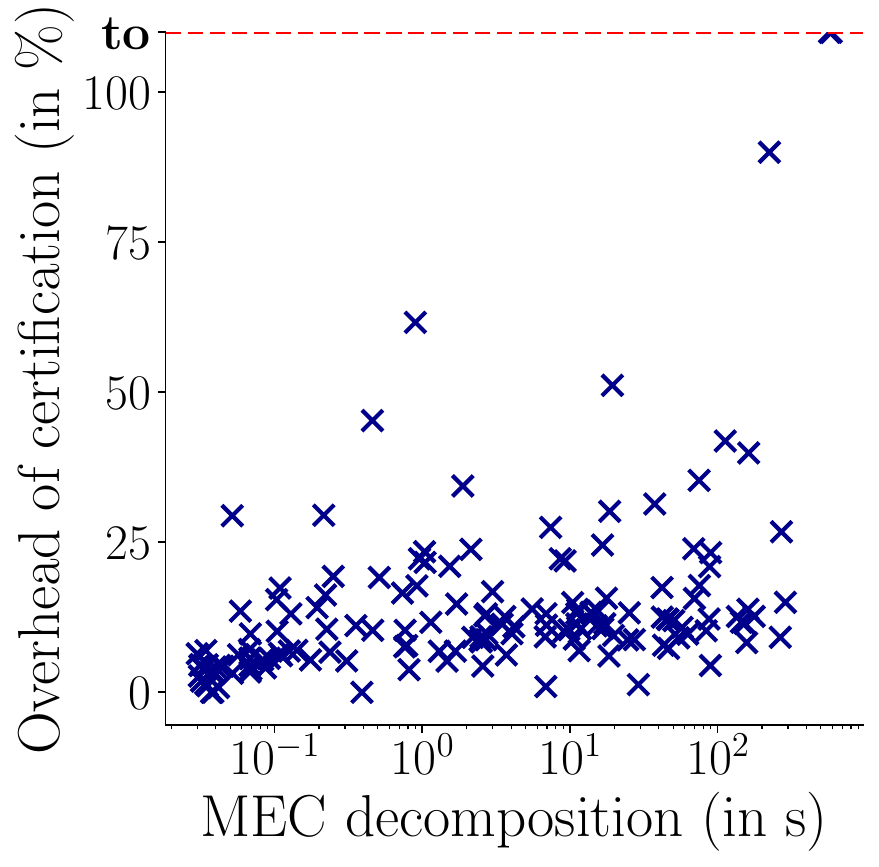}
        \caption{\scriptsize MEC decomposition.}
        \label{subfig:mec-overhead}
    \end{subfigure}%
        \begin{subfigure}[t]{0.33\textwidth}
        \centering
        \includegraphics[scale=0.24]{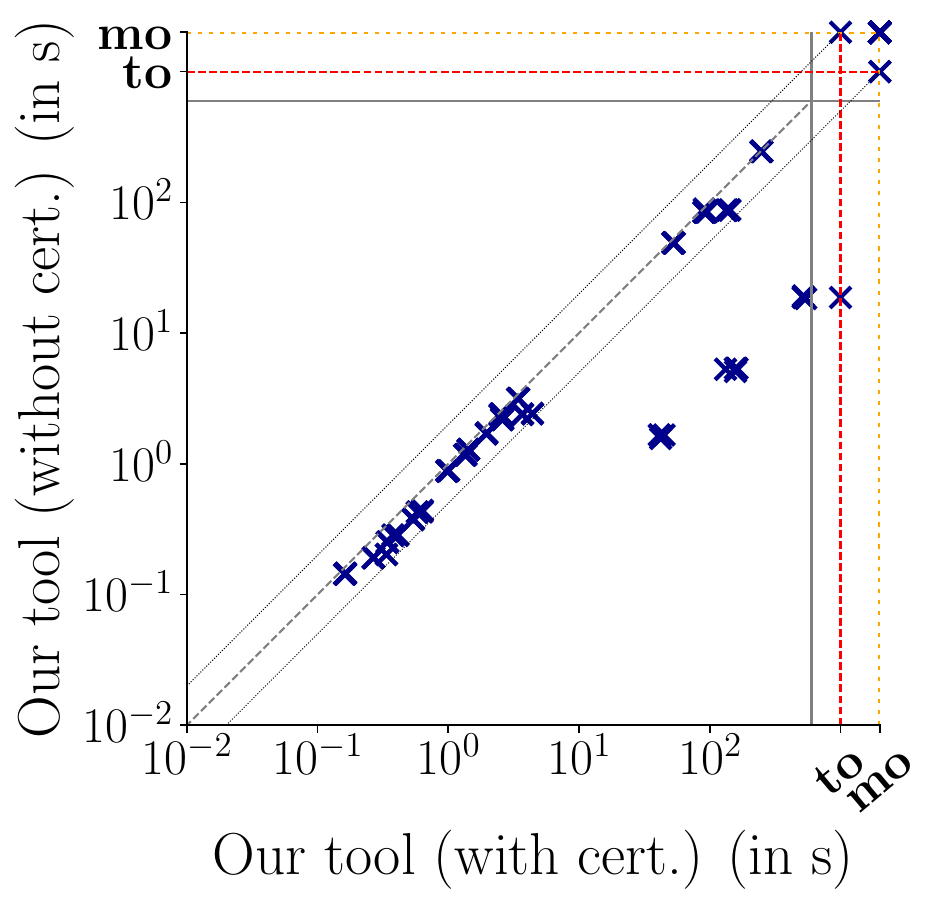}
        \caption{\scriptsize Multi-objective $\omega$-regular.}
        \label{subfig:omega-overhead}
    \end{subfigure}
    \begin{subfigure}[t]{0.33\textwidth}
        \centering
        \includegraphics[scale=0.24]{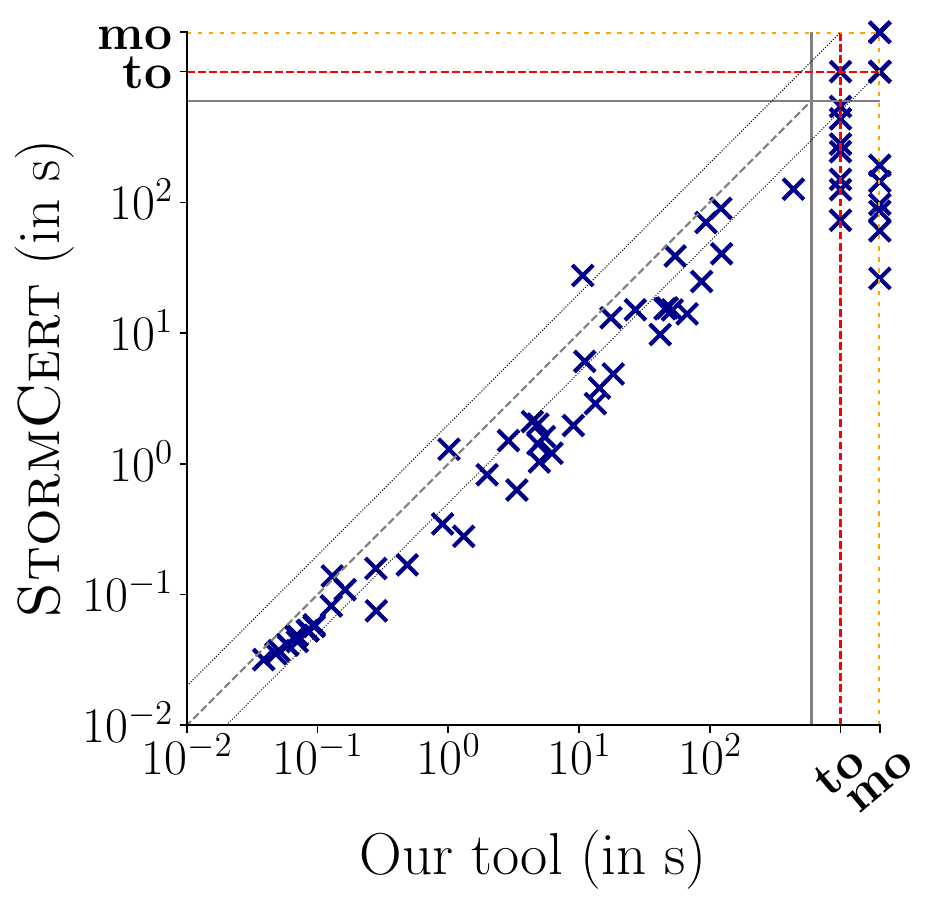}
        \caption{\scriptsize Single-objective reachability.}
        \label{subfig:reach-scatter}
    \end{subfigure}
    \caption{Results for \textsf{RQ1} \& \textsf{RQ2}. In \Cref{subfig:omega-overhead,subfig:reach-scatter} the diagonal corresponds to equal runtimes and the parallel lines to runtimes differing by a factor of $2$.}
\end{figure*}

\medskip

\noindent\textbf{Overhead of MEC certification \textsf{(RQ1)}.} The overhead of certifying the MEC decomposition is shown in \Cref{subfig:mec-overhead}. The horizontal axis shows the runtime of the MEC decomposition (of \textsc{Storm}) and the vertical the relative overhead when providing certificates, i.e. the percentage of the runtime for additionally computing the certificate. Instances for which the MEC decomposition already timed-out are omitted. For roughly $88\%$ of the considered models the overhead is \emph{less than $25\%$}. However, there are outliers where the additional certification almost requires double the amount of time or times out. Upon inspecting these instances, we could observe that they have many states ($\card{\states} > 15 \cdot 10^6$) or many MECs. While further research on difficult instances is needed, we can conclude that for most models from QVBS \emph{the cost of providing certificates is small}.

\medskip

\noindent\textbf{Overhead of query certification \textsf{(RQ2)}.} The runtime of our tool with (resp. without) computation of certificates is shown in \Cref{subfig:omega-overhead}, where points below the diagonal indicate that the certification is slower. For many queries the certification is at most two times slower (points within parallel lines), but there are a couple of outliers for which additionally computing certificates incurs a significant overhead. Recall that the size of our certificates directly depends on the MEC structure in the product MDP and thus also the automaton representing the property. Thus, it seems that certification incurs a reasonable overhead, but can be quite sensitive to the considered model and query. For the specific case of \emph{single-objective} reachability, we compare our approach with the implementation of the fixed-point certificates from \cite{chatterjee_fixed_2025}. The comparison is shown in \Cref{subfig:reach-scatter}. Unsurprisingly, the fixed-point certificates can be computed faster, since they only need to perform policy iteration once. In contrast, we need to perform policy iteration and then additionally compute the $\vect{y}$ certificate. However, fixed-point certificates do not certify \emph{multi-objective reachability} queries. As expected, both approaches always produced valid certificates (due to the exact methods used).

\medskip

\noindent\textbf{Effectiveness of UBA approach \textsf{(RQ3)}.} We consider the DTMC model of \emph{Herman's self stabilising algorithm} from \cite{HERMAN199063,kwiatkowsa_prism_2012}. The protocol involves $N$ processes that are connected via a ring in which tokens are passed around. A stable configuration is reached once there is exactly one token. We consider the model with $N=7$ processes, resulting in a DTMC with $128$ states, and LTL property that ``a stable configuration is eventually reached and $k$ steps before that there are $7$ tokens in the ring'' where $k \in \{3,\ldots,14\}$. It is known that deterministic automata require at least $2^k$ states for this property, while the size of UBAs is linear in $k$ \cite{baier_markov_2023-1}. For each query, we compute state-minimal witnesses via the DRA (\Cref{section:witnesses-mdps}) and UBA approach (\Cref{section:dtmc}). The results are shown in \Cref{table:uba-dra}. For the studied properties, the UBA approach is \emph{always faster}. The UBAs and their products are, for the most part, significantly smaller than their DRA counterparts. \textsc{Spot} needs more time for translating the LTL formula to UBAs and fails to find small UBAs for $k \geq 13$. Overall, this case study demonstrates the usefulness of the UBA approach for finding minimal witnesses.
\begin{table}[!t]
\centering

\caption{Comparison of UBA (\Cref{section:dtmc}) and DRA (\Cref{section:witnesses-mdps}) approach.}
\setlength{\tabcolsep}{3pt}
\def\arraystretch{0.8}%
\scriptsize{

\begin{tabular}{r|cc|cc|cc|cc}
  & \multicolumn{2}{c|}{Total runtime (in s)} & \multicolumn{2}{c|}{\textsc{Spot} runtime (in s)} & \multicolumn{2}{c|}{Automata} & \multicolumn{2}{c}{Product} \\
 $k$ & UBA & DRA & \text{UBA} & \text{DRA} & $\card{\automaton_{\text{UBA}}}$ & $\card{\automaton_\text{DRA}}$ & $\card{\mathcal{C} \times \automaton_\text{UBA}}$ & $\card{\mathcal{C} \times \automaton_\text{DRA}}$ \\
\hline
3 & \textbf{0.241} & 0.365 & 0.026 & 0.008 & 5 & 10 & 358 & 481 \\
4 & \textbf{39.668} & 311.41 & 0.032 & 0.008 & 6 & 18 & 472 & 819 \\
5 & \textbf{39.866} & to & 0.036 & 0.009 & 7 & 34 & 586 & 1269 \\
6 & \textbf{113.597} & 485.13 & 0.039 & 0.011 & 8 & 66 & 700 & 1831 \\
7 & \textbf{328.487} & to & 0.045 & 0.012 & 9 & 130 & 814 & 2505 \\
8 & \textbf{366.616} & to & 0.048 & 0.019 & 10 & 258 & 928 & 3291 \\
9 & to & to & 0.055 & 0.018 & 11 & 514 & 1042 & 4189 \\
10 & \textbf{118.723} & to & 0.067 & 0.022 & 12 & 1026 & 1156 & 5199 \\
11 & \textbf{104.276} & to & 0.095 & 0.037 & 13 & 2050 & 1270 & 6321 \\
12 & \textbf{110.402} & to & 0.178 & 0.063 & 14 & 4098 & 1384 & 7555 \\
13 & to & to & 45.283 & 0.113 & 6146 & 8194 & 14060 & 8901 \\
14 & to & to & 91.649 & 0.219 & 12290 & 16386 & 16662 & 10359
\end{tabular}

}

\label{table:uba-dra}

\end{table}

\section{Conclusion}
We have presented certificates for multi-objective $\omega$-regular queries in MDPs, consisting of improved certificates for the MEC decomposition, new certificates for satisfaction (resp. violation) of Rabin properties in ECs and certificates for reachability queries from \cite{baier_certificates_2024}. Towards explicating verification results, we have shown how to obtain minimal witnesses. For DTMCs and LTL, our approach to witnesses is based on UBAs, resulting in MILPs of single- instead of double-exponential size.
The experimental results for the MEC certificates are promising, showing that certification often only incurs a small overhead. The certification of $\omega$-regular queries seems to be quite sensitive to the properties under consideration, but was often reasonable in our setup. Lastly, we have showcased the benefit of the UBA approach on a case study.

We plan on extending our certificates to more general models, e.g stochastic games. For the implementation, an even tighter integration with \textsc{Storm} would be beneficial. Finally, we want to apply our techniques to more case studies.

\bibliographystyle{splncs04}
\bibliography{ref}

\appendix
\newpage
\section{General Definitions and Concepts}
\label{appendix:subsystems}
\noindent\textbf{Subsystems.} In this work we define subsystems to be sub-stochastic MDPs. This is equivalent to definitions provided in previous works \cite{jantsch_certificates_2022,funke_farkas_2020,baier_certificates_2024}. Let $\mdp = (\states, \actions, \init, \transMat)$ be an MDP and $\mdp' = (\states', \actions, \init, \transMat')$ a sub-stochastic subsystem of $\mdp$. We can obtain a corresponding \emph{stochastic} subsystem as in previous works, by considering the MDP $\mdp_{\bot}' = (\states' \union \{\exit\}, \actions \union \{\tau\}, \init, \transMat_{\exit}')$ where only action $\tau$ is available in $\exit$ and $\transMat_{\exit}'$ is defined by
\[
\transMat_{\exit}'((\state, \action), \state') =
\begin{cases}
\sum_{t \in \states \setminus \states'}\transMat(\state, \action, t) & \text{if } \state \neq \exit \land \state' = \exit\\
\transMat(\state, \action, \state') &  \text{if } \state \neq \exit \land \state' \neq \exit\\
1 & \text{if } \state = \state' = \exit\\
0 & \text{otherwise }
\end{cases}
\]
for all $\state \in \states' \union \{\exit\}$, $\action \in \actions(\state)$ and $\state' \in \states' \union \{ \exit \}$. Throughout this work, we consider properties $\prop \subseteq \states^\omega$. Hence, whenever $\exit$ is reached the property is immediately violated. Due to the one-to-one correspondence of schedulers between $\mdp_{\bot}'$ and $\mdp'$, we see that both definitions coincide.

In the following, we state the monotonicity property of witnessing subsystems for $\omega$-regular queries.

\begin{lemma}[Monotonicity]
\label{lemma:monotonicity-subsystems}
Let $\prop_1, \ldots, \prop_k \subseteq \states^\omega$ be $\omega$-regular properties and $\mdp'$ be a subsystem of $\mdp$. Then the following hold:
\begin{enumerate}[itemsep=0mm, topsep=0.8mm, parsep=0mm]
\item $\mdp' \models\exists \scheduler \centerdot \Conj_{i=1}^k \queryProb_{\geq \lambda_i}(\prop_i) \implies \mdp \models\exists \scheduler \centerdot \Conj_{i=1}^k \queryProb_{\geq \lambda_i}(\prop_i)$
\item $\mdp' \models \forall \scheduler \centerdot \Disj_{i=1}^k \queryProb_{\geq \lambda_i}(\prop_i) \implies \mdp \models \forall \scheduler \centerdot \Disj_{i=1}^k \queryProb_{\geq \lambda_i}(\prop_i)$
\end{enumerate}
\end{lemma}
\begin{proof}
We use the observation that $\pathsFin(\mdp') \subseteq \pathsFin(\mdp)$ from the proof of \cite[Proposition 4.4]{jantsch_certificates_2022} and the fact that paths in $\pathsFin(\mdp')$ have the same probabilities attached to them in $\mdp'$ and $\mdp$. We begin with the first statement. For a given scheduler $\scheduler'$ for $\mdp'$, we can construct a corresponding scheduler $\scheduler$ for $\mdp$ as follows. For all $\Path \in \pathsFin(\mdp)$, we define $\scheduler(\Path) \coloneqq \scheduler(\Path)$ if $\Path \in \pathsFin(\mdp')$ and $\scheduler(\Path) \coloneqq \boldMu$ for some $\boldMu \in \distr{\actions(\last{\Path})}$ otherwise. Together with the observation above, we can conclude that $\prob_{\mdp'}^{\scheduler'}(\prop_i) \leq \prob_{\mdp}^\scheduler(\prop_i)$ for all $i \in [k]$. The second statement can be shown via contraposition and the arguments are analogous. \qed
\end{proof}
Next, we show that witnessing subsystems of the product can be used to obtain witnessing subsystems for the original MDP.
\begin{lemma}[Subsystems of product]
Let $\mdp = (\states, \actions, \init, \transMat)$ be an MDP, $\query$ be a multi-objective query and $\mathcal{N}$ be the product of $\mdp$ and deterministic $\omega$-automata $\automaton_1, \ldots, \automaton_\ell$. Further, let $N'$ be a subset of states of $\mathcal{N}$ and $\states' = \{\state \in \states \mid \exists \autoState_1, \ldots, \autoState_\ell \centerdot \tuple{\state, \autoState_1, \ldots, \autoState_\ell} \in N'\}$. Then, $\mathcal{N}_{N'} \models \query$ implies $\mdp_{\states'} \models \query$.
\end{lemma}
\begin{proof}
We only need to observe that $\{\Path \restriction_{\mdp} \mid \Path \in \pathsFin(\mathcal{N}_{N'}) \} \subseteq \pathsFin(\mdp_{\states'})$ and that paths present in $\mathcal{N}_{N'}$ and $\mdp_{\states'}$ have the same probabilities attached to them. The remainder of the proof is then analogous to the proof of \Cref{lemma:monotonicity-subsystems}.
\end{proof}

\section{Supplementary Material for \Cref{subsection:satisfying-ecs}}
Our certificates for the MEC decomposition build on certificates for qualitative reachability from \cite{chatterjee_fixed_2025}.
\begin{lemma}
Let $\mdp = (\states, \actions, \init, \transMat)$ be an MDP and $\targets \subseteq \states$. We have $\prob_{\state}^{\min}(\eventually \targets) > 0$ for all $\state \in \states$ if and only if there exists $\mathtt{r} \in \naturals^{\states}$ such that $\mathtt{r}(f) = 0$ for all $f \in \targets$ and
\[
1 + \min_{\state' \in \supp{\transMat(\state, \action)}} \mathtt{r}(\state') \leq \mathtt{r}(\state)
\]
for all $\state \in \states \setminus \targets$ and $\action \in \actions(\state)$.
\label{lemma:fixed-point-qualitative}
\end{lemma}
\begin{proof}
Directly follows from \cite[Lemma 2]{chatterjee_fixed_2025} and \cite[Proposition 1]{chatterjee_fixed_2025}.
\end{proof}

\mecCertificate*
\begin{proof}
Let $\mdp_{/ \mathcal{D}} = (\partition \union \{\exit\}, \actions \union \{\tau\}, \mecClass[\partition]{\init}, \transMat')$ be the quotient with respect to the partition $\partition$ where 
\begin{itemize}[align=left, leftmargin=*, itemsep=0mm, topsep=0.8mm, parsep=0mm]
\item $\transMat'(\mecClass[\partition]{\state}, \action, \mecClass[\partition]{\state'}) = \transMat(\state, \action, \mecClass[\partition]{\state'})$ for all $(\state, \action) \in \SA_\partition$ and $\state' \in \states$,
\item $\transMat'(\mecClass[\partition]{\state}, \tau, \exit) = 1$ for all non-trivial $D \subseteq \partition$, and $\exit$ is absorbing.
\end{itemize}

To show the correctness of the MEC certificates, we consider the characterization of the MEC decomposition from \cite[Proposition~3.26]{jantsch_certificates_2022}. Specifically, $\partition$ corresponds to the MEC partition if and only if
\begin{itemize}
\item each $D \in \partition$ is strongly connected and
\item the quotient w.r.t. $\partition$ does not contain ECs (except for the added absorbing states).
\end{itemize}
Since the first condition is equivalent to the existence of EC certificates, we only need to show that the existence of $\mathtt{r}$ is equivalent to the absence of ECs in the quotient $\mdp_{/ \mathcal{D}}$, i.e. its only EC is formed by $\exit$.
This is equivalent to showing that in $\mdp_{/ \mathcal{D}}$ the minimal probability of reaching $\exit$ is positive from all states. We first show that from $\mathtt{r}$ we can construct a $\mathtt{r}'$ that satisfies the conditions of \Cref{lemma:fixed-point-qualitative} w.r.t. $\mdp_\partition$. Conversely, we show that a vector $\mathtt{r}'$ of \Cref{lemma:fixed-point-qualitative} satisfies the conditions imposed on $\mathtt{r}$.
\begin{itemize}
\item[$\Rightarrow$:] By assumption, for all $(\state, \action) \in \SA_{\partition}$ we have
\[
1 + \min_{\state' \in \supp{\transMat(\state, \action)}} \mathtt{r}(\mecClass[\partition]{\state'}) \leq \mathtt{r}(\mecClass[\partition]{\state}).
\]
Consider $\mathtt{r}' \in \naturals^{\partition \union \{\exit\}}$ defined by $\mathtt{r}'(D) = \mathtt{r}(D) + 1$ for all $D \in \partition$ and $\mathtt{r}'(\exit) = 0$. For all $(\state, \action) \in \SA_\partition$ we have
\begin{align*}
\mathtt{r}'(\mecClass[\partition]{\state}) = \mathtt{r}(\mecClass[\partition]{\state}) + 1 &\geq (\min_{\state' \in \supp{\transMat(\state, \action)}} \mathtt{r}(\mecClass[\partition]{\state'}) + 1) + 1\\
&= \min_{\state' \supp{\transMat(\state, \action)}} \mathtt{r}'(\mecClass[\partition]{\state'}) + 1
\end{align*} 
By construction, we have $\mathtt{r}'(D) \geq 1$ for all $D \in \partition$ and thus also $\mathtt{r}'(D) \geq 1 + \mathtt{r}'(\exit) = 1$ for all non-trivial $D$. By \Cref{lemma:fixed-point-qualitative} this implies that $\prob_{D}^{\min}(\eventually \exit) > 0$ for all $D \in \partition$. 
\item[$\Leftarrow$:] Suppose the minimal probability of reaching $\exit$ is positive from every state in $\mdp_\partition$. Then, by \Cref{lemma:fixed-point-qualitative} there exists a corresponding $\mathtt{r}'$ that satisfies
\[
1 + \min_{D' \in \supp{\transMat'(D, \action)}} \mathtt{r}'(D') \leq \mathtt{r}'(D)
\]
for all $D \in \partition$ and $(\state, \action) \in \SA_\partition$ with $D = \mecClass[\partition]{\state}$. However, this implies that the conditions on the MEC certificate are satisfied.
\end{itemize}
\qed 
\end{proof}

\absenceCertificate*
\begin{proof}~
\begin{itemize}
\item[$\Rightarrow$:] Suppose there does not exist an EC $\ec$ with $\Conj_{i=1}^k \ec \modelsRabin \pairs_i$. Then for all $(F_1, E_1) \in \pairs_1, \ldots, (F_k, E_k) \in \pairs_k$ we have that the sub-MDP formed by $\bigcap_{i=1}^k E_i$ does not contain an EC $\ec$ such that $\states(\ec) \intersection F_i \neq \emptyset$ for all $i \in [k]$. Otherwise there would be an EC satisfying all Rabin properties. For each of these sub-MDPs we can then determine a corresponding MEC certificate.
\item[$\Leftarrow$:] Suppose there exists a absence certificate and for the sake of contradiction assume that $\mdp$ has an EC $\ec$ such that $\Conj_{i=1}^k \ec \modelsRabin \pairs_i$. Let $(F_1, E_1) \in \pairs_1, \ldots, (F_k, E_k) \in \pairs_k$ denote the Rabin pairs that are satisfied by $\ec$, that is $\states(\ec) \subseteq \bigcap_{i=1}^k E_i$ and $\states(\ec) \intersection F_i \neq \emptyset$ for all $i \in [k]$. Since $\mathtt{Abs}$ is a absence certificate, there exists a \emph{MEC certificate} $(\mathtt{EC}, \mathtt{r}, \partition)$ for the sub-MDP formed by $\bigcap_{i=1}^k E_i$. Further, there exists $D \in \partition$ such that $\states(\ec) \subseteq D$ and we know that $D \intersection F_i = \emptyset$ for some $i \in [k]$. However, this contradicts the assumption that $\ec$ satisfies all properties. Hence, there does not exist a satisfying EC in $\mdp$.
\end{itemize}
\qed
\end{proof}

\hardness*
\begin{proof}
\textsf{NP} membership is clear, since given an EC, we can check for all Rabin properties $\pairs_i$ whether there exists a Rabin pair $(F, E) \in \pairs_i$ such that $\states(\ec) \intersection F \neq \emptyset$ and $\states(\ec) \subseteq E$.
\textsf{NP}-hardness follows via reduction from SAT. For a CNF formula 
\[
\varphi = C_1 \land \ldots \land C_m \text{ over variables } X = \{x_0, \ldots, x_n\},
\] 
we consider the MDP $\mdp$ with states $\states = X \union \{\init\}$. Each state $x \in X$ has a transition to state $\init$ with probability $1$ and $\init$ has a transition to all states in $X$ and itself with probability $1$. Let $\pairs_i = \{ (\{x \}, \states) \mid x \in C_i \} \union \{ (\states, \states \setminus \{x\} ) \mid \neg x \in C_i \}$ for all $i \in [m]$. Note that the construction of $\mdp$ and $\pairs_1, \ldots, \pairs_m$ only requires polynomial time. Intuitively, $\pairs_i$ encodes the clause $C_i$. Whenever a variable $x$ appears in $C_i$ and is not negated, we add a Rabin pair that requires an EC to contain it. If $x$ appears negated in $C_i$, then a Rabin pair expresses that an EC should not contain it.
\begin{enumerate}
\item[$\Rightarrow$:] Suppose there exists an EC $\ec$ that satisfies all Rabin properties $\pairs_1, \ldots, \pairs_m$. Consider the assignment $\beta \colon X \to \{0, 1\}$ that sets all variables $x \in \states(\ec) \intersection X$ to $1$ and all variables $x \in (\states \setminus \states(\ec)) \intersection X$ to $0$. We claim that $\beta$ is a satisfying assignment for $\varphi$. For each clause $C_i$, we know that there exists a pair $(F, E) \in \pairs_i$ such that $\states(\ec) \intersection F \neq \emptyset$ and $\states(\ec) \subseteq E$. Further, $(F, E)$ is either of the form $(\{x\}, \states)$ or $(\states, \states \setminus \{x\})$. If the former holds, then we know $x \in \states(\ec)$, $x \in C_i$ and $\beta(x) = 1$, hence the clause $C_i$ is satisfied by $\beta$. If the latter holds, then $x \notin \states(\ec)$, $\neg x \in C_i$ and $\beta(x) = 0$ and again the clause is satisfied.
\item[$\Leftarrow$:] Now suppose $\beta$ is a satisfying assignment for $\varphi$ and let $\states' = \{x \in X \mid \beta(x) = 1\}$. We then consider the EC $\ec$ formed by $\states' \union \{\init\}$ (which exists by construction of $\mdp$). For each clause $C_i$ there exists $x \in C_i$ with $\beta(x) = 1$ or there exists $\neg x \in C_i$ with $\beta(x) = 0$. If the former holds, then $x \in \states(\ec)$ and the Rabin pair $(\{x\}, \states) \in \pairs_i$ is satisfied. If the latter holds, then $x \notin \states(\ec)$ and the Rabin pair $(\states, \states \setminus \{x\}) \in \pairs_i$ is also satisfied.
\end{enumerate}
\qed
\end{proof}

\section{Supplementary Material for \Cref{subsection:certificate-rabin-streett}}
\label{appendix:certificates}
\subsection{Omitted Proofs}
\theoremReduction*
\begin{proof}
The first statement follows from \cite{forejt_quantitative_2011} and \cite[Lemma~2]{randour_percentile_2015}. We show the second statement for $\triangleright = \ \geq$. The proof for $\triangleright = \ >$ is completely analogous.
\begin{align*}
\mdp \models \forall \scheduler \centerdot \Disj_{i=1}^k \queryProb_{\geq \lambda_i}(\bar{\prop}_i) &\iff \mdp \not\models\exists \scheduler \centerdot \Conj_{i=1}^k \queryProb_{< \lambda_i}(\bar{\prop}_i) \\
&\iff \mdp \not\models\exists \scheduler \centerdot \Conj_{i=1}^k \queryProb_{> 1- \lambda_i}(\prop_i) \\
&\iff  \mdp_{/ \partition_{\MECS}}^{\mathcal{I}^*} \not\models \exists \scheduler \centerdot \Conj_{i=1}^k \queryProb_{> 1- \lambda_i}(\targetSet_i)\\
&\iff \mdp_{/ \partition_{\MECS}}^{\mathcal{I}^*} \models  \forall \scheduler \centerdot \Disj_{i=1}^k \queryProb_{\leq 1 - \lambda_i}(\targetSet_i)\\
&\iff \mdp_{/ \partition_{\MECS}}^{\mathcal{I}^*} \models  \forall \scheduler \centerdot \Disj_{i=1}^k \queryProb_{\geq  \lambda_i}(\bar{\targetSet}_i)
\end{align*}
The last equivalence follows from the fact that in $\mdp_{/ \partition_{\MECS}}^{\mathcal{I}^*}$ the absorbing states $\{\bot_I\}_{I \subseteq [k]}$ are reached almost surely, regardless of the choice of the scheduler. \qed
\end{proof}

\begin{lemma}
Let $\widehat{\query} = \exists \scheduler \centerdot \Conj_{i=1}^k \queryProb_{\triangleright \lambda_i}(\eventually G_i)$ and $\widehat{\Phi} = \forall \scheduler \centerdot \Disj_{i=1}^k \queryProb_{\triangleright \lambda_i}(\eventually \bar{G}_i)$ where $G_i = \{\exit_I \mid I \subseteq [k], i \in I\}$ and $\bar{G}_i = \{\exit_I \mid I \subseteq [k], i \notin I\}$ for all $i \in [k]$.
\begin{enumerate}[1), itemsep=0mm, topsep=0.8mm, parsep=0mm]
\item If $\mathcal{I} \subseteq \mathcal{I}^*$ and $\mdp_{/ \partition_{\MECS}}^{\mathcal{I}} \models \widehat{\query}$, then $\mdp_{/ \partition_{\MECS}}^{\mathcal{I}^*} \models \widehat{\query}$.
\item If $\mathcal{I} \supseteq \mathcal{I}^*$ and $\mdp_{/ \partition_{\MECS}}^{\mathcal{I}} \models \widehat{\Phi}$, then $\mdp_{/ \partition_{\MECS}}^{\mathcal{I}^*} \models \widehat{\Phi}$.
\end{enumerate}
\end{lemma}
\begin{proof}
For the first statement, we simply observe that a scheduler satisfying the query $\widehat{\query}$ in $\mdp_{/ \partition_{\MECS}}^{\mathcal{I}}$ is also a satisfying $\widehat{\query}$ in $\mdp_{/ \partition_{\MECS}}^{\mathcal{I}^*}$. The reason is that both MDPs have the same states and, due to $\mathcal{I} \subseteq \mathcal{I}^*$, any action that is available in a state in $\mdp_{/ \partition_{\MECS}}^{\mathcal{I}}$ is also available the same state in $\mdp_{/ \partition_{\MECS}}^{\mathcal{I}^*}$.

For the second statement, we prove via contraposition. Suppose we have $\mdp_{/ \partition_{\MECS}}^{\mathcal{I}^*} \models \exists \scheduler \centerdot \Conj_{i=1}^k \queryProb_{< \lambda_i}(\eventually \bar{G}_i)$. This is equivalent to 
\[
\mdp_{/ \partition_{\MECS}}^{\mathcal{I}^*} \models \exists \scheduler \centerdot \Conj_{i=1}^k \queryProb_{> 1-\lambda_i}(\eventually G_i).
\]
Since now we have $\mathcal{I} \supseteq \mathcal{I}^*$, by same reasoning as before, there exists a satisfying scheduler for $\mdp_{/ \partition_{\MECS}}^{\mathcal{I}}$ as well.
\qed
\end{proof}

\subsubsection{Certificates for Rabin and Streett Queries}

\soundAndComplete*
\begin{proof}
Let us begin with the first statement.
\begin{itemize}
\item[$\Rightarrow$:] From \Cref{lemma:index}, we know that $\mecQuotient \models \widehat{\query}$. We set $\partition = \partition_\MECS$ and $\mathcal{I} = \mathcal{I}^*$. By \Cref{lemma:mec-certificate} there exists a MEC certificate $\mathtt{MEC}$. Further, whenever we have $(D, I) \in \mathcal{I}^*$, we know by construction that there exists an $\ec$ that satisfies all Rabin properties indexed by $I$. From \Cref{lemma:ec-certificate} we know that there exists valid EC-certificate for $\states(\ec)$. The existence of $\vect{y}$ follows from \Cref{lemma:reach-certs}.
\item[$\Leftarrow$:] From \Cref{lemma:mec-certificate} we know that the existence of a valid MEC certificate implies $\partition = \partition_{\MECS}$. From the existence of $\vect{y}$ in the last condition, we know that $\mdp_{/ \partition}^\mathcal{I} \models \exists \scheduler \centerdot \Conj_{i=1}^k \queryProb_{\geq \lambda_i}(\eventually G_i)$. Further, for all $(D, I) \in \mathcal{I}$ we know that the MEC formed by $D$ contains some EC in which all Rabin properties indexed by $I$ are satisfied. It follows that $\mathcal{I} \subseteq \mathcal{I}^*$. Altogether, this implies $\mecQuotient \models \exists \scheduler \centerdot \Conj_{i=1}^k \queryProb_{\geq \lambda_i}(\eventually G_i)$. The statement then follows from \Cref{lemma:index}.
\end{itemize}
Now we proof the second statement.
\begin{itemize}
\item[$\Rightarrow$:] From \Cref{lemma:index}, we know that $\mecQuotient \models \widehat{\query}$. Again, we set $\partition = \partition_\MECS$ and $\mathcal{I} = \mathcal{I}^*$. By \Cref{lemma:mec-certificate} we know there exists a MEC certificate. For all $D \subseteq \states$ and $I \subseteq [k]$ with $(D, I) \notin \mathcal{I}^*$ we know that $D$ \emph{does not contain} an EC that satisfies all Rabin properties indexed by $I$. From \Cref{proposition:absence-certificate}, we know there exists an absence certificate for the sub-MDP $\mdp[D]$ and $\{\pairs_i\}_{i \in I}$. The existence of $(\vect{x}, \vect{y})$ again follows from \Cref{lemma:reach-certs}.
\item[$\Leftarrow$:] From \Cref{lemma:mec-certificate} we know that $\partition = \partition_{\MECS}$. From the existence of $(\vect{x}, \vect{z})$, we have that $\mdp_{/ \partition}^\mathcal{I} \models \forall \scheduler \centerdot \Disj_{i=1}^k \queryProb_{\geq \lambda_i}(\eventually \bar{G}_i)$. Further, for all $D \subseteq \states$ and $I \subseteq [k]$ with $(D, I) \notin Z$ we know that the existence of an absence certificate for $\mdp[D]$ and $\{\pairs_i\}_{i \in I}$ implies that $D$ not contain an EC that satisfies the Rabin properties $\{ \pairs_i \}_{i  \in I}$. Thus we have $\mathcal{I}^* \subseteq \mathcal{I}$. Altogether, we get $\mecQuotient \models \forall \scheduler \centerdot \Disj_{i=1}^k \queryProb_{\geq \lambda_i}(\eventually \bar{G}_i)$. Applying \Cref{lemma:index} then complete the proof.
\end{itemize}
\end{proof}

\subsection{Fixed-point Certificates and Farkas Certificates for Single Objective Reachability}
\label{appendix:fixed-point-farkas-comparison}
In this section, we compare the Farkas certificates originally introduced in \cite{funke_farkas_2020} and the more recently proposed fixed-point certificates \cite{chatterjee_fixed_2025}. Both types of certificates certify \emph{single} reachability probabilities, although in slightly different settings. While fixed point certificates can be seen as probability bounds for \emph{all states}, Farkas cerificates certify that a probability constraint \emph{in the initial state} is satisfied for a \emph{given bound}. 

The remainder of this section is structured as follows. We briefly recall both types of certificates and how they can be computed. Afterwards we highlight the differences and then provide a modified version of Farkas certificates with which the reachability probabilities in \emph{all states} can be certified. We fix an MDP $\mdp = (\states \union \{\targetSet\}, \actions, \init, \transMat)$ where $\targetSet$ and $\states$ are disjoint and $\targetSet$ is the set of designated target states. W.l.o.g. we may assume that $\targetSet$ is absorbing. Further, in the following we exclude state-action pairs of $\targetSet$ from $\SA$. We write $\prob_{\state}^{\max}(\eventually \targetSet)$ and $\prob_{\state}^{\min}(\eventually \targetSet)$ to denote the maximal and minimal probabilities of reaching $\targetSet$ from state $\state \in \states$, respectively.

\medskip

\noindent\textbf{Fixed-point certificates.}
Fixed-point certificates are essentially \emph{bounds} for the \emph{optimal} reachability probabilities. More precisely, \emph{valid} fixed-point certificates are lower (resp. upper) bounds $\vect{x} \in [0,1]^\states$ for the vector $\vect{x}_{\opt} = (\prob_{\state}^{\opt}(\eventually \targetSet))_{\state \in \states}$ where $\opt \in \{\max, \min\}$. To check whether such a vector $\vect{x}$ is indeed a valid bound, the \emph{Bellman operator} $\mathcal{B}^{\opt} \colon [0,1]^\states \to [0,1]^\states$ defined by
\[
\mathcal{B}^{\opt}(\vect{x})(\state) = 
\begin{cases}
1 & \text{if } \state \in \targetSet\\
\opt_{\action \actions(\state)} \sum_{\state' \in \states} \transMat(\state, \action, \state') \cdot \vect{x}(\state') & \text{otherwise}
\end{cases}
\] 
is used. Since the optimal reachability probabilities are the \emph{least fixed point} of $\mathcal{B}^{\opt}$, to certify that a vector $\vect{x} \in [0,1]^\states$ is an \emph{upper bound} for $\vect{x}_{\opt}$, i.e. $\vect{x}_{\opt} \leq \vect{x}$, it suffices to check that $\vect{x}$ is a \emph{pre fixed point} of $\mathcal{B}^{\opt}$, i.e. $\mathcal{B}^{\opt}(\vect{x}) \leq \vect{x}$. 
For \emph{lower bounds}, i.e. $\vect{x} \leq \vect{x}_{\opt}$, it is checked whether $\vect{x}$ is \emph{post fixed point}, i.e. $\vect{x} \leq \mathcal{B}^{\opt}(\vect{x})$. 
However, for lower bounds this does not suffice (note that $\vect{1}$ is always a trivial post-fixed point). Essentially, for $\min$-probabilities it needs to be additionally checked that $\vect{x}(\state) > 0 \implies \prob_{\state}^{\min}(\eventually \targetSet) > 0$ for each $\state \in \states$. For $\max$-probabilities this is analogous, but an (implicit) witness strategy needs to be included in addition. To certify these conditions, the authors present certificates for \emph{qualitative reachability}, i.e. certificates for $\prob_\state^{\opt}(\eventually \targetSet) > 0$. We omit their definitions here and refer to \cite{chatterjee_fixed_2025} for details.

\smallskip

\noindent\textit{Computation of fixed-point certificates.} In their implementation, the authors compute \emph{both} a lower and upper bound for the optimal reachability probabilities. The difference between the lower and upper bound be specified by the user.
To compute these certificates, they propose different methods, most prominently \emph{interval iteration}. While the MDP and returned certificate are represented exactly (i.e. numerical values are represented as rational numbers), the computations are performed using floating point arithmetic. As noted by the authors \cite{chatterjee_fixed_2025}, imprecise arithmetic usually breaks the inductivity of the certificates, so \emph{safe rounding} and \emph{smoothing} are proposed to mitigate these problems. We refer to \cite{chatterjee_fixed_2025} for details. 

\medskip

\noindent\textbf{Farkas certificates.} Farkas certificates \cite{funke_farkas_2020,jantsch_certificates_2022} are derived from known LP characterizations for optimal reachability and certify the \emph{constraints on optimal reachability probabilities}, i.e. they certify constraints of the form $\prob_{\init}^{\opt}(\eventually \targetSet) \bowtie \lambda$ where $\opt \in \{\max, \min\}$, $\bowtie \ \in \{\leq, <, \geq, >\}$ and $\lambda \in [0,1]$. While subtle, the difference compared to fixed-point certificates is that Farkas certificates only certify the probability in the initial state and for a fixed bound $\lambda$. In the following, let $\triangleright \in \{\geq, >\}$ and $\triangleleft \in \{\leq, < \}$.

\smallskip

\noindent\textit{Farkas certificates for maximal probabilities.} A certificate for $\prob_{\init}^{\max}(\eventually \targetSet) \triangleright \lambda$ is a vector $\vect{y} \in \realsnn^\SA$ that satisfies
\begin{enumerate}[(1)]
\item for all $\state \in \states$: $\initDistr(\state) + \sum_{(\state', \action') \in \SA} \transMat(\state', \action') \cdot \vect{y}(\state, \action) \geq \sum_{\action \in \actions(\state)} \vect{y}(\state, \action)$,
\item $\sum_{(\state, \action) \in \SA} \transMat(\state, \action, \targetSet) \cdot \vect{y}(\state, \action) \triangleright \lambda$.
\end{enumerate}
A certificate for $\prob_{\init}^{\max}(\eventually \targetSet) \triangleleft \lambda$ is a vector $\vect{x} \in \realsnn^\states$ that satisfies
\begin{enumerate}[(1)]
\item for all $(\state, \action) \in \SA$: $\vect{x}(\state) \geq \transMat(\state, \action, \targetSet) + \sum_{\state' \in \states} \transMat(\state, \action, \state') \cdot \vect{x}(\state')$,
\item $\vect{x}(\init) \triangleleft \lambda$.
\end{enumerate}
Note that the certificates for upper bounds are essentially the same as for the fixed-point certificates, except for the addition of the constraints on the bound.

\smallskip

\noindent\textit{Farkas certificates for minimal probabilities.} Farkas certificates for minimal probabilities are defined for \emph{EC-free} MDPs, that is, for MDPs whose ECs are all formed by states in $\targetSet$. Thus, in the following we assume that all ECs are formed by states in $\targetSet$. A certificate for $\prob_{\init}^{\min}(\eventually \targetSet) \triangleleft \lambda$ is a vector $\vect{y} \in \realsnn^\SA$ that satisfies
\begin{enumerate}[(1)]
\item for all $\state \in \states$: $\initDistr(\state) + \sum_{(\state', \action') \in \SA} \transMat(\state', \action') \cdot \vect{y}(\state, \action) \geq \sum_{\action \in \actions(\state)} \vect{y}(\state, \action)$,
\item $\sum_{(\state, \action) \in \SA} \transMat(\state, \action, \targetSet) \cdot \vect{y}(\state, \action) \triangleleft \lambda$.
\end{enumerate}
A certificate for $\prob_{\init}^{\min}(\eventually \targetSet) \triangleright \lambda$ is a vector $\vect{x} \in \realsnn^\states$ that satisfies
\begin{enumerate}[(1)]
\item for all $(\state, \action) \in \SA$: $\vect{x}(\state) \leq \transMat(\state, \action, \targetSet) + \sum_{\state' \in \states} \transMat(\state, \action, \state') \cdot \vect{x}(\state')$,
\item $\vect{x}(\init) \triangleright \lambda$.
\end{enumerate}
Now the lower bounds essentially are in correspondence with the lower bounds for fixed-point certificates. To obtain an EC-free MDP, the \emph{MEC quotient}, as discussed in the main body of this work, can be constructed. This construction can be certified with the certificates for the MEC decomposition proposed in \cite{jantsch_certificates_2022} and which has been adapted in this work.

\smallskip

\noindent\textit{Computation of Farkas certificates.} Computing Farkas certificates amounts to finding solutions of systems of linear inequalities, which can be done via LPs. Alternatively, the computation via value iteration and policy iteration was proposed in \cite{jantsch_certificates_2022} but not practically considered thus far.

\medskip

\noindent\textbf{Modified Farkas certificates.} We now show that Farkas certificates can be adapted, so that they also certify the probabilities \emph{in all states}. To this end, we revisit the primal and dual formulation of the LPs for computing optimal reachability probabilities \cite{baier_principles_2008-1}, from which Farkas certificates are derived.
\smallskip

\noindent\textit{Modified Farkas certificates for maximal probabilities.} Let us consider the LP formulations for maximal reachability probabilities first.
\begin{myframe}{{$\mathsf{LP}_{\mathsf{primal}}^{\max}$}}

\medskip

$\min \sum_{\state \in \states} \vect{x}(\state)$ where $\vect{x} \in \realsnn^\states$ subject to
\begin{enumerate}[(1), itemsep=0mm, topsep=0.8mm, parsep=0mm]
\item for all $(\state, \action) \in \SA$: $\vect{x}(\state) \geq \sum_{\state' \in \states} \transMat(\state, \action, \state') \cdot \vect{x}(\state')$
\item for all $\state \in \targetSet$: $\vect{x}(\state) = 1$
\end{enumerate}
\end{myframe}

\begin{myframe}{{$\mathsf{LP}_{\mathsf{dual}}^{\max}$}}

\medskip

$\max \sum_{\state \in \SA} \transMat(\state, \action, \targetSet) \cdot \vect{y}(\state, \action)$ where $\vect{y} \in \realsnn^\SA$ and $\vect{x} \in \reals^\states$ subject to
\begin{enumerate}[(1), itemsep=0mm, topsep=0.8mm, parsep=0mm]
\item for all $\state \in \states$: $1 + \sum_{(\state', \action') \in \SA} \transMat(\state', \action') \cdot \vect{y}(\state, \action) \geq \sum_{\action \in \actions(\state)} \vect{y}(\state, \action)$,
\end{enumerate}
\end{myframe}
By \emph{weak duality} of LPs, for any feasible solution $\vect{x}$ of the primal $\mathsf{LP}_{\mathsf{primal}}^{\max}$ and any feasible solution $\vect{y}$ of the dual $\mathsf{LP}_{\mathsf{primal}}^{\max}$ we have
\[
\sum_{(\state, \action) \in \SA} \transMat(\state, \action, \targetSet) \cdot \vect{y}(\state, \action) \leq \sum_{\state \in \states} \vect{x}(\state).
\]
From \cite{jantsch_certificates_2022}, we know that $\prob_{\state}^{\max}(\eventually \targetSet) \leq \vect{x}(\state)$ for all $\states$ and thus $\sum_{\state \in \states} \prob_{\state}^{\max}(\eventually \targetSet) \leq \sum_{\state \in \states} \vect{x}(\state)$. Note that for the optimal solution $\vect{x}^*$ of $\mathsf{LP}_{\mathsf{primal}}^{\max}$ equality is attained. Altogether we thus get:
\[
\sum_{(\state, \action) \in \SA} \transMat(\state, \action, \targetSet) \cdot \vect{y}(\state, \action) \leq \sum_{\state \in \states} \prob_{\state}^{\max}(\eventually \targetSet) \leq \sum_{\state \in \states} \vect{x}(\state).
\]
Let $\Delta(\vect{x}, \vect{y}) = \card{\sum_{\state \in \states} \vect{x}(\state) - \sum_{\state \in \SA} \transMat(\state, \action, \targetSet) \cdot \vect{y}(\state, \action)}$ denote the \emph{duality gap}. It then follows that $\vect{x}(\state) - \Delta(\vect{x}, \vect{y}) \leq \prob_{\state}^{\max}(\eventually \targetSet) \leq \vect{x}(\state)$ for all $\state \in \states$. Suppose this were not the case, then there exists $\state \in \states$ with $\prob_{\state}^{\max}(\eventually \targetSet) < \vect{x}(\state) - \Delta(\vect{x}, \vect{y})$. Unfolding the definition gives us
\[
\prob_{\state}^{\max}(\eventually \targetSet) + \sum_{\state' \in \states \setminus \{\state\}} \vect{x}(\state') < \sum_{\state \in \SA} \transMat(\state, \action, \targetSet) \cdot \vect{y}(\state, \action).
\]
Since the left side is bounded from below by $\sum_{\state \in \states} \prob_{\state}^{\max}(\eventually \targetSet)$, we get a contradiction. Thus, we have $\vect{x}(\state) - \Delta(\vect{x}, \vect{y}) \leq \prob_{\state}^{\max}(\eventually \targetSet) \leq \vect{x}(\state)$ for all $\state \in \states$.

In summary, we can certify maximal reachability probabilities of all states as in \cite{chatterjee_fixed_2025}, by taking any feasible pair of primal $\vect{x}$ and dual $\vect{y}$ solution. The absolute difference between the upper and lower bounds is then given by the duality gap $\Delta(\vect{x}, \vect{y})$. In particular, the certificates for lower bounds are simpler than those of fixed-point certificates, in the sense that no additional qualitative certificate is required. The only difference to the Farkas certificates presented in \cite{funke_farkas_2020,jantsch_certificates_2022} is the flow constraint in the dual, where we have replaced $\boldDelta(\state)$ with $1$.

\smallskip

\noindent\textit{Modified Farkas certificates for minimal probabilities.} For minimal reachability probabilities we proceed analogously. Let us recall the LPs for minimal reachability probabilities \cite{baier_principles_2008-1}. Let us define $\states_{= 0}^{\min} = \{\state \in \states \mid \prob_{\state}^{\min}(\eventually \targetSet) = 0\}$. This set can be computed via graph algorithms.

\begin{myframe}{{$\mathsf{LP}_{\mathsf{primal}}^{\min}$}}

\medskip

$\max \sum_{\state \in \states} \vect{x}(\state)$ where $\vect{x} \in \realsnn^\states$ subject to
\begin{enumerate}[(1), itemsep=0mm, topsep=0.8mm, parsep=0mm]
\item for all $(\state, \action) \in \SA$: $\vect{x}(\state) \leq \transMat(\state, \action, \targetSet) + \sum_{\state' \in \states} \transMat(\state, \action, \state') \cdot \vect{x}(\state')$
\item for all $\state \in \states_{=0}^{\min}$: $\vect{x}(\state) = 0$
\end{enumerate}
\end{myframe}
\begin{myframe}{{$\mathsf{LP}_{\mathsf{dual}}^{\min}$}}

\medskip

$\min \sum_{\state \in \SA} \transMat(\state, \action, \targetSet) \cdot \vect{y}(\state, \action)$ where $\vect{y} \in \realsnn^\SA$ subject to
\begin{enumerate}[(1), itemsep=0mm, topsep=0.8mm, parsep=0mm]
\item for all $\state \in \states \setminus \states_{= 0}^{\min}$: $1 + \sum_{(\state', \action') \in \SA} \transMat(\state', \action') \cdot \vect{y}(\state, \action) \leq \sum_{\action \in \actions(\state)} \vect{y}(\state, \action)$
\end{enumerate}
\end{myframe}
Again, by weak duality, for any feasible solution $\vect{x}$ of the primal $\mathsf{LP}_{\mathsf{primal}}^{\min}$ and any feasible solution $\vect{y}$ of the dual $\mathsf{LP}_{\mathsf{primal}}^{\min}$ we have for all $\state \in \states$,
\[
\sum_{\state \in \states} \vect{x}(\state) \leq \prob_{\state}^{\min}(\eventually \targetSet) \leq \sum_{(\state, \action) \in \SA} \transMat(\state, \action, \targetSet) \cdot \vect{y}(\state, \action) .
\]
As a consequence, we get $\vect{x}(\state) \leq \prob_{\state}^{\min}(\eventually \targetSet) \leq \vect{x}(\state) + \Delta(\vect{x}, \vect{y})$ for all $\state \in \states$ by the same reasoning as for maximal probabilities. Notice how the primal solution now provides the lower bound, whereas for maximal probabilities they were used as upper bounds.

However, for minimal probabilities it is \emph{not} sufficient to return a feasible pair $(\vect{x}, \vect{y})$, since we additionally need to provide a certificate for the set $\states_{=0}^{\min}$. It actually suffices to consider the set $\states_{\MECS} = \{\state \in \states \mid \states \text{ is contained in an MEC}\}$ instead of $\states_{=0}^{\min}$ (see \cite[Theorem 3.4]{de_alfaro_formal_1997}). \todo{the idea is as follows. we can first construct the quotient, cast the reach problem to the SSSP problem from de alfaro (theorem 3.4) and then setup the corresponding LP. however that is just the same as setting values of MECs to zero.} To certify this set, we can simply use the MEC certificate from \Cref{lemma:mec-certificate}.

\subsection{Certifying Algorithm for Multi-objective Reachability Queries}
\label{appendix:certifying-verification-algorithm}
Now we discuss how the algorithm from \cite{forejt_pareto_2012} for checking satisfaction of multi-objective reachability queries can be modified so that it becomes certifying, i.e. returns the Farkas certificates from \cite{baier_certificates_2024}. In the following, we assume that $\mdp = (\states \union \targets, \actions, \init, \transMat)$ is in reachability form and $\targetSet_1, \ldots, \targetSet_k \subseteq F$. Further, we may assume that all states in $\states$ can reach $F$. We consider the query $\query = \exists \scheduler \centerdot \Conj_{i=1}^k \queryProb_{\geq \lambda_i}^\scheduler(\eventually G_i)$.

\renewcommand{\algorithmicrequire}{\textbf{Input:}}
\renewcommand{\algorithmicensure}{\textbf{Output:}}
\begin{algorithm}[t]
\begin{algorithmic}[1]
\Require MDP $\mdp$ and reachability query $\query = \exists \scheduler \centerdot \Conj_{i=1}^k \queryProb_{\geq}^\scheduler(\eventually G_i)$
\Ensure $\mdp \models \query$ and Farkas certificate
\State $X \gets \emptyset$, $\boldLambda \gets (\lambda_1, \ldots, \lambda_k)$, {$\schedulers \gets \emptyset$}
\Repeat
\State Find $\vect{z} \in \realsnn^k$ separating $\boldLambda$ from $\down{X}$
\State $\scheduler^* \gets \argmax_{\scheduler} \sum_{i=1}^k \vect{z}(i) \cdot \prob^\scheduler(\eventually \targetSet_i)$
\State {$ \schedulers \gets \schedulers \union \{\scheduler^*\}$}
\State $\vect{q} \gets \bigl(\prob_{\init}^{\scheduler^*}(\eventually \targetSet_1), \ldots, \prob_{\init}^{\scheduler^*}(\eventually \targetSet_k)\bigr)$
\If{$\vect{z}^\top \vect{q} < \vect{z}^\top \boldLambda$}
{

\ForAll{$\state \in \states$}
\State $\vect{x}(\state) \gets \sum_{i=1}^k \vect{z}(i) \cdot \prob_\state^{\scheduler^*}(\eventually \targetSet_i)$
\EndFor
}
\State \Return $(\text{false}, \vect{x}, \vect{z})$
\EndIf
\State $X \gets X \union \{\vect{q}\}$
\Until{$\boldLambda \in \down{X}$}
{

\ForAll{$\scheduler_j \in \schedulers = \{\scheduler_0, \ldots, \scheduler_\ell \}$}
\State $\vect{y}_j \gets \text{Expected visiting times under } \scheduler_j$
\EndFor
}
\State {{Find $\gamma_0, \ldots, \gamma_\ell \in [0, 1]$ s.t. $\sum_{j=1}^\ell \gamma_j = 1$ and $\sum_{j=0}^\ell \gamma_j \cdot \prob^{\scheduler_j}(\eventually \targetSet_i) \geq \lambda_i$ for all $i \in [k]$.}}
\State {$\vect{y} \gets \sum_{j=1}^{\ell} \gamma_j \cdot \vect{y}_j$}
\State \Return $(\text{true}, \vect{y})$
\end{algorithmic}
\caption{Modified \emph{certifying} variant of the algorithm from \cite{forejt_pareto_2012} for checking reachability queries. Modifications and additions are indicated in.}
\label{algorithm:certifying-algorithm-reach}
\end{algorithm}

Our modified variant of the algorithm is shown in \Cref{algorithm:certifying-algorithm-reach}. The correctness of the algorithm with respect to the result, that is, whether $\mdp \models \query$ holds, was already established in \cite{forejt_pareto_2012}. Hence, we only need to argue that the returned certificates are indeed valid. In the following, we thus consider the cases when the query is determined to be satisfied and not satisfied.

\begin{lemma}
If \Cref{algorithm:certifying-algorithm-reach} returns $(\text{false}, \vect{x}, \vect{z})$, then $(\vect{x}, \vect{z})$ is a Farkas certificate \cite{baier_certificates_2024} for $\mdp \not\models \exists \scheduler \centerdot \Conj_{i=1}^k \queryProb_{\geq \lambda_i}^\scheduler(\eventually G_i)$ or equivalently $\mdp \models \forall \scheduler \centerdot \Disj_{i=1}^k \queryProb_{< \lambda_i}^\scheduler(\eventually G_i)$.
\end{lemma}
\begin{proof}
By construction of $\vect{x}$, we have $\vect{x}(\state) = \sup_{\scheduler} \sum_{i=1}^k \vect{z}(i) \cdot \prob_{\state}^{\scheduler}(\eventually \targetSet_i)$ for all $\state \in \states$. We then observe the following equivalences
\begin{align*}
\sup_{\scheduler} \sum_{i=1}^k \vect{z}(i) \cdot \prob_{\state}^{\scheduler}(\eventually \targetSet_i)
 &= \sup_{\scheduler}  \sum_{i=1}^k \vect{z}(i) \cdot \sum_{f \in \targetSet_i} \prob_{\state}^\scheduler(\eventually f)\\
  &= \sup_{\scheduler}  \sum_{i=1}^k \sum_{f \in \targetSet_i} \vect{z}(i) \cdot \prob_{\state}^\scheduler(\eventually f)\\
  &= \sup_{\scheduler} \sum_{f \in F} \prob_{\state}^\scheduler(\eventually f) \cdot \sum_{i=1}^k \vect{z}(i) \cdot \vect{1}_{\targetSet_i}(f)
\end{align*}
for all $\state \in \states$. We define $\vect{r}(f) = \sum_{i=1}^k \vect{z}(i) \cdot \vect{1}_{\targetSet_i}(f)$ for all $f \in F$. It follows from the work of \cite{de_alfaro_formal_1997}, that $\vect{x}$ and $\vect{r}$ satisfy
\[
\vect{x}(\state) \geq \sum_{\state \in \states'} \transMat(\state, \action, \state') \cdot \vect{x}(\state') + \sum_{f \in F} \transMat(\state, \action, f) \cdot \vect{r}(f) \quad \text{ for all } (\state, \action) \in \SA.
\]
Expanding the definition of $\vect{r}$ and reordering summation gives us
\[
\vect{x}(\state) \geq \sum_{\state \in \states'} \transMat(\state, \action, \state') \cdot \vect{x}(\state') + \sum_{i=1}^k \transMat(\state, \action, \targetSet_i) \cdot \vect{z}(i) \quad \text{ for all } (\state, \action) \in \SA.
\]
Together with the fact that \Cref{algorithm:certifying-algorithm-reach} only returns $(\text{false}, \vect{x}, \vect{z})$ when $\vect{x}(\init) = \vect{z}^\top \vect{q} < \vect{z}^\top \boldLambda$, we can conclude that $(\vect{x}, \vect{z})$ indeed correspond to the Farkas certificates from \cite{baier_certificates_2024} and this completes the proof. \qed
\end{proof}
We now prove the affirmative case.
\begin{lemma}
If \Cref{algorithm:certifying-algorithm-reach} returns $(\text{true}, \vect{y})$, then $\vect{y}$ is a Farkas certificate \cite{baier_certificates_2024} for $\mdp \models \exists \scheduler \centerdot \Conj_{i=1}^k \queryProb_{\geq \lambda_i}^\scheduler(\eventually G_i)$.
\end{lemma}
\begin{proof}
It follows from \cite{forejt_pareto_2012} that the coefficients $\gamma_0, \ldots, \gamma_\ell$ in Line 15 of \Cref{algorithm:certifying-algorithm-reach} exist. Further, consider the scheduler $\scheduler$ that randomly selects $\scheduler_j$ with probability $\gamma_\ell$ in the initial state and then behaves like $\scheduler_j$. Then, $\scheduler$ is a witness scheduler for $\query$ \cite{forejt_pareto_2012}. Finally, the expected visiting times of $\scheduler$ are exactly given by $\vect{y}$, which is known to be a valid Farkas certificate for $\mdp \models \exists \scheduler \centerdot \Conj_{i=1}^k \queryProb_{\geq \lambda_i}^\scheduler(\eventually G_i)$ \cite{jantsch_certificates_2022}. \qed
\end{proof}

\begin{remark}
Compared to \cite[Algorithm~1]{forejt_pareto_2012}, we have added the additional requirement in Line 3 of \Cref{algorithm:certifying-algorithm-reach} that $\vect{z}$ is nonnegative. However, such a nonnegative separating vector always exists \cite{etessami_multi-objective_2008}.
\end{remark}
\begin{remark}
For queries with upper bounds, the modification of the algorithm is analogous, but additionally requires that the ECs of $\mdp$ are only formed by states in $F$.
\end{remark}

\section{Supplementary Material for \Cref{section:witnesses-mdps}}
\subsection{Omitted Proofs}
In the following, we denote by $\mecClass[\MECS]{\state}$ the class in $\partition_{\MECS}$ a state belongs to.

\begin{lemma}
Let $\states' \subseteq \states$ and $X = \{\mecClass[\MECS]{\state} \mid \state \in \states'\}$. Further, let $\labeling \colon \states \to 2^\labels$ be a \emph{MEC-equivalent} labeling. Then, we have $\labeling(\states') = \widehat{\labeling}(X)$.
\label{lemma:mec-equivalent}
\end{lemma}
\begin{proof}
The statement follows from these equivalences:
\begin{align*}
\labeling(\states') \overset{\text{def}}{=} \Union_{\state \in \states'} \labeling(\state) \overset{(\dag)}{=} \Union_{\state \in \states'} \Union_{\state' \in \mecClass[\MECS]{\state}} \labeling(\state')  \overset{\text{def}}{=} \Union_{\state \in \states'} \widehat{\labeling}(\mecClass[\MECS]{\state}) = \widehat{\labeling}(X).
\end{align*}
The equality $(\dag)$ follows from $\labeling$ being MEC-equivalent.
\qed
\end{proof}

\mecQuotientWitness*
\begin{proof}
Suppose we have $\mecQuotient' \models \widehat{\existsQuery}$ and $\mecQuotient'$ is minimal w.r.t. $\widehat{\labeling}$.
From \Cref{lemma:index} we know that for the corresponding subsystem $\mdp'$ with states $\states'$ we have $\mdp' \models \existsQuery$.
Suppose $\mdp'$ is not minimal. Then, there exists a subsystem $\mdp''$ with states $\states''$ such that $\mdp'' \models \query$ and $\card{\labeling(\states'')} < \card{\labeling(\states')}$. Let $Y = \{\mecClass[\MECS]{\state} \mid \state \in \states'' \}$. From \Cref{lemma:index}, we know that $\mecQuotient_Y \models \widehat{\query}$. Further, from \Cref{lemma:mec-equivalent} we get $\labeling{(\states'')} = \widehat{\labeling}(Y)$. Altogether, we have:
\[
\card{\widehat{\labeling}(Y)} = \card{\labeling(\states'')} < \card{\labeling(\states')} = \card{\widehat{\labeling}(X)}.
\]
However, this contradicts the minimality of $\mecQuotient_X$. \qed
\end{proof}
\forallQueriesMinimality*
\begin{proof}
We first show the following claim:
\begin{claim}
Let $\forallQuery$ be a $(\forall, \lor)$-query, $\mdp = (\states, \actions, \init, \transMat)$ an MDP and $\states' \subseteq \states$ such that there exists an MEC $\mec \in \MECS(\mdp)$ with $\states' \intersection \states(\mec) \neq \emptyset $ and $\states(\mec) \not\subseteq \states'$, i.e.\ $\states'$ partially includes an MEC of $\mdp$. Let $\states'' = \states' \setminus \states(\mec)$. If $\mdp_{\states'} \models \forallQuery$, then $\mdp_{\states''} \models \forallQuery$.
\end{claim}
\begin{proof}
Suppose this was not the case, i.e. $\mdp_{\states''} \not\models \forallQuery$, then there exists a scheduler $\scheduler$ for $\mdp_{\states''}$ such that $\Conj_{i=1}^k \prob_{\mdp_{\states''}}^\scheduler(\prop_i) < \lambda_i$. Since $\mec$ is only partially included it is possible to reach $\exit$ almost surely from $\mec$ \cite[~Chapter 4.4]{jantsch_certificates_2022}. Let $\scheduler'$ be a scheduler for $\mdp_{\states'}$ that behaves like $\scheduler$ in $\states''$. Inside $\mec$, that is, once a state $\states' \setminus \states''$ is reached, it ensures that $\exit$ is reached almost surely. We then get $\prob_{\mdp_{\states'}}^{\scheduler'}(\prop_i) = \prob_{\mdp_{\states''}}^\scheduler(\prop_i)$ for all $i \in [k]$, contradicting $\mdp_{\states'} \models \forallQuery$.
\qed
\end{proof}

Suppose $\mecQuotient$ is a witnessing subsystem for $\widehat{\forallQuery}$ and minimal w.r.t. $\widehat{\labeling}$. Let us denote its states by $X \subseteq \partition_{\MECS}$. The corresponding subsystem for $\mdp$ is denoted by $\mdp'$ and its states with $\states'$. For the sake of contradiction, suppose $\mdp'$ is not minimal w.r.t. $\labeling$. Then, there exists a witnessing subsystem $\mdp''$ with states $\states''$ such that $\card{\labeling(\states'')} < \card{\labeling(\states')}$. Due to the claim above, we may assume that $\states''$ does not partially contain an MEC. Let $Y = \{ \mecClass[\MECS]{\state} \mid \state \in \states''\}$. It follows from \Cref{lemma:index} that $\mecQuotient_Y \models \widehat{\forallQuery}$. We get the following:
\[
\labeling(\states'') = \Union_{\state \in \states''} \labeling(\state) \overset{(\dag)}{=} \Union_{\state \in \states''} \Union_{\state' \in \mecClass[\MECS]{\state}} \labeling(\state') = \widehat{\labeling}{(Y)}.
\]
The equality $(\dag)$ follows from the fact that if a state from $\mecClass[\MECS]{\state}$ is contained in $\mdp''$, then all states in $\mecClass[\MECS]{\state}$ are. Altogether, we get:
\[
\card{\widehat{\labeling}(Y)} = \card{\labeling(\states'')} < \card{\labeling(\states')} = \card{\widehat{\labeling}(X)}
\]
However, this contradicts the minimality of $\mecQuotient'$ w.r.t. $\widehat{\labeling}$. Thus, $\mdp'$ has to be minimal w.r.t. $\labeling$.
\qed
\end{proof}

\subsection{Minimal Witnessing Subsystems for \existsCQ-queries}
\label{appendix:milp-exists-queries}
Unlike for $(\forall, \lor)$-queries, a minimal witnessing subsystem for \existsCQ-queries may \emph{partially} contain MECs \cite[~Chapter 4.4]{jantsch_certificates_2022}, meaning that the previous approach based on the MEC quotient does not yield minimal witnesses in general. We first provide a MILP for determining minimal satisfying ECs in an MDP. This is interesting and useful on its own, as minimal ECs can serve as explanations for how a property is satisfied in the long run. We then use this encoding to provide a MILP for determining minimal witnesses for $\existsCQ$-queries with Rabin properties. For $i \in [k]$, we let $i \oplus 1 \coloneqq 1$ if $i = k$ and $i \oplus 1 \coloneqq i + 1$ otherwise.

\begin{definition}[Copy MDP]
Let $\mathfrak{p} = \tuple{(F_i, E_i)}_{i \in [k]}$ be a tuple of Rabin pairs and $E = \bigcap_{i=1}^k E_i$.  \todo{The way that Rabin conditions are defined, this needs to be $S \setminus \bigcup_{i=1}^k E_i$ I believe.} The \emph{copy MDP} is given by ${\mdp^\mathfrak{p}} = (E \times [k], \actions, \tuple{\init, 1}, \transMat^\mathfrak{p})$ where, for all $\tuple{\state, i}, \tuple{\state', i'} \in E \times [k]$ and $\action \in \actions(\state)$, $\transMat^\mathfrak{p}$ is defined by
\[
\transMat^\mathfrak{p}(\tuple{\state, i}, \action, \tuple{\state', i'}) =
\begin{cases}
  \transMat(\state,\action, \state') & \text{if } i = i' \land \state \notin F_i \text{ or } i \oplus 1 = i' \land \state \in F_i \\
  0 & \text{otherwise}.
\end{cases}  
\]
\label{def:copy-mdp}
\end{definition}
Intuitively, ${\mdp^\mathfrak{p}}$ consists of $k$ copies of $\mdp$ and a ``jump'' from copy $i$ to copy $i \oplus 1$ happens when a state in $F_i$ is reached.
\begin{proposition}
\label{theorem:ec-and-product-ecs}
$\ec \subseteq \SA$ is an EC of $\mdp$ with $\Conj_{i \in [k]}\ec \modelsRabin (F_i, E_i)$ if and only if there is an EC $\ec'$ of ${\mdp^\mathfrak{p}}$ such that
\begin{enumerate*}[1), itemsep=0mm, topsep=0.8mm, parsep=0mm]
\item $\ec'$ contains a state-action pair $(\tuple{\state, i}, \action)$ with $\state \in F_i$ for some $i \in [k]$
\item and $\{(\state, \action) \mid \exists i \in [k] \centerdot (\tuple{\state, i}, \action) \in \ec' \} = \ec$.
\end{enumerate*}
\end{proposition}
\begin{proof}~
\begin{itemize}
\item[$\Rightarrow$:] For states $\tuple{\state, j}$ and $\tuple{\state', j'}$ of ${\mdp^\mathfrak{p}}$, we write $\tuple{\state, j} \rightarrow \tuple{\state', j'}$ if there exists $(\state, \action) \in \ec$ such that $\transMat^\mathfrak{p}(\tuple{\state, j}, \action, \tuple{\state', j'}) > 0$. We then define $\mathsf{Reach}(\tuple{\state, j})$ to be the set of states in ${\mdp^\mathfrak{p}}$ reachable under $\rightarrow$ from $\tuple{\state, j}$.
We claim that there exists a state $\tuple{\state^*, j^*}$ such that for every $\tuple{\state', j'} \in \mathsf{Reach}(\tuple{\state^*, j^*})$ we have $\tuple{\state^*, j^*} \in \mathsf{Reach}(\tuple{\state', j'})$, that is, the states that are reachable from $\tuple{\state^*, j^*}$ can also reach it. To find such a state, we consider the following sequence of states. Let $\tuple{\state_0, j_0}$ be an arbitrary state of ${\mdp^\mathfrak{p}}$ with $\state_0 \in \states(\ec)$ and let $n = 0$.
\begin{enumerate}
\item Check whether all states in $\mathsf{Reach}(\tuple{\state_n, j_n})$ can reach $\tuple{\state_n, j_n}$. If so, then $\tuple{\state_n, j_n}$ is the desired state.
\item If this is not the case, then let $\tuple{\state_{n+1}, j_{n+1}}$ be a state reachable from $\tuple{\state_n, j_n}$ that cannot reach it, increment $n$ and go back to 1.
\end{enumerate}
To see that this computation eventually terminates, observe that we have 
\[
\emptyset \neq \mathsf{Reach}(\tuple{\state_{n+1}, i_{n+1}}) \subseteq \mathsf{Reach}(\tuple{\state_n, j_n}).
\]
Further, for all $0 \leq h \leq n$, we have $\tuple{\state_h, j_h} \notin \mathsf{Reach}(\tuple{\state_{n+1}, j_{n+1}})$. Otherwise $\tuple{\state_{h+1}, j_{h+1}}$ could have reached $\tuple{\state_h, j_h}$, contradicting the choice of $\tuple{\state_{h+1}, j_{h+1}}$. Thus, altogether we have
\[
\emptyset \neq \mathsf{Reach}(\tuple{\state_{n+1}, j_{n+1}}) \subseteq \mathsf{Reach}(\tuple{\state_n, j_n}) \setminus \{\tuple{\state_0, j_0}, \ldots, \tuple{\state_n, j_n}\}
\]
and we can conclude the existence of $\tuple{\state^*, j^*}$. We then define $\ec_{\mdp^\mathfrak{p}}$ by
\[
\ec' = \{ (\tuple{\state, j}, \action) \mid (\state, \action)\in \ec \land \tuple{\state, j} \in \mathsf{Reach}(\tuple{\state^*, j^*})\}.
\]
By construction, $\ec'$ is strongly connected. To see that both conditions are satisfied by $\ec'$, we note that for any state $\tuple{\state, j}$ in ${\mdp^\mathfrak{p}}$ and $\state' \in \states(\ec)$, there exists some $j'$ such that $\tuple{\state', j'} \in \mathsf{Reach}(\tuple{\state, j})$. This follows from the fact that $\ec$ is an EC and contains some state from $F_i$ for every $i \in [k]$.   
\item[$\Leftarrow$:] By assumption, we have $\ec = \{(\state, \action) \mid \exists i \in [k] \centerdot (\tuple{\state, i}, \action) \in \ec'\}$. It is easy to see that $\ec$ forms an EC in $\mdp$. Due to the first condition imposed on $\ec'$, we know that \emph{for every} $i \in [k]$, there exists $(\tuple{\state, i}, \action)$ with $\state \in F_i$. Intuitively, if $\ec'$ contains one such pair, it needs to ``go through all copies'' in order to form an EC. Together with the fact that $\ec$ only contains pairs from $\bigcap_{i=1}^k E_i$ by construction of ${\mdp^\mathfrak{p}}$, we get $\ec \models_{\mathsf{Rabin}} (F_i, E_i)$ for all $i \in [k]$.
\end{itemize}
\qed
\end{proof}
The previous theorem relates ECs that satisfy all Rabin pairs in $\mdp$ with corresponding ECs in ${\mdp^\mathfrak{p}}$. %

\medskip

\noindent\emph{Minimal satisfying ECs:} 
For an MDP $\mdp$ and tuple of Rabin pairs $\mathfrak{p} = \tuple{(F_i, E_i)}_{i \in [k]}$, the \emph{copy MDP} $\mdp^{\mathfrak{p}}$ consists of $k$ copies of $\mdp$ and a ``jump'' from copy $i$ to $i \oplus 1$ happens from states in $F_i$ (\Cref{def:copy-mdp}). We have established the relation between ECs of $\mdp$ and $\mdp^\mathfrak{p}$ in \Cref{theorem:ec-and-product-ecs}, namely that ECs in $\mdp^\mathfrak{p}$ including such ``jumps'' correspond to satisfying ECs for pairs in $\mathfrak{p}$.
We now encode this observation into a system of linear equalities. It has been shown in \cite[Lemma~3.8]{jantsch_certificates_2022} that vectors $\vect{x} \in \realsnn^{\SA_{\mdp^\mathfrak{p}}}$ that satisfy, for all states $\tuple{\state, i}$ of $\mdp^\mathfrak{p}$,\todo{Would it make sense to use $\vec{y}$ here instead of $\vec{x}$? }
\[
\sum\nolimits_{(\tuple{\state', j}, \action') \in \SA_{\mdp^\mathfrak{p}}} \transMat^\mathfrak{p}(\tuple{\state', j}, \action', \tuple{\state, i}) \cdot \vect{x}(\tuple{\state', j}, \action') = \sum\nolimits_{\action \in \actions(\state)} \vect{x}(\tuple{\state, i}, \action),
\]
induce ECs in $\mdp^\mathfrak{p}$. More precisely, for $\ec_1, \ldots, \ec_\ell \subseteq \ECS(\mdp^\mathfrak{p})$ there exists $\vect{x}$ as above, such that $\stateSupp{\vect{x}} = \ec_1 \union \ldots \union \ec_\ell$ and vice versa. Let $\mathcal{S}_\mathfrak{p}$ denote the set of vectors $\vect{x}$ that satisfy these equalities w.r.t. $\mdp^\mathfrak{p}$. Then, we only need to enforce that the ``jumps'' have positive values. The MILP $\mathsf{RabinEC}$ in \Cref{figure:milp-rabin} builds on this idea. For a labeling $\labeling$ and $\ec \subseteq \SA$, we define $\labeling(\ec) = \Union_{(\state, \action) \in \ec} \labeling(\state)$.
\begin{theorem}
Let $\mdp$ be an MDP, $\mathfrak{p} = \tuple{(F_j, E_j)}_{j \in [\ell]}$, $\ec \subseteq \SA_\mdp$ and $\labeling \colon \states \to (2^\labels \setminus \emptyset)$ a labeling. Then, $\ec$ is an EC of $\mdp$ such that $\Conj_{j=1}^\ell\ec \modelsRabin (F_j, E_j)$ and $\card{\labeling(\ec)}$ is minimal if and only if there is an \emph{optimal solution} $(\boldBeta, \vect{x})$ of $\mathsf{RabinEC}\left(\mdp, \labeling, \tuple{(F_j, E_j)}_{j \in [k]}\right)$ with $\ec = \{(\state, \action) \in \SA \mid \exists j \in [k] \centerdot \vect{x}(\tuple{\state, j}, \action) > 0\}$.
\end{theorem}
\begin{proof}~
\begin{itemize}
\item[$\Rightarrow$:] From \Cref{theorem:ec-and-product-ecs}, we know that there exists a corresponding EC $\ec_\mathfrak{p}$ in $\mdp^\mathfrak{p}$ such that $\ec = \{(\state, \action) \in \SA_\mdp \mid \exists j \in [\ell] \centerdot (\tuple{\state, j}, \action) \in \ec'\}$ and $(\tuple{\state, j}, \action) \in \ec_\mathfrak{p}$ for some $j \in [k]$, $\state \in F_j$ and $\action \in \actions$. As discussed above, it was shown in \cite[Lemma~3.8]{jantsch_certificates_2022} that there exists an $\vect{x} \in \realsnn^{\SA_{\mdp^\mathfrak{p}}}$ such that $\transMat_{\mdp^\mathfrak{p}}^\top \vect{x} = 0$ and $\stateSupp{\vect{x}} = \ec_\mathfrak{p}$. It directly follows that there exists some $j \in [\ell]$, $\state \in F_j$ and $\action \in \actions$ for which we have $\vect{x}(\tuple{\state, j}, \action) > 0$. Further, since $\vect{x}$ is in the kernel of $\transMat_{\mdp^\mathfrak{p}}^\top$, w.l.o.g. we may assume $\vect{x}(\tuple{\state, j}, \action) \geq 1$.

\medskip

It remains to be shown that there exists $\boldBeta \in \{0, 1\}^\labels$ so that $(\boldBeta, \vect{x})$ is an optimal solution.
We define $\boldBeta(l) = 1$ if $\vect{x}(\tuple{\state, j}, \action) > 0$ and $l \in \labeling(\state)$ for some $j \in [\ell]$, $\state \in \states$ and $\action \in \actions$. Otherwise we define $\boldBeta(l)$ to be zero. For the sake of contradiction, suppose $(\boldBeta, \vect{x})$ is not optimal. Let $(\boldBeta', \vect{x}')$ be an optimal solution. It follows from \cite[Lemma~3.8]{jantsch_certificates_2022} that $\stateSupp{\vect{x}'} \eqqcolon\ec_{\mathfrak{p}}'$ forms an EC in $\mdp^\mathfrak{p}$\todo{here, the restriction that $\labeling$ does not map to the empty set is crucial.}. Let $\ec'$ be the corresponding EC in $\mdp$. It follows that:
\begin{align*}
\card{\labeling(\ec')} = \card{\Union\nolimits_{(\tuple{\state,j}, \action) \in \ec_\mathfrak{p}'} \labeling(\state)} &= \sum_{l \in \labels} \boldBeta'(\state) \\ &< \sum_{l \in \labels} \boldBeta(\state) = \card{\Union\nolimits_{(\tuple{\state, j}, \action) \in \ec_\mathfrak{p}} \labeling(\state)} = \card{\labeling(\ec)}.
\end{align*}
However, this contradicts the minimality of $\ec$. Hence, $(\boldBeta, \vect{x})$ is an optimal solution.
\item[$\Leftarrow$:] Let $(\boldBeta, \vect{x})$ be an optimal solution. Again, from \cite[Lemma~3.8]{jantsch_certificates_2022} we know that $\stateSupp{\vect{x}} \eqqcolon \ec_\mathfrak{p}$ forms an EC in $\mdp^\mathfrak{p}$\todo{again, this follows from optimality}. We observe that $(\tuple{\state, j}, \action) \in \ec_\mathfrak{p}$ for some $j \in [\ell]$ and $\state \in F_j$, due to the second constraint in the MILP. By \Cref{theorem:ec-and-product-ecs} this yields an EC $\ec$ of $\mdp$ that satisfies all Rabin pairs. It remains to be shown that $\card{\labeling(\ec)}$ is minimal. Again, we prove by contradiction. Suppose $\ec$ is not minimal, then there exists a minimal EC $\ec'$ for which we can find a corresponding EC $\ec_\mathfrak{p}'$ in $\mdp^\mathfrak{p}$. Similarly, we can then show that there must exist a solution $(\boldBeta', \vect{x}')$ such that $\sum_{l \in \labels} \boldBeta'(l) < \sum_{l \in \labels} \boldBeta(l)$, contradicting the optimality of $(\boldBeta, \vect{x})$. Hence, $\ec$ is minimal w.r.t. $\labeling$.
\end{itemize}
\qed
\end{proof}

\medskip

\noindent\emph{Minimal witnesses for $\existsCQ$-queries:} Based on the previous result, we can now provide a MILP encoding for determining minimal witnessing subsystems for arbitrary labelings and Rabin properties. For the remainder of this section, we fix Rabin properties $\{\prop_i\}_{i \in [k]}$ with pairs $\{\pairs_i\}_{i \in [k]}$. For all $I = \{i_1, \ldots, i_\ell\} \subseteq [k]$,  we define $\pairs_I = \pairs_{i_1} \times \cdots \times \pairs_{i_\ell}$. Note that an EC satisfies all Rabin properties $\{\pairs_i\}_{i \in I}$ if and only if it satisfies all Rabin pairs in $\mathfrak{p}$ for some $\mathfrak{p} \in \pairs_I$.
\begin{figure}[t]
\centering

\begin{myframe}{{$\mathsf{RabinEC}\bigl(\mdp, \labeling,\mathfrak{p} = \tuple{(F_j, E_j)}_{j \in [\ell]}\bigr)$}}

\smallskip

$\min \sum\nolimits_{l \in \labels} \boldBeta(\state)$ s.t. $\boldBeta \in \{0, 1\}^\labels$, $\vect{x} \in \mathcal{S}_\mathfrak{p}$ and
\begin{enumerate}[(1), itemsep=0mm, topsep=0.8mm, parsep=0mm]
\item for all $(\tuple{\state, j}, \action) \in \SA_{\mdp^\mathfrak{p}}$ and $l \in \labeling(\state)$: $\vect{x}(\tuple{\state, j}, \action) \leq \boldBeta(l) \cdot M$, \label{milp:rabin-ec-transient-indicator}
\item $\sum_{j \in [\ell]} \sum_{\state \in F_j} \sum_{\action \in \actions(\state)} \vect{x}(\tuple{\state, j}, \action) \geq 1$. \label{milp:rabin-ec-positive}
\end{enumerate}
\end{myframe}

\begin{myframe}{{$\mathsf{Rabin}\bigl(\mdp, \labeling, \{\pairs_i\}_{i \in [k]}, \{\lambda_i\}_{i \in [k]}\bigr)$}}

\smallskip

$\min\sum\nolimits_{l \in \labels} \boldBeta(\state)$ s.t. $\boldBeta \in \{0, 1\}^\labels$, $\vect{y} \in \realsnn^{\SA \union \states}$, $\{\vect{y}_I \in \realsnn^\states\}_{I \subseteq [k]}$, $\{\vect{x}_{\mathfrak{p}} \in \mathcal{S}_{\mathfrak{p}} \}_{I \subseteq [k], \mathfrak{p} \in \pairs_I}$ and
\begin{enumerate}[(1), itemsep=0mm, topsep=0.8mm, parsep=0mm]
\item for all $(\state, \action) \in \SA$ and $l \in \labeling(\state)$: $\vect{y}(\state, \action) \leq \boldBeta(l) \cdot M$, \label{milp:rabin-transient-indicator}
\item for all $I \subseteq [k]$, $\mathfrak{p} \in \pairs_{I}$, $(\tuple{\state, j}, \action) \in \SA_{\mathfrak{p}}$ and $l \in \labeling(\state)$: $\vect{x}_{\mathfrak{p}}(\tuple{\state, j}, \action) \leq \boldBeta(l) \cdot M$, \label{milp:rabin-recurrent-indicator}
\item for all $\state \in \states$: $\initDistr(\state) + \sum_{(\state', \action') \in \SA}\vect{y}(\state', \action') \cdot \transMat(\state', \action', \state) = \vect{y}(\state)$, \label{milp:rabin-flow-1}
\item for all $\state \in \states$: $\vect{y}(\state) \geq \sum_{\action \in \actions(\state)} \vect{y}(\state, \action) + \sum_{I \subseteq [k]} \vect{y}_I(\state)$, \label{milp:rabin-flow-2}
\item for all $i \in [k]$: $\sum_{\state \in \states} \sum_{i \in I \subseteq [k]} \vect{y}_I(\state) \geq \lambda_i$,  \label{milp:rabin-threshold}
\item for all $I \subseteq [k]$ and $\state \in \states$: $\vect{x}_I(\state) \geq \vect{y}_I(\state)$ \label{milp:rabin-positive}
\end{enumerate}
\end{myframe}
\caption{MILP for finding minimal satisfying ECs and witnessing subsystems for $(\exists,\land)$-queries with simple Rabin objectives, respectively.}
\label{figure:milp-rabin}
\end{figure}
For $\vect{x}_\mathfrak{p} \in \mathcal{S}_\mathfrak{p}$, we define $\stateSupp{\vect{x}_\mathfrak{p}} = \{\state \in \states \mid \exists i \in I \centerdot \exists \action \in \actions(\state) \centerdot (\tuple{\state, i}, \action) \in \supp{\vect{x}_\mathfrak{p}}\}$. For readability, for all $I \subseteq [k]$ and $\state \in \states$ we define 
\[
\vect{x}_I(\state) \coloneqq \sum_{\mathfrak{p} = \tuple{(F_1, E_1), \ldots, (F_\ell, E_\ell)} \in \pairs_I} \sum_{j \in [\ell]} \sum_{f \in F_j} \sum_{\action \in \actions(f)} \vect{1}_{\state}(f) \cdot \vect{x}_\mathfrak{p}(\tuple{\state, j}, \action).
\]
The idea is that if $\vect{x}_I(\state) > 0$, then $\state$ is contained in an EC that satisfies all Rabin properties indexed by $I$.
The MILP $\mathsf{Rabin}$ is given in \Cref{figure:milp-rabin}. 
Let us provide an intuition for the encoding. The constraints \ref{milp:rabin-transient-indicator} and \ref{milp:rabin-recurrent-indicator} enforce that if $\boldBeta(l) = 0$, then the label $l$ may not appear in the subsystem. The constraints \ref{milp:rabin-flow-1} and \ref{milp:rabin-flow-2} encode a transient flow. Intuitively, \ref{milp:rabin-threshold} requires that the Rabin property $\pairs_i$ is satisfied with probability at least $\lambda_i$. The variable $\vect{y}_I(\state)$ encodes the probability that starting from a state $\state$ an EC, in which all properties $\{\pairs\}_{i \in I}$ are satisfied, is realized. Lastly, we need to encode that if $\vect{y}_I(\state) > 0$, such an EC actually exists and can be realized. This is handled in \ref{milp:rabin-positive}. Note the similarity to constraint \ref{milp:rabin-ec-positive} in \textsf{RabinEC}. The following theorem proves the correctness of the encoding.
\begin{theorem}
Let $\mdp$ be an MDP and $\query = \exists \scheduler \centerdot \Conj_{i=1}^k \queryProb_{\geq \lambda_i}(\prop_i)$ an $\existsCQ$-query. Further, let $\states' \subseteq \states$. Then, $\mdp_{\states'}$ is a witnessing subsystem for $\query$ (i.e. $\mdp_{\states'} \models \query$) if and only if there is a solution $(\boldBeta, \vect{y}, \{\vect{y}_I\}_{I \subseteq [k]}, \{\vect{x}_\mathfrak{p}\}_{I \subseteq [k], \mathfrak{p} \in \pairs_I})$ of 
\[\mathsf{Rabin}(\mdp, \labeling, \{\pairs_i\}_{i \in [k]}, \{\lambda_i\}_{i \in [k]})
\] 
such that 
\[
\stateSupp{\vect{y}} \union \Union\nolimits_{I \subseteq [k]} \Union\nolimits_{\mathfrak{p} \in \pairs_I} \stateSupp{\vect{x}_\mathfrak{p}} \subseteq \states'.
\]
\label{theorem:rabin-milp}
\end{theorem}
\begin{proof}
We define $\mdp' \coloneqq \mdp_{\states'}$.
\begin{itemize}
\item[$\Rightarrow$:] Suppose $\mdp'$ is a witnessing subsystem for $\query$ and let $\scheduler$ be a satisfying scheduler. We need to show that there exists a corresponding solution of $\mathsf{Rabin}\bigl(\mdp, \labeling, \{\pairs_i\}_{i \in [k]}, \{\lambda_i\}_{i \in [k]}\bigr)$. 
For all $I \subseteq [k]$, we define the events
\[
A_I \coloneqq \Conj_{i \in I} \prop_i \land \Conj_{i \in [k]\setminus I} \bar{\prop}_i.
\]
Let $I = \{i_1, \ldots, i_\ell \} \subseteq [k]$. We collect the MECs in which the event $A_I$ is satisfied with positive probability, formally given by
\[
\mathcal{U}_I = \{\mec \in \MECS(\mdp') \mid \prob_{\mdp'}^\scheduler(A_I \mid \eventually \globally \mec) > 0 \}.
\]
We observe that for any MEC $\mec \in \mathcal{U}_I$ there exists an EC $\ec' \subseteq \mec$ and $\mathfrak{p} = \tuple{(F_j, E_j)}_{j \in [\ell]}\in \pairs_I$ such that $\Conj_{j \in [\ell]}\ec' \modelsRabin (F_j, E_j)$. Let us denote such satisfying EC with $\ec_\mathfrak{p}$. Note that this EC is not necessarily unique. By \Cref{theorem:ec-and-product-ecs}, for all such ECs $\ec_\mathfrak{p}$ there exists a corresponding EC $\ec_\mathfrak{p}'$ in $\mdp^\mathfrak{p}$. We collect all these corresponding ECs of $\mdp^\mathfrak{p}$ in the set $\mathcal{D}_{\mathfrak{p}} \subseteq \ECS(\mdp^\mathfrak{p})$. It follows from \cite[Lemma~3.8]{jantsch_certificates_2022} that there exists $\vect{x}_\mathfrak{p} \in \mathcal{S}_\mathfrak{p}$ so that the support of $\vect{x}_\mathfrak{p}$ induces the ECs in $\mathcal{D}_\mathfrak{p}$. More precisely, $\supp{\vect{x}_\mathfrak{p}} = \Union_{\ec' \in \mathcal{D}_\mathfrak{p}} \ec'$ and $\transMat_{\mdp^\mathfrak{p}}^\top \vect{x}_\mathfrak{p} = 0$. Further, we have $\stateSupp{\vect{x}_\mathfrak{p}} \subseteq \states'$ for all $I \subseteq [k]$ by definition of $\mathcal{U}_I$.

\medskip

It remains to be shown that there exist a corresponding $\vect{y}$ and $\{\vect{y}_I\}_{I \subseteq [k]}$. \todo{the following is quite technical. Intuitively, we need to ``connect'' the transient with the recurrent part at the right states.}

\medskip

By \Cref{theorem:ec-and-product-ecs}, for every $\mec \in \mathcal{U}_I$ there exists $\mathfrak{p} = \tuple{(F_j, E_j)}_{j \in [\ell]} \in \pairs_I$ and satisfying EC $\ec_\mathfrak{p} \subseteq \mec$ and $(\state, \action) \in \ec_\mathfrak{p}$ and $j \in [\ell]$ such that $\state \in F_j$ and $\vect{x}_\mathfrak{p}(\tuple{\state, j}, \action) > 0$. Let $\state_{\mec, \mathfrak{p}}$ denote such a state where this property is satisfied.

\medskip

Consider the MDP $\mathcal{N}$ obtained from $\mdp'$ as follows. For every MEC $\mec \in \MECS(\mdp')$, we add an action $\tau_\mec$ from all of its states leading to absorbing state $\bot_\mec$ with probability $1$. Further, for every $I \subseteq [k]$ and $\mec \in \mathcal{U}_I$ we add an action $\tau_I$ from $\state_{\mec,\mathfrak{p}}$ leading to absorbing state $\bot_{I}$ with probability $1$. We show that there exists a scheduler $\scheduler''$ such that for all $I \subseteq [k]$ we have:
\[
\prob_{\mathcal{N}}^{\scheduler''}(\eventually \bot_I) = \prob_{\mdp'}^\scheduler(A_I).
\]
Such a scheduler $\scheduler''$ is obtained as follows. By \cite[Lemma 4]{randour_percentile_2015}, there exists a scheduler $\scheduler'$ for which we have $\prob_{\mathcal{N}}^{\scheduler'}(\eventually \bot_\mec) = \prob_{\mdp'}^\scheduler(\eventually \globally \mec)$ for all $\mec \in \MECS(\mdp')$. Then, we modify $\scheduler'$ into a scheduler $\scheduler''$ in the following way. Whenever $\scheduler'$ chooses $\tau_\mec$, $\scheduler''$ instead switches to a scheduler $\scheduler_{\mec,\mathfrak{p}}$ with probability $\prob_{\mdp'}^\scheduler(A_I \mid \eventually \globally \mec)$. The scheduler $\scheduler_{\mec,\mathfrak{p}}$ ensures that $\state_{\mec,\mathfrak{p}}$ is reached almost surely and once it is reached chooses action $\tau_I$. This is possible, since $\prob_{\mdp'}^\scheduler(A_I \mid \eventually \globally \mec) > 0$ implies the existence of $\state_{\mec,\mathfrak{p}}$ and inside an MEC each state can be reached almost surely. Then, we have for all $I \subseteq [k]$:
\begin{align*}
\prob_{\mathcal{N}}^{\scheduler''}(\eventually \bot_I) &= \sum_{\mec \in \MECS(\mdp')} \prob_{\mathcal{N}}^{\scheduler'}(\eventually \bot_\mec) \cdot \prob_{\mdp'}^{\scheduler}(A_I \mid \eventually \globally \mec)\\
&= \sum_{\mec \in \MECS(\mdp')} \prob_{\mdp'}^{\scheduler}(\eventually \globally \mec) \cdot \prob_{\mdp'}^{\scheduler}(A_I \mid \eventually \globally \mec) = \prob_{\mdp'}^\scheduler(A_I).
\end{align*}
It follows from \Cref{lemma:reach-certs}, that there exists a vector $\vect{y}'$ with $\stateSupp{\vect{y}'} \subseteq \states'$, satisfying the constraints \ref{milp:rabin-flow-1} and \ref{milp:rabin-flow-2} of the MILP
\[
\mathcal{F}_{\triangleright}\bigl(\mathcal{N}, \labeling, \{\bot_I\}_{I \subseteq [k]}, \{\gamma_I\}_{I \subseteq [k]}\bigr)
\] 
in \Cref{fig:milps-reachability} where $\gamma_I = \prob_{\mdp'}^\scheduler(A_I)$ for all $I \subseteq [k]$. We obtain the desired $\vect{y}$ and $\{\vect{y}_I\}_{I \subseteq [k]}$ by setting $\vect{y}(\state, \action) = \vect{y}'(\state, \action)$ for all $(\state, \action) \in \SA_{\mdp'}$ and zero otherwise. Further, we set $\vect{y}_I(\state) = \vect{y}'(\state, \tau_I)$ for all $\state \in \states'$ and where $\tau_I$ is available and zero otherwise.

\medskip

By construction, $\vect{y}$ and $\{\vect{y}_I\}_{I \subseteq [k]}$ satisfy constraint \ref{milp:rabin-flow-1} and \ref{milp:rabin-flow-2} in \Cref{figure:milp-rabin}. Constraint \ref{milp:rabin-threshold} is satisfied since
\begin{align*}
\sum_{\state \in \states} \sum_{i \in I \subseteq [k]} \vect{y}_I(\state) \geq \sum_{i \in I \subseteq [k]} \gamma_I &= \sum_{i \in I \subseteq [k]} \prob_{\mdp'}^{\scheduler}(A_I) 
\\&= \prob_{\mdp'}^{\scheduler}(\prop_i) \geq \lambda
\end{align*}
for all $i \in [k]$. 
To see that constraint \ref{milp:rabin-positive} is satisfied, we observe that, $\vect{y}_I(\state) > 0$ implies that the $\tau_I$ action is available in $\state$ in $\mathcal{N}$. By construction, this implies
\begin{align*}
\vect{x}_\mathfrak{p}(\tuple{\state, j}, \action) > 0 \text{ and } \state \in F_j \text{ for some } & \mathfrak{p} = \tuple{(F_1, E_1), \ldots, (F_\ell, E_\ell)} \in \pairs_I,\\ &j \in [\ell] \text{ and } \action \in \actions.
\end{align*}
Together with observation that for any $\xi > 0$, we have $\transMat_{\mdp^\mathfrak{p}}^\top \xi \cdot \vect{x}_\mathfrak{p} = 0$, we see that constraint \ref{milp:rabin-positive} is satisfied. Lastly, it is easy to see that we can then choose $\boldBeta$ accordingly.
\item[$\Leftarrow$:] Let $(\boldBeta, \vect{y}, \{\vect{y_I}\}_{I \subseteq [k]}, \{\vect{x}_{\mathfrak{p}} \in \mathcal
{S}_{\mathfrak{p}} \}_{I \subseteq [k], \mathfrak{p} \in \pairs_I})$ be a solution of the MILP such that 
\[
\stateSupp{\vect{y}} \union \Union_{I \subseteq [k]} \Union_{\mathfrak{p} \in \pairs_I} \stateSupp{\vect{x}_\mathfrak{p}} \subseteq \states'.
\]
We will now show that there exists a satisfying scheduler $\scheduler$ for $\mdp'$. To this end, we make the following observations.

\medskip

\textit{Observation 1:} Let $\mathcal{N}$ be the MDP obtained from $\mdp'$ by adding a fresh action $\tau_I$ in $\state$ leading to a fresh absorbing state $\bot_I$ with probability $1$, for all $\state \in \states'$ and $I \subseteq [k]$ with $\vect{y}_I(\state) > 0$.
We define the vector $\vect{y}_\mathcal{N} \in \realsnn^{\SA_{\mathcal{N}}}$ by $\vect{y}_\mathcal{N}(\state, \action) = \vect{y}(\state, \action)$ for all $(\state, \action) \in \SA_{\mdp'}$ and $\vect{y}_\mathcal{N}(\state, \tau_I) = \vect{y}_I(\state)$ for all $\state \in \states'$ and $I \subseteq [k]$.
Then, $\vect{y}_\mathcal{N}$ satisfies the constraints in \Cref{lemma:reach-certs} for $\existsCQ$-queries where the target states are given by  $\{\bot_I \}_{I \subseteq [k]}$ and the thresholds by $\gamma_I = \sum_{\state \in \states} \vect{y}_I(\state)$ for all $I \subseteq [k]$. 

\smallskip

From \Cref{lemma:reach-certs}, we know that there exists a memoryless scheduler $\scheduler'$ for $\mathcal{N}$ such that $\prob_{\mathcal{N}}^{\scheduler'}(\eventually \bot_I) \geq \gamma_I$ for all $I \subseteq [k]$.

\medskip

\textit{Observation 2:} Because of constraint \ref{milp:rabin-positive}, we have that 
\begin{align*}
\vect{y}_I(\state)> 0 \implies & \text{there exists } \mathfrak{p} = \tuple{(F_1, E_1), \ldots, (F_\ell, E_\ell)} \in \pairs_I \text{ and } j \in [\ell]\\
& \text{such that } \state \in F_j \text{ and } \vect{x}_\mathfrak{p}(\tuple{\state, j}, \action) > 0 \text{ for some } \action \in \actions.
\end{align*}
We know from \cite[Lemma~3.8]{jantsch_certificates_2022} that $\stateSupp{\vect{x}_\mathfrak{p}}$ induces ECs in $\mdp'$. Thus, if $\vect{y}_I(\state) > 0$, we know that $\state$ is part of an EC. Intuitively, $\vect{x}_\mathfrak{p}(\tuple{\state, j}, \action) > 0$ implies that $\state$ is contained in an EC of $\mdp^\mathfrak{p}$ that contains a jump from copy $j$ to copy $j \oplus 1$. By \Cref{theorem:ec-and-product-ecs}, we then know that all states $\state \in \states$ with $\vect{y}_I(\state) > 0$ are contained in an EC $\ec_\mathfrak{p}$ of $\mdp'$ that satisfies all pairs in $\mathfrak{p}$.

\medskip

We are now ready to construct a satisfying scheduler $\scheduler$. In the first phase, $\scheduler$ follows the memoryless scheduler $\scheduler'$. Whenever, $\scheduler'$ chooses action $\tau_I$ in a state $\state$, $\scheduler$ switches to a scheduler that realizes the corresponding $\ec_\mathfrak{p}$, for some $\mathfrak{p} \in \pairs_I$, that contains $\state$ and satisfies all Rabin pairs indexed by $I$. Using observation 2 and the fact that by construction $\tau_I$ is only available in states $\state \in \states'$ with $\vect{y}_I(\state) > 0$, we know such an EC must exist.

\medskip 

Finally, for all $i \in [k]$ we get:
\begin{align*}
\prob_{\mdp'}^{\scheduler}(\prop_i) &\geq \prob_{\mdp'}^{\scheduler}(\text{``Switch to satisfying EC for } I \text{ with } i \in I \text{''})\\
&= \sum_{i \in I \subseteq [k]} \prob_{\mathcal{N}}^{\scheduler'}(\eventually \bot_I)\\
&\geq \sum_{i \in I \subseteq [k]} \gamma_I = \sum_{\state \in \states} \sum_{i \in I \subseteq [k]} \vect{y}_I(\state) \geq \lambda_i
\end{align*}
\qed
\end{itemize}
\end{proof}

\section{Supplementary Material for \Cref{section:dtmc}}
\label{appendix:dtmc}
Given matrix $\vect{B}$ and $D, D' \subseteq \states \times \autoStates$, let $\vect{B}_{D,D'}$ be the restriction of $\vect{B}$ to the rows and columns in $D$ and $D'$, respectively.

\begin{remark}[Sub-stochastic chains.] The results from \cite{baier_markov_2023-1} can easily be extended to sub-stochastic Markov chains, in which we require only $\sum_{\state' \in \states}\transMat(\state, \state') \leq 1$ for each $\state \in \states$ (instead of $=1$).
Let $\dtmc = (\states, \init, \transMat)$ be a sub-stochastic Markov chain and $\automaton$ be an UBA over alphabet $\alphabet = S$.
Define $\dtmc' = (\states \cup \{ \exit \},\init,\transMat')$, where $\transMat'(s, \exit) = 1 - \sum_{\state' \in \states} \transMat(\state, \state')$, $\transMat'(\exit, \exit) = 1$ and $\transMat'(\state,\state') = \transMat(\state,\state')$ for all $\state, \state' \in \states$.
Clearly, $\prob_{\dtmc',\state}(\lang{\automaton, \autoState}) = \prob_{\dtmc,\state}(\lang{\automaton, \autoState})$, since $\automaton$ immediately rejects any word containing letter $\exit$.
Furthermore, the matrix $B_{\dtmc' \times \automaton}$ restricted to $\states \times \autoStates$ is equal to $B_{\dtmc \times \automaton}$, since $\automaton$ does not have any transition labeled by $\exit$. Hence, the approach outlined to compute the value vector for $\dtmc \times \automaton$ can also be applied to sub-stochastic $\dtmc$.
\end{remark}

\begin{lemma}
\label{lemma:strongly-connected-square-matrices}
Let $\vect{A} \in \realsnn^{\states \times \states}$ be strongly connected and $\vect{x} \in \realsnn^\states$. Then, the following statements hold:
\begin{enumerate}[itemsep=0mm, topsep=0.8mm, parsep=0mm]
\item If $\vect{A} \vect{x} \geq \rho(\vect{A}) \cdot \vect{x}$ holds, then we have $\vect{A} \vect{x} = \rho(\vect{A}) \cdot \vect{x}$.
\item The eigenspace associated with $\rho(\vect{A})$ is one-dimensional.
\end{enumerate}
\end{lemma}
\begin{proof}
The first statement follows from \cite[Theorem 1.3.31]{berman_nonnegative_1994} and the second from \cite[Theorem 2.1.4]{berman_nonnegative_1994}.
\end{proof}

\begin{restatable}[Monotonicity]{lemma}{dtmcMonotonicity}
  \label{lemm:monotonicity}
  For all $\vect{x} \in [0,1]^{\states \times \autoStates}$ such that $\vect{B} \vect{x} \geq \vect{x}$ and $\boldMu_D^{\top} \vect{x}_D \leq 1$ for each $D \in \mathcal{D}^+$, we have $\vect{x} \leq \valueVector$.
\end{restatable}
\begin{proof}
We prove the statement via induction over the DAG of SCCs of $\vect{B}$. Let $D$ be an SCC of $\vect{B}$ (possibly trivial) and let $\partition \downarrow$ denote the SCCs directly below $D$. We then have $\vect{x}_D \leq \vect{B}_{D} \vect{x}_D + \sum_{C \in \partition \downarrow} \vect{B}_{D, C} \vect{x}_C$. We distinguish between recurrent and non-recurrent SCCs.
\begin{case}[Recurrent SCCs]
Let $D$ be a recurrent SCC. By induction hypothesis and \cite[Proposition 13]{baier_markov_2023-1} we have $\vect{x}_C = 0$ for all $C \in \partition \downarrow$. Thus, we have $\vect{B}_D \vect{x}_D \geq \vect{x}_D$ and $\boldMu_D^\top \ \vect{x}_D \leq 1$. 
Recall that by definition of recurrent SCCs, we have $\rho(\vect{B}_D) = 1$.
Hence, from the first statement in \Cref{lemma:strongly-connected-square-matrices}, we know that $\vect{B}_D \vect{x}_D = \vect{x}_D$ holds. Further, due to the second statement in \Cref{lemma:strongly-connected-square-matrices}, we know that the eigenspace associated with $\rho(\vect{B}_D)$ is one-dimensional. It follows that $\vect{x}_D = \valueVector_D \cdot \gamma$ for some $\gamma \in [0, 1]$ and thus $\vect{x}_D \leq \valueVector_D$.

\end{case}
\begin{case}[Non-recurrent SCCs]
Now let $D$ be a non-recurrent SCC. By induction hypothesis, we have:
\[
\vect{x}' \; \coloneqq \; \sum_{C \in \partition \downarrow} \vect{B}_{D, C} \vect{x}_C \; \leq \; \sum_{C \in \partition \downarrow} \vect{B}_{D, C} \valueVector_C \; \eqqcolon \; \vect{v}'
\]
By definition of non-recurrent SCCs, we have $\rho(D) < 1$. In particular, this means that $(\vect{I} - \vect{B}_D)^{-1}$ exists and is nonnegative. This gives us:
\[
\vect{x}_D \leq (\vect{I} - \vect{B}_D)^{-1} \vect{x}' \qquad \valueVector_D = (\vect{I} - \vect{B}_D)^{-1} \vect{v}'
\]
Because $(\vect{I} - \vect{B}_D)^{-1}$ is nonnegative and $\vect{x}' \leq \vect{v}'$, we get $\vect{x}_D \leq \valueVector_D$. \qed
\end{case}
\end{proof}

\solToSubsystem*
\begin{proof}
Recall that $\stateSupp{\vect{x}} = \{\state \in \states \mid \exists \autoState \in \autoStates \centerdot (\state,\autoState) \in \supp{\vect{x}}\}$ and let $\states' \supseteq \stateSupp{\vect{x}}$.
Our aim is to show $\prob_{\dtmc'}(\prop) \geq \lambda$.
To this end, we first derive an equation system characterizing the value vector $\valueVector' \in [0,1]^{\states \times \autoStates}$ of $\dtmc'$, defined by $\valueVector'(\state,\autoState) = \prob_{\dtmc', \state}\bigl(\lang{\automaton, \autoState}\bigr)$.
Let $\vect{B}' = \vect{B}_{\dtmc' \times \automaton}$ be the matrix of the product $\dtmc' \times \automaton$, after removing non-accepting recurrent SCCs.
  \begin{claim}
  If $D \subseteq \states \times \autoStates$ is an accepting recurrent SCC of $\vect{B}'$, then it is also an accepting and recurrent SCC of $\vect{B}$.
  \end{claim}
  \begin{proof}
  	By construction, $\vect{B}'_D \leq \vect{B}_D$ and hence $\rho(\vect{B}'_D) \leq \rho(\vect{B}_D)$.
  	Hence, $\rho(\vect{B}'_D) = 1$ implies $\rho(\vect{B}_D) = 1$ and therefore if $\vect{B}'_D$ is recurrent, then so is $\vect{B}_D$.
    If $\vect{B}'_D$ is accepting, then so is $\vect{B}_D$, since this depends only on $D$.
  \end{proof}
This implies that for every accepting and recurrent SCC $D$ of $B'$, the vector $\boldMu_D$ is a $D$-normalizer.
  \begin{claim}
    $\vect{B}' \vect{x} \geq \vect{x}$
  \end{claim}
  \begin{proof}
    Our aim is to show for all $\state, \autoState \in \states \times \autoStates$:
    \[
    \sum_{\substack{\state',\autoState' \in \states \times \autoStates \\ \autoState' \in \autoTransition(\autoState,\state)}}\transMat'(s,s') \cdot \vect{x}(\state',\autoState') \geq \vect{x}(\state,\autoState)
    \]
By definition of $\transMat'$, we have $\transMat'(\state, \state') = \transMat(\state,\state')$ for all $\state' \in \states'$. Further, observe that $\vect{x}(\state, \autoState) = 0$ for all $\state \in \states\setminus\states'$ and $\autoState \in \autoStates$. Thus, we get:
\begin{align*}
 \sum_{\substack{\state',\autoState' \in \states \times \autoStates \\ \autoState' \in \autoTransition(\autoState,\state)}}\transMat'(\state,\state') \cdot \vect{x}(\state',\autoState') &= \sum_{\substack{\state',\autoState' \in \states' \times \autoStates \\ \autoState' \in \autoTransition(\autoState,\state)}}\transMat'(\state,\state') \cdot \vect{x}(\state',\autoState') \\
 &= \sum_{\substack{\state',\autoState' \in \states \times \autoStates \\ \autoState' \in \autoTransition(\autoState,\state)}} \transMat(\state,\state') \cdot \vect{x}(\state',\autoState') \ \geq \ \vect{x}(\state,\autoState)
\end{align*}
    The last inequality holds by the assumption $\vect{B} \vect{x} \geq \vect{x}$.
\end{proof}
  Having established these two claims and using the assumption $\boldMu_D^\top \vect{x}_D \leq 1$, we can apply \Cref{lemm:monotonicity} to $\vect{B}'$ to conclude $\vect{x} \leq \valueVector'$.
  But this already implies:
  \[ 
  \prob_{\dtmc'}(\prop) \ = \ \valueVector'(s_0,q_0) \ \geq \ \vect{x}(\init,\autoInit) \ \geq \ \lambda
  \]
  where the final inequality is again by assumption.
  \qed
  \end{proof}

\subsystemToSol*
\begin{proof}
  Let $\transMat'$ be the probabilistic transition matrix of $\dtmc'$.
  By definition, we have $\transMat' \leq \transMat$ and therefore for all $(\state, \autoState) \in \states \times \autoStates$:
  \[ \sum_{\substack{\state',\autoState' \in \states \times \autoStates \\ \autoState' \in \autoTransition(\autoState,\state)}}\transMat(\state,\state') \cdot \valueVector'(\state',\autoState') \geq  \sum_{\substack{\state',\autoState' \in \states \times \autoStates \\ \autoState' \in \autoTransition(\autoState,\state)}}\transMat'(\state,\state') \cdot \valueVector'(\state',\autoState') \; = \; \valueVector'(\state,\autoState). \]
  This shows $\vect{B} \valueVector' \geq \valueVector'$. The last equation follows from the defining equation of $\valueVector'$ \cite[Lemma 12]{baier_markov_2023-1}. Since the probability in a subsystem can only decrease, we have $\valueVector' \leq \valueVector$, which immediately shows that $\boldMu_D^\top \valueVector'_D \leq \boldMu_D^\top \valueVector_D  \leq 1$.
\qed
\end{proof}

\mcAndUba*
\begin{proof}
$\text{(i)} \Rightarrow \text{(ii)}$ follows from~\Cref{prop:subsystosol} and $\text{(ii)} \Rightarrow \text{(i)}$ from~\Cref{prop:soltosubsys}.
\end{proof}

\section{Supplementary Material for \Cref{section:experimental}}
\label{appendix:experiments}
Our implementation, experimental data and scripts for reproducing our results can be found in \cite{baier_2025_15680332}. Note that the latest version can be found at \url{https://doi.org/10.5281/zenodo.15680331}.

\medskip

\noindent\textbf{Comparison with \textsc{Prism}.} The runtime comparison for the considered models can be seen in \Cref{figure:verification-prism-switss}. Out of the $165$ considered problem instances, our tool was faster or equally as fast as \textsc{Prism} \cite{kwiatkowska_prism_2011-1} in $72$ instances. Conversely, \textsc{Prism} was faster in $58$ instances. Generally, we observe that the runtimes of the tool heavily depend on the query under consideration. We attribute the differences of the runtimes to the LTL translation tools. While we use \textsc{Spot} \cite{duret.22.cav}, \textsc{Prism} uses an internal implementation for LTL translation. Further, the fact that \textsc{Prism} is implemented in Java and our tool in C\texttt{++} might be also be contributing factor.

\begin{figure}[t]
\centering
\includegraphics[scale=0.3]{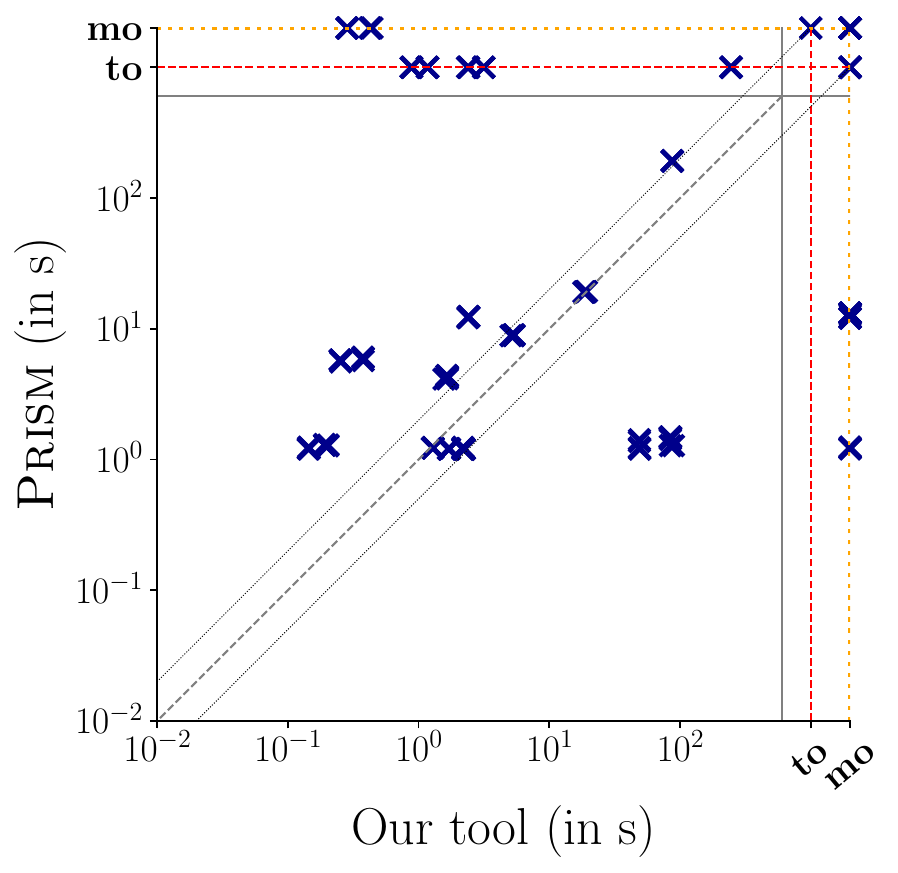}
\caption{Runtime comparison of \textsc{Prism} and our tool for the \emph{verification} of multi-objective $\omega$-regular queries.}
\label{figure:verification-prism-switss}
\end{figure}

\medskip

\noindent\textbf{LTL property in \textsf{RQ3}.} We have considered the property that ``a stable configuration is eventually reached and $k$ steps before that there are $7$ tokens in the ring''. Let $\texttt{sevenTokens}$ and $\texttt{stable}$ denote the corresponding atomic propositions. The LTL property is given by
\[
\Bigl(\neg\texttt{stable}\Bigr) \UntilOp \Bigl(\texttt{sevenTokens} \land \bigl((\neg \texttt{stable}) \UntilOp^{=k} \texttt{stable}\bigr)\Bigr)
\]
where $(\neg \texttt{stable}) \UntilOp^{=k} \texttt{stable}$ is syntactic sugar for $(\neg \texttt{stable}) \land \NextOp (\neg \texttt{stable}) \land \NextOp\NextOp (\neg \texttt{stable}) \land \ldots \land \NextOp^{k-1} (\neg \texttt{stable}) \land \NextOp^k (\texttt{stable})$. We note that the property is adapted from \cite{baier_markov_2023-1} for our experimental evaluation.

\end{document}